\documentclass[preprint,sort&compress]{elsarticle}

\usepackage{underscore}           

\usepackage{mathtools}
\usepackage{xspace}
\usepackage{amssymb}
\usepackage{amsmath}
\usepackage{amsthm}
\usepackage{thmtools}
\usepackage{bm}
\usepackage{hyperref}

\newcommand{\functionDot}{\,\cdot\,}

\newcommand{\setn}[1]{[#1]}
\newcommand{\oddsetn}[1]{\setn{#1}^{\mathit{odd}}}

\newcommand{\naturals}{\mathbb{N}}
\newcommand{\setnocond}[1]{\{#1\}}
\newcommand{\setcond}[2]{\{\, #1 \mid #2 \,\}}
\newcommand{\enumerationcond}[3][\stateOrder]{\setcond{#2}{#3}_{#1}}
\newcommand{\enumerationnocond}[2][\stateOrder]{\setnocond{#2}_{#1}}

\newcommand{\wordletter}[2]{{#1}[#2]}

\newcommand{\odd}[1]{\textit{odd}(#1)}
\newcommand{\even}[1]{\textit{even}(#1)}
\newcommand{\vertex}[2]{\langle {#1}, {#2} \rangle}
\newcommand{\edges}[1][E]{#1}
\newcommand{\vertices}[1][V]{#1}
\newcommand{\nvertices}[1][V]{\vertices[#1]_{N}}
\newcommand{\dvertices}[1][V]{\vertices[#1]_{D}}
\newcommand{\reducededges}[1][E]{#1^{e}}
\newcommand{\graph}[2]{\langle {#1}, {#2} \rangle}
\newcommand{\branch}{\gamma}

\newcommand{\alphabet}{\Sigma}

\newcommand{\infwords}{\alphabet^{\omega}}

\newcommand{\langsymb}[0]{\mathsf{L}}
\newcommand{\lang}[1]{\langsymb(#1)}
\newcommand{\inflang}[1]{\langsymb(#1)}

\newcommand{\states}{Q}
\newcommand{\nstates}{\states_{N}}
\newcommand{\dstates}{\states_{D}}
\newcommand{\trans}{\delta}
\newcommand{\cotrans}[1][e]{\delta^{e}}
\newcommand{\ntrans}{\delta_{N}}
\newcommand{\jtrans}{\delta_{J}}
\newcommand{\dtrans}{\delta_{D}}
\newcommand{\initialStates}{I}
\newcommand{\acc}{F}
\newcommand{\run}{\rho}
\newcommand{\initial}{\iota}

\newcommand{\mappingRunBranch}{M}

\newcommand{\transition}[3]{#1 \stackrel{#2}{\longrightarrow} #3}

\newcommand{\stateOrder}{\preccurlyeq}

\newcommand{\range}[1]{\setn{#1}}
\newcommand{\irange}[1]{\range{#1}_{\setminus 0}}

\newcommand{\dagAW}[2]{G_{#1, #2}}
\newcommand{\reduceddagAW}[3][e]{G^{#1}_{#2, #3}}
\newcommand{\idagAW}[3]{G^{#1}_{#2, #3}}
\newcommand{\ireduceddagAW}[4][e]{G^{#1, #2}_{#3, #4}}
\newcommand{\sdagAW}[2]{G^{s}_{#1, #2}}
\newcommand{\lcodagAW}[2]{G^{l}_{#1, #2}}

\newcommand{\aut}[1][A]{\mathcal{#1}}
\newcommand{\autNewInitial}[2]{{#1}[#2]}

\newcommand{\bigO}{\mathsf{O}}

\newcommand{\coveredBy}[2]{\leq^{#1}_{#2}}

\newcommand{\ranking}{\mathfrak{r}}
\newcommand{\levelRanking}{\mathfrak{l}}
\newcommand{\levelRankingFunctions}{\mathfrak{L}}
\newcommand{\priority}{\mathfrak{p}}

\newcommand{\statesWithValidValue}[1]{\alpha(#1)}

\newcommand{\buchi}{B\"uchi\xspace}
\newcommand{\rkc}{\textsf{RKC}\xspace}
\newcommand{\slc}{\textsf{SLC}\xspace}

\newcommand{\size}[1]{|#1|}

\newcommand{\lab}[1]{\textsf{#1}\xspace}
\newcommand{\labDie}{\lab{die}}
\newcommand{\labInf}{\lab{inf}}
\newcommand{\labNew}{\lab{new}}

\newcommand{\slcLabellingFunction}{\lambda^{\stableLevel}}

\newcommand{\infstates}[1]{\mathit{inf}({#1})}

\newcommand{\stableLevel}{\mathbf{s}}
\newcommand{\separatingLevel}{\mathbf{d}}
\newcommand{\runJumped}{\mathbf{j}}
\newcommand{\runFirstInfiniteAccepting}{\mathbf{f}}
\newcommand{\macrorunJump}{\bm \ell}
\newcommand{\bothJumped}{\mathbf{l}}
\newcommand{\runNotInS}{\mathbf{k}}
\newcommand{\runEntersB}{\mathbf{b}}

\newcommand{\runset}{\mathbf{R}}
\newcommand{\branchset}{\mathbf{B}}

\newtheorem{definition}{Definition}
\newtheorem{lemma}{Lemma}

\newtheorem{corollary}{Corollary}
\newtheorem{proposition}{Proposition}
\newdefinition{remark}{Remark}
\newtheorem{example}{Example}
\newenvironment{markedexample}{\begin{example}}{\qed\end{example}}
\newenvironment{markedremark}{\begin{remark}}{\qed\end{remark}}

\begin{document}

\begin{frontmatter}
	\title{On the Power of Finite Ambiguity in \buchi Complementation}

	\author[sklcs,ucas]{Weizhi Feng}
	\ead{fengwz@ios.ac.cn}
	
	\author[sklcs,liv]{Yong Li\corref{corAut}}
	\ead{liyong@ios.ac.cn}
	
	\author[sklcs,gziis]{Andrea Turrini\corref{corAut}}
	\ead{turrini@ios.ac.cn}
	
	\author[rice]{Moshe Y. Vardi}
	\ead{vardi@cs.rice.edu}
	
	\author[sklcs,ucas,gziis]{Lijun Zhang\corref{corAut}}
	\ead{zhanglj@ios.ac.cn}
	
	\cortext[corAut]{Corresponding author}

	\affiliation[sklcs]{
		organization={State Key Laboratory of Computer Science, Institute of Software, Chinese Academy of Sciences},
		addressline={South Fourth Street 4\#, Zhong Guan Cun},
		postcode={100190},
		city={Beijing},
		country={P.R. China}
	}
	
	\affiliation[ucas]{
		organization={University of the Chinese Academy of Sciences},
		addressline={No.19(A) Yuquan Road},
		postcode={100049},
		city={Beijing},
		country={P.R. China}
	}

  \affiliation[liv]{
		organization={Department of Computer Science, University of Liverpool},
		addressline={Ashton St.},
		postcode={L69 3BX},
		city={Liverpool},
		country={UK}
	}

	\affiliation[gziis]{
		organization={Institute of Intelligent Software, Guangzhou},
		addressline={Jinzhu Plaza},
		postcode={511458},
		city={Guangzhou},
		country={P.R. China}
	}

	\affiliation[rice]{
		organization={Department of Computer Science, Rice University},
		addressline={6100 Main St.},
		postcode={TX 77005-1827},
		city={Houston},
		country={USA}
	}

	\begin{abstract}
		In this work, we exploit the power of \emph{finite ambiguity} for the complementation problem of \buchi automata by using reduced run directed acyclic graphs (DAGs) over infinite words, in which each vertex has at most one predecessor;
		these reduced run DAGs have only a finite number of infinite runs, thus obtaining the finite ambiguity in \buchi complementation.
		We show how to use this type of reduced run DAGs as a \emph{unified tool} to optimize \emph{both} rank-based and slice-based complementation constructions for \buchi automata with a finite degree of ambiguity.
		As a result, given a \buchi automaton with $n$ states and a finite degree of ambiguity, the number of states in the complementary \buchi automaton constructed by the classical rank-based and slice-based complementation constructions can be improved 
		from $2^{\bigO(n \log n)}$ and $\bigO((3n)^{n})$ to $\bigO(6^{n}) \subseteq 2^{\bigO(n)}$ and $\bigO(4^{n})$, respectively.
		We further show how to construct such reduced run DAGs for limit deterministic \buchi automata and obtain a specialized complementation algorithm, thus demonstrating the generality of the power of finite ambiguity.
	\end{abstract}
	
 	\journal{arXiv}
\end{frontmatter}

\section{Introduction}
\label{sec:intro}

\sloppy
The complementation of nondeterministic \buchi automata (NBAs for short)~\cite{buchi90decision} is a classic problem and it is a fundamental construction used to solve many other important problems such as model checking~\cite{DBLP:conf/lics/VardiW86} and program-termination analysis~\cite{DBLP:conf/cav/HeizmannHP14}.

The model checking problem essentially asks whether the behavior of the system $A$ satisfies the given specification $B$.
In the automata-based model checking framework~\cite{DBLP:conf/lics/VardiW86},
the model checking problem reduces to a language-containment problem between $A$ and $B$, where the languages of $A$ and $B$ are the sets of all possible behaviors of $A$ and $B$, respectively.
On the one hand, when the system $A$ is modelled as an NBA while the specification $B$ is given as a formula $\phi$ in, e.g., the Linear-time Temporal Logic (LTL)~\cite{DBLP:conf/focs/Pnueli77}, the standard approach is simple: 
one first takes the negated formula $\lnot \phi$, then translates $\lnot \phi$ to an NBA $B_{\lnot \phi}$, and lastly checks the emptiness of the intersection of the languages of $A$ and $B_{\lnot \phi}$.
On the other hand, when both the system $A$ and the specification $B$ are given as NBAs, the complementation of NBAs becomes particularly important:
one first constructs a complementary automaton $B^{c}$ such that $\lang{B^{c}} = \infwords \setminus \lang{B}$ and then checks the emptiness of $\lang{A} \cap \lang{B^{c}}$, where $\alphabet$ is the set of actions allowed in both $A$ and $B$. 
Various implementations of this approach with optimizations~\cite{DBLP:conf/cav/AbdullaCCHHMV10, DBLP:conf/concur/AbdullaCCHHMV11, DBLP:journals/corr/abs-0902-3958, DBLP:journals/lmcs/ClementeM19} have been proposed to improve its practical performance. 
All these implementations, however, directly or indirectly resort to constructing $B^{c}$, which can be exponentially larger than $B$~\cite{Schewe09,Yan/08/lowerComplexity}.

The complementation of \buchi automata is also a key component in the automata-based program-termination checking framework proposed in~\cite{DBLP:conf/cav/HeizmannHP14}.
To prove that a program terminates on all inputs, one first constructs a termination proof for a sample path of the given program and then generalizes it to a \buchi automaton, whose language (by construction) represents a set of paths all sharing the same termination argument.
All these terminating paths are then removed from the program.
The removal of such paths is done by means of the NBA language difference operation, still based on \buchi complementation and intersection.
By iteratively removing terminating paths, one may obtain an empty program; 
in such a case, all paths of the program have a termination argument, thus the program is proved to be terminating.
It has been shown in~\cite{DBLP:conf/pldi/ChenHLLTTZ18} that efficient complementation algorithms for \buchi automata can significantly improve the performance of the program-termination checking framework.

Given its importance, the NBA complementation operation has been studied extensively in literature, where several kinds of algorithms have been proposed for the full class of NBAs as well as for specific subclasses.
The complexity for complementing an NBA $B$ has been proved to be in $\Omega((0.76n)^{n})$~\cite{Yan/08/lowerComplexity,Schewe09}.
A classic line of research on complementation aims at developing optimal (or close to optimal) complementation algorithms. 
Currently there are mainly four types of practical complementation algorithms for NBAs, namely \emph{Ramsey-based}~\cite{sistla1987complementation}, \emph{determinization-based}~\cite{safra1988complexity}, \emph{rank-based}~\cite{kupferman2001weak}, and \emph{slice-based}~\cite{kahler2008complementation} algorithms.
These algorithms, however, all unavoidably lead to a super-exponential growth of the size of $B^{c}$ in the worst case~\cite{Yan/08/lowerComplexity}. 
With the growing understanding of the worst-case complexity for these algorithms, searching for specialized complementation algorithms for certain subclasses of NBAs with better complexity has become an important line of research. 

On the one hand, researchers can easily identify subclasses of NBAs based on their transition structures; 
this led to the subclasses of deterministic, limit deterministic, and reverse deterministic \buchi automata. 
Deterministic \buchi automata (DBAs) have only one successor for each given state and letter; 
limit deterministic \buchi automata (LDBAs), also known as semi-deterministic \buchi automata, are NBAs that behave deterministically after visiting an accepting state; 
reverse deterministic \buchi automata (RDBAs) are backward deterministic, i.e., every state has at most one predecessor on each letter.
For these subclasses, specialized complementation algorithms with better complexity than $\Omega((0.76n)^{n})$ have been proposed in literature: 
it has been shown in~\cite{DBLP:journals/jcss/Kurshan87} that deterministic \buchi automata can be complemented in $\bigO(n)$, 
the more general subclass of limit deterministic \buchi automata has a complementation complexity of $\bigO(4^{n})$~\cite{Blahoudek16}, 
while reverse deterministic \buchi automata of $2^{\bigO(n)}$~\cite{DBLP:journals/corr/abs-1110-6183}.

On the other hand, some researchers proposed different subclasses of NBAs based on counting the number of accepting runs over a word.
For instance, we can identify the class of \emph{unambiguous} nondeterministic \buchi automata (UNBAs)~\cite{carton2003unambiguous} that are NBAs having at most one accepting run for each word;
we can also consider the more general class of \emph{finitely ambiguous} nondeterministic \buchi automata (FANBAs)~\cite{LodingP18} that are NBAs having a finite number of accepting runs instead of only one accepting run for each word. 
Unambiguous and finitely ambiguous \buchi automata have already been used in probabilistic verification~\cite{baier2016markov,CourcoubetisY95}, where nondeterminism can cause an imprecise computation of the probability of satisfying the given property~\cite{BustanRV04}:
despite being nondeterministic, UNBAs do not affect the correctness of the computed probability values;
moreover, they can be exponentially smaller than their equivalent deterministic counterpart~\cite{baier2016markov} usually used in probabilistic verification, such as deterministic Rabin or Parity automata (see, e.g.,~\cite{baier2008principles}).
Despite recognizing the full class of $\omega$-regular automata~\cite{LodingP18} as done by NBAs and LDBAs, research on complementation of FANBAs has seen less effort than the recent research on complementation of LDBAs~\cite{Blahoudek16,DBLP:conf/pldi/ChenHLLTTZ18,DBLP:conf/tacas/HavlenaLS22}. 
The complementation operation for FANBAs has been shown to be doable in $\bigO(5^{n})$~\cite{rabinovich18}, in contrast to $2^{\Omega(n \log n)}$ for general NBAs~\cite{Schewe09}.
Note that checking whether an NBA is an FANBA can be done in polynomial time~\cite{LodingP18}.
Therefore, once an FANBA has been identified, the specialized complementation construction for FANBAs can be applied. 

In this paper, we focus on an in-depth study of the complementation problem for FANBAs by extending our previous work~\cite{DBLP:journals/GandALF/LiVZ20} to cover also LDBAs, whose specialized construction cannot be applied to FANBAs.
Our main technical tool is the construction of reduced directed acyclic graphs (DAGs) of runs of FANBAs over infinite words called \emph{codeterministic run DAGs}, in which each vertex has \emph{at most one} predecessor.
This type of codeterministic run DAGs was previously introduced in~\cite{DBLP:journals/corr/abs-1110-6183,rabinovich18} under different names;
earlier works focus on providing concrete ways for constructing codeterministic run DAGs, while our work extracts the essence of those constructions --- obtaining a run DAG that is codeterministic, regardless of the ways for constructing it;
we defer the comparison of~\cite{DBLP:journals/corr/abs-1110-6183,rabinovich18} with our construction to the related work given in Section~\ref{sec:relatedWork}.
We show that such codeterministic run DAGs can be used to simplify and improve both the classical rank-based and slice-based complementation constructions. 
Our contributions are the following.
\begin{itemize}
\item 
    First, we apply codeterministic run DAGs of FANBAs over infinite words, as a \emph{unified tool} to show how finite ambiguity works in \buchi complementation, to optimize \emph{both} rank-based complementation (\rkc) and slice-based complementation (\slc).
\item
    Second, we show that the construction of codeterministic run DAGs in different complementation algorithms~\cite{TsaiFVT14} helps to achieve simpler and theoretically better complementation algorithms for FANBAs. 
    Given an FANBA with $n$ states, we show that the number of states of the complementary NBA constructed by the classical \rkc and \slc algorithms can be improved, respectively, to $\bigO(6^n) \subseteq 2^{\bigO(n)}$ from $2^{\bigO(n \log n)}$ and to $\bigO(4^{n})$ from $\bigO((3n)^{n})$.
\item
    Third, we reveal that \slc for general NBAs is basically an algorithm based on the construction of codeterministic run DAGs and a specialized complementation algorithm for FANBAs.
    We also provide in Proposition~\ref{prop:sumbsumptionRelationBetweenMacrostatesSLCforFANBAs} a subsumption relation between states in the complementary NBAs of FANBAs, which can be used to improve the containment checking between an NBA and an (FA)NBA and also to reduce the number of redundant states in the complementary NBA.
\item 
    Lastly, we apply codeterministic run DAGs to the complementation of LDBAs, in order to demonstrate their usefulness.
\end{itemize}

\paragraph*{Organization of the paper}
In the remainder of this paper, we first recap some definitions about \buchi automata in Section~\ref{sec:preliminaries} and then introduce the concept of codeterministic run DAGs in Section~\ref{sec:reduced-run-dag}. 
We present our improved algorithms for the rank-based and slice-based algorithms in Section~\ref{sec:rank-based} and Section~\ref{sec:slice-based}, respectively. 
Then, we apply codeterministic run DAGs to the complementation of limit deterministic \buchi automata in Section~\ref{sec:applicationToLDBAs}.
Lastly, we compare our work with the literature in Section~\ref{sec:relatedWork} and we conclude the paper with some future works in Section~\ref{sec:conclusion}.

We postpone the comparison with the literature to Section~\ref{sec:relatedWork} since this allows our readers to have a detailed knowledge of the constructions and results we develop in this work, making their comparison with existing work easier to understand.
To improve the readability of the paper, we provide the full formal proofs for the results presented in this paper in the appendix.

\section{Preliminaries}
\label{sec:preliminaries}

Let $n \in \naturals$ be a natural number;
we denote by $\range{n}$ the set of numbers $\setnocond{0, 1, \cdots, n} \subseteq \naturals$, by $\oddsetn{n}$ the set of odd numbers in $\range{n}$, and by $\irange{n}$ the set of numbers $\range{n} \setminus \setnocond{0}$.

Given a function $f \colon X \to Z$, we extend $f$ to sets in the usual way, that is, $f(X') = \cup_{x \in X'} f(x)$, where $X' \subseteq X$;
we adopt a similar notation for functions with higher arity, for instance for $\trans \colon \states \times \alphabet \to 2^{\states}$, we let $\trans(\states', a) = \bigcup_{q \in \states'} \trans(q, a)$.
We implicitly consider a function $f \colon X \to Z$ also as a function $f' \colon X \to 2^{Z}$ such that $f'(x) = \setnocond{f(x)}$ for each $x \in X$; 
symmetrically, if $\size{f'(x)} \leq 1$ for each $x \in X$, then we consider $f'$ also as a function $f \colon X \to Z$ such that $f(x) = z$ whenever $f'(x) = \setnocond{z}$.

Given a finite set $X$ and a binary relation $\stateOrder$ on $X$, we say that $\stateOrder$ is a \emph{total order} if $\stateOrder$ is reflexive (for each $x \in X$, $x \stateOrder x$), antisymmetric (for each $x,y \in X$, $x \stateOrder y$ and $y \stateOrder x$ implies $x = y$), transitive (for each $x,y,z \in X$, $x \stateOrder y$ and $y \stateOrder z$ implies $x \stateOrder z$), and strongly connected (for each $x,y \in X$, either $x \stateOrder y$ or $y \stateOrder x$).
We say that $\enumerationnocond{x_{1}, x_{2}, \cdots, x_{n}}$ is an \emph{enumeration} of $X$ under $\stateOrder$ if $n = \size{X}$, $\setnocond{x_{1}, x_{2}, \cdots, x_{n}} = X$, and for each $1 \leq i < j \leq n$, $x_{i} \stateOrder x_{j}$, that is, we order the elements of $X$ according to $\stateOrder$.
We always use indices starting from $1$ for enumerations.

We fix a finite \emph{alphabet} $\alphabet$ that we will use throughout the paper without further mentioning.
An \emph{$\omega$-word} is an infinite sequence $w$ of letters in $\alphabet$.
We denote by $\infwords$ the set of all $\omega$-words.
A \emph{language} is a subset of $\infwords$.
Let $L$ be a language; 
the complement language $L^{c}$ of $L$ is the language $L^{c} = \infwords \setminus L$.
Let $\sigma = r_{0} r_{1} r_{2} \cdots$ be a sequence of elements;
we denote by $\wordletter{\sigma}{i}$ the element $r_{i}$ of $\sigma$ at position $i$.

\begin{definition}[\buchi Automata (BAs)]
\label{def:BAs}
    A \emph{\buchi automaton} (BA) is a tuple $\aut = (\states, \initialStates, \trans, \acc)$, where
    $\states$ is a finite set of states, 
    $\initialStates \subseteq \states$ is a set of initial states, 
    $\trans \colon \states \times \alphabet \to 2^{\states}$ is a total transition function, and 
    $\acc \subseteq \states$ is a set of accepting states. 
\end{definition}
We may write $\transition{q}{a}{q'}$ for $q' \in \trans(q,a)$. 
We assume without loss of generality that each BA~$\aut$ is \emph{complete}, i.e., for each state $q \in \states$ and letter $a \in \alphabet$, we have that $\trans(q, a) \neq \emptyset$. 
If a BA is not complete, we make it complete by adding a fresh non-accepting state $q_{\bot} \notin \states$ and redirecting all missing transitions to $q_{\bot}$.

A \emph{run} of $\aut$ on a word $w$ is an infinite sequence of states $\run = q_{0} q_{1} \cdots$ such that $q_{0} \in \initialStates$ and for every $i \in \naturals$, $q_{i + 1} \in \trans(q_{i}, \wordletter{w}{i})$. 
We denote by $\infstates{\run}$ the set of states that occur infinitely often in the run $\run$. 
A word $w \in \infwords$ is \emph{accepted} by $\aut$ if there exists a run $\run$ of $\aut$ over $w$ such that $\infstates{\run} \cap \acc \neq \emptyset$.
We denote by $\lang{\aut}$ the \emph{language} recognized by $\aut$, i.e., the set of words accepted by $\aut$.

Given two runs $\run$ and $\run'$ of $\aut$ over $w$, we say that $\run$ and $\run'$ \emph{join} (or, \emph{merge}) if there is $n \in \naturals$ such that $\wordletter{\run}{i} = \wordletter{\run'}{i}$ for each $i \geq n$.

\begin{figure}
    \centering
    \includegraphics{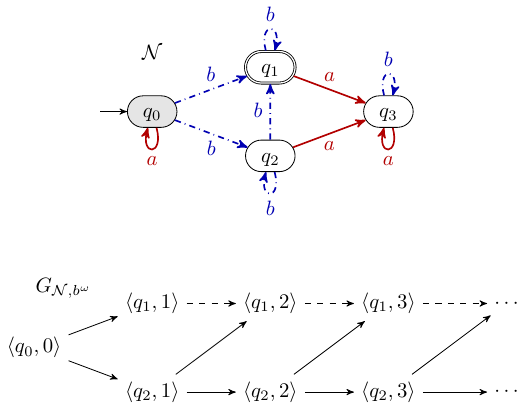}
    \caption{A BA~$\aut[N]$ with $\initialStates = \setnocond{q_{0}}$ and $\acc = \setnocond{q_{1}}$ and the run DAG~$\dagAW{\aut[N]}{w}$ over $w = b^{\omega}$.}
    \label{fig:exampleBA}
\end{figure}

\begin{markedexample}
\label{ex:BA}
    As an example of \buchi automaton, consider the automaton~$\aut[N]$ depicted on the top part of Figure~\ref{fig:exampleBA}.
    The alphabet is $\alphabet = \setnocond{a,b}$; 
    the set of states is $\states = \setnocond{q_{0}, q_{1}, q_{2}, q_{3}}$; 
    the initial states are $\initialStates = \setnocond{q_{0}}$, identified by the small incoming arrow and the light gray background color; 
    the transition function $\trans$ is shown by means of the arrows between states; and 
    the set of accepting states is $\acc = \setnocond{q_{1}}$, denoted by the double border.
    
    As we can see, the BA~$\aut[N]$ has language $\lang{\aut[N]} = \setcond{a^{i} b^{\omega} \in \infwords}{i \in \naturals}$;
    in fact, given a word $w = a^{i} b^{\omega}$, we have several runs over it:
    the sequence of $a$ letters at the beginning of $w$ is consumed by taking the transition $q_{0} \stackrel{a}{\longrightarrow} q_{0}$.
    The first $b$ is consumed either by the transition $q_{0} \stackrel{b}{\longrightarrow} q_{1}$ or by $q_{0} \stackrel{b}{\longrightarrow} q_{2}$, since we have $\trans(q_{0}, b) = \setnocond{q_{1}, q_{2}}$;
    if $q_{0} \stackrel{b}{\longrightarrow} q_{2}$ has been chosen, each letter $b$ of the remaining infinite sequence $b^{\omega}$ is consumed either by the transition $q_{2} \stackrel{b}{\longrightarrow} q_{2}$ or by $q_{2} \stackrel{b}{\longrightarrow} q_{1}$.
    Once $q_{1}$ is reached, on the next $b$ the BA~$\aut[N]$ can only take $q_{1} \stackrel{b}{\longrightarrow} q_{1}$ forever.
    
    Given a word $w = a^{i} b^{\omega}$, we have in practice two classes of runs:
    \begin{align*}
        \run_{r} & = \underbrace{q_{0} \cdots q_{0}}_{i+1} \underbrace{q_{2} q_{2} \cdots}_{\omega} \\
        \run_{j} & = \underbrace{q_{0} \cdots q_{0}}_{i+1} \underbrace{q_{2} \cdots q_{2}}_{j} \underbrace{q_{1} q_{1} \cdots}_{\omega}
    \end{align*}
    where $j \in \naturals$.
    The run $\run_{r}$ is not accepting, since $\infstates{\run_{r}} = \setnocond{q_{2}}$ and $\infstates{\run_{r}} \cap \acc = \emptyset$;
    for each $j \in \naturals$, the run $\run_{j}$ is accepting since $\infstates{\run_{j}} = \setnocond{q_{1}}$ and $\setnocond{q_{1}} \cap \acc \neq \emptyset$.
    Given two runs $\run_{j}$ and $\run_{k}$, we have that they merge at step $n = (i + 1) + \max \setnocond{j, k} + 1$ since we have $\wordletter{\run_{j}}{l} = q_{2} = \wordletter{\run_{k}}{l}$ for each $l \geq n$.
    $\run_{r}$, instead, do not merge with any other run.
    
    As we will see in Section~\ref{sec:reduced-run-dag}, all these runs on the same $\omega$-word $w$ can be represented as a run DAG; 
    one instance is the one depicted in the bottom part of Figure~\ref{fig:exampleBA}:
    for the word $w = b^{\omega}$, the corresponding run $\run_{r}$ is the bottom one; 
    the runs $\run_{j}$ follow the bottom run for $j$ steps, then take the transition from $\vertex{q_{2}}{j+1}$ to $\vertex{q_{1}}{j+2}$ and continue forever on the top run, by following the dashed arrows.
\end{markedexample}

Let $\aut$ be a BA; 
a complementary BA of $\aut$ is a BA that accepts the complementary language $\infwords \setminus \lang{\aut}$ of $\lang{\aut}$.
We denote by $\autNewInitial{\aut}{S}$ the automaton $(\states, S, \trans, \acc)$ obtained from $\aut$ by setting its set of initial states to be $S \subseteq \states$.
In particular, we use $\autNewInitial{\aut}{q}$ as the shorthand for $\autNewInitial{\aut}{\setnocond{q}}$.
We say a state $q$ of $\aut$ \emph{subsumes} a state $q'$ of $\aut$ if $\lang{\autNewInitial{\aut}{q'}} \subseteq \lang{\autNewInitial{\aut}{q}}$.
We classify $\aut$ into the following types of BAs according to its transition structure:
\begin{definition}[Types of BAs by transitions]
\label{def:typesOfBAsbyTransitions}
    Given a \buchi automaton~$\aut = (\states, \initialStates, \trans, \acc)$, we say that $\aut$ is
    \begin{itemize}
    \item 
        \emph{nondeterministic} (an NBA) if $\size{\initialStates} > 1$ or $\size{\trans(q, a)} > 1$ for some state $q \in \states$ and letter $a \in \alphabet$;
    \item
        \emph{deterministic} (a DBA) if $\size{\initialStates} = 1$ and for each $q \in \states$ and each $a \in \alphabet$, we have that $\size{\trans(q, a)} = 1$;
    \item
        \emph{reverse deterministic} (an RDBA) if for each state $q' \in \states$ and letter $a \in \alphabet$, there is at most one state $q \in \states$ such that $q' \in \trans(q, a)$; and
    \item 
        \emph{limit deterministic} (an LDBA) if the state set $\states$ can be partitioned into two disjoint sets $\nstates$ and $\dstates$ such that
        $\acc \subseteq \dstates$ and 
        for each state $q \in \dstates$ and $a \in \alphabet$, we have that $\size{\trans(q, a)} = 1$ and $\trans(q, a) \subseteq \dstates$. 
    \end{itemize}
\end{definition}
In particular, given an LDBA~$\aut$ whose states are partitioned into $\nstates$ and $\dstates$, these two sets induce a partition of $\trans$ into three disjoint transitions functions $\ntrans$, $\dtrans$, and $\jtrans$, representing the nondeterminism of $\aut$ inside $\nstates$, its deterministic nature in $\dstates$, and the jumps connecting $\nstates$ to $\dstates$, respectively. 
Formally, 
$\ntrans \colon \nstates \times \alphabet \to 2^{\nstates}$ is defined for each $n \in \nstates$ and $a \in \alphabet$ as $\ntrans(n, a) = \trans(n, a) \cap \nstates$;
$\jtrans \colon \nstates \times \alphabet \to 2^{\dstates}$ is defined for each $n \in \nstates$ and $a \in \alphabet$ as $\jtrans(n, a) = \trans(n, a) \cap \dstates$;
and 
$\dtrans \colon \dstates \times \alphabet \to \dstates$ is defined for each $d \in \dstates$ and $a \in \alphabet$ as $\dtrans(d, a) = d'$ where $\setnocond{d'} = \trans(d, a)$.

It is easy to see that each DBA is also an LDBA and that each LDBA is also an NBA.
An RDBA, instead, can be any of the other types.

From the perspective of the number of accepting runs of $\aut$, we have the following types of NBAs.
\begin{definition}[Types of NBAs by accepting runs]
\label{def:typesOfNBAsbyAcceptingRuns}
    Given an NBA~$\aut$, we say that $\aut$ is
    \begin{itemize}
    \item
        \emph{finitely ambiguous} (an FANBA) if for each $w \in \lang{\aut}$, the number of accepting runs of $\aut$ over $w$ is finite;
    \item
        \emph{$k$-ambiguous} if for each $w \in \lang{\aut}$, the number of accepting runs of $\aut$ over $w$ is at most $k$; and 
    \item
        \emph{unambiguous} if it is $1$-ambiguous.
    \end{itemize}
\end{definition}
By Definition~\ref{def:typesOfNBAsbyAcceptingRuns}, it holds that both $k$-ambiguous NBAs and unambiguous NBAs are special classes of FANBAs. 
We want to remark that an NBA can have as many non-accepting runs as we want on a given word $w$; 
this does not affect whether the automaton is an FANBA, $k$-ambiguous, or unambiguous.

\begin{figure}
    \centering
    \includegraphics{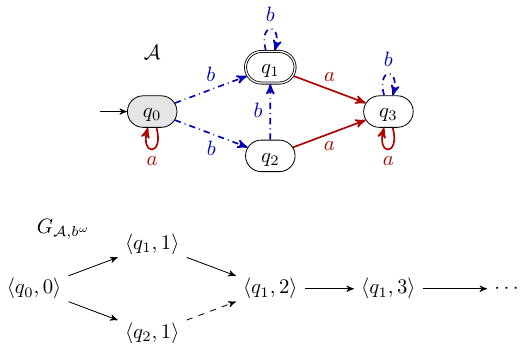}
    \caption{An FANBA~$\aut$ with $\initialStates = \setnocond{q_{0}}$ and $\acc = \setnocond{q_{1}}$ and the run DAG~$\dagAW{\aut}{w}$ over $w = b^{\omega}$.}
    \label{fig:example2ANBA}
\end{figure}
\begin{markedexample}
\label{ex:typesOfBAs}
    As an example of the different types of \buchi automata, consider again the automaton~$\aut[N]$ shown in Figure~\ref{fig:exampleBA}:
    this automaton is nondeterministic since e.g.\@ $\size{\trans(q_{0}, b)} = \size{\setnocond{q_{1}, q_{2}}} = 2 > 1$, thus it is not deterministic;
    also, it is not reverse deterministic since e.g.\@ $q_{1}$ has more than one predecessor via letter $b$, namely $q_{0}$, $q_{1}$, and $q_{2}$.
    It is, however, limit deterministic, since we can split $\states$ as $\nstates = \setnocond{q_{0}, q_{2}}$ and $\dstates = \setnocond{q_{1}, q_{3}}$; 
    it is immediate to see that $\acc = \setnocond{q_{1}} \subseteq \dstates$ and for each $q \in \dstates$ and $a \in \alphabet$, we have that $\size{\trans(q, a)} = 1$ and $\trans(q, a) \subseteq \dstates$.
    We can also partition $\trans$ into three parts: 
    $\ntrans \colon \nstates \times \alphabet \to 2^{\nstates}$, $\jtrans \colon \nstates \times \alphabet \to 2^{\dstates}$, and $\dtrans \colon \dstates \times \alphabet \to \dstates$ where $\ntrans = \setnocond{\transition{q_{0}}{a}{q_{0}}, \transition{q_{0}}{b}{q_{2}}, \transition{q_{2}}{b}{q_{2}}}$, $\jtrans = \setnocond{\transition{q_{0}}{b}{q_{1}}, \transition{q_{2}}{b}{q_{1}}, \transition{q_{2}}{a}{q_{3}}}$, and $\dtrans = \setnocond{\transition{q_{1}}{a}{q_{3}}, \transition{q_{1}}{b}{q_{1}}, \transition{q_{3}}{a}{q_{3}}, \transition{q_{3}}{b}{q_{3}}}$.
    Moreover, $\aut[N]$ is not finitely ambiguous (hence, not $k$-ambiguous or unambiguous) since, as we have seen in Example~\ref{ex:BA}, for e.g.\@ the word $w = b^{\omega}$ we have countably many accepting runs $\run_{j}$ over $w$, one for each $j \in \naturals$.
    Note that we can ignore the run $\run_{r}$ since it is not accepting.
    
    Consider now the BA~$\aut$ depicted in Figure~\ref{fig:example2ANBA}; 
    it is the BA~$\aut[N]$ without the transition $\transition{q_{2}}{b}{q_{2}}$ and it is easy to note that $\lang{\aut} = \lang{\aut[N]}$. 
    In this case we have that $\aut$ is still nondeterministic as well as not reverse deterministic; 
    it is also still limit deterministic, since $\nstates = \setnocond{q_{0}, q_{2}}$ and $\dstates = \setnocond{q_{1}, q_{3}}$ is a valid partition. 
    Note that also $\states'_{N} = \setnocond{q_{0}}$ and $\states'_{D} = \setnocond{q_{1}, q_{2}, q_{3}}$ is a valid partition, while $\states''_{N} = \setnocond{q_{0}, q_{1}, q_{2}}$ and $\states''_{D} = \setnocond{q_{3}}$ is not, since $\acc \not\subseteq \states''_{D}$.
    Differently from $\aut[N]$, we have that $\aut$ is a 2-ambiguous NBA, thus also an FANBA: 
    the language of $\aut$ is again $\lang{\aut} = \setcond{a^{i} b^{\omega} \in \infwords}{i \in \naturals}$ and for each $i \in \naturals$, there are only the two accepting runs $(q_{0})^{i+1} q_{1}^{\omega}$ and $(q_{0})^{i+1} q_{2} q_{1}^{\omega}$ over the word $a^{i} b^{\omega} \in \lang{\aut}$.
    These runs are depicted in the run DAG shown in the bottom part of Figure~\ref{fig:example2ANBA}, for the $\omega$-word $w = b^{\omega}$, i.e., $w = a^{i} b^{\omega}$ with $i = 0$.
\end{markedexample}

\begin{figure}
    \centering
    \includegraphics{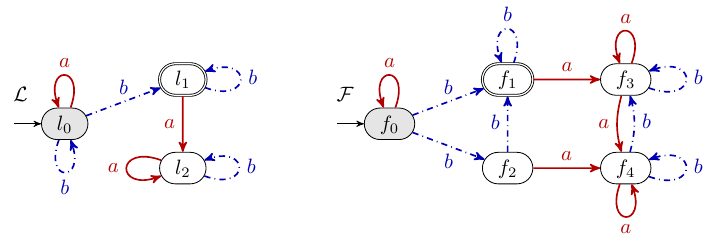}
    \caption{An LDBA~$\aut[L]$ and an FANBA~$\aut[F]$.}
    \label{fig:LDBAandFANBAareIncomparable}
\end{figure}
\begin{proposition}
\label{prop:LDBAandFANBAareIncomparable}
    Given the sets $\mathrm{LDBA}$ and $\mathrm{FANBA}$ of limit deterministic and finitely ambiguous NBAs, respectively, we have that $\mathrm{LDBA} \not \subseteq \mathrm{FANBA}$ and $\mathrm{FANBA} \not \subseteq \mathrm{LDBA}$.
\end{proposition}
To justify this proposition, we exhibit two BAs that are in one set but not in the other.
These automata are depicted in Figure~\ref{fig:LDBAandFANBAareIncomparable}.
Consider the BA~$\aut[L]$: 
it is trivial to verify that $\aut[L]$ is an LDBA, with partition $\nstates = \setnocond{l_{0}}$ and $\dstates = \setnocond{l_{1}, l_{2}}$.
On the other hand, $\aut[L]$ is not an FANBA:
the $\omega$-word $b^{\omega}$ accepted by $\aut[L]$ is such that for each $i \in \naturals$, the run $l_{0}^{i+1} l_{1}^{\omega}$ over $b^{\omega}$ is accepting, so $\aut[L]$ has infinitely many accepting runs over $b^{\omega}$.

Consider now the BA~$\aut[F]$: 
it is just the FANBA~$\aut$ from Figure~\ref{fig:example2ANBA} where we duplicated state $q_{3}$ as $f_{3}$ and $f_{4}$.
Exactly as in Example~\ref{ex:typesOfBAs}, we can show that $\aut[F]$ is an FANBA, given that each accepted word has only two accepting runs over it.
However, it is not an LDBA because the constraints $\acc = \setnocond{f_{1}} \subseteq \dstates$ and $\trans(f_{1}, a) \subseteq \dstates$ imply that $f_{3}$ must belong to $\dstates$, so we should have that $\size{\trans(f_{3}, a)} = 1$. 
Instead, we have that $\size{\trans(f_{3}, a)} = 2$, thus the conditions on $\aut[F]$ for being limit deterministic cannot be satisfied.

\begin{figure}
    \centering
    \includegraphics{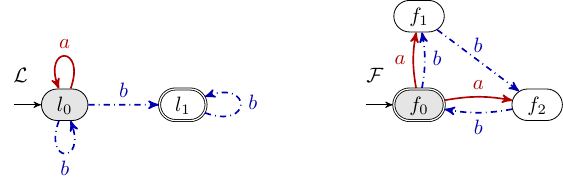}
    \caption{An LDBA and an FANBA that are not complete.}
    \label{fig:nonCompleteLDBAandFANBAareIncomparable}
\end{figure}
\begin{markedremark}
    In this work, we consider only complete automata; 
    the counterexamples we used in the proof of Proposition~\ref{prop:LDBAandFANBAareIncomparable} make use of this property, so one may wonder whether Proposition~\ref{prop:LDBAandFANBAareIncomparable} still holds once restricted to automata, not necessarily complete, in which we consider only the reachable fragment after pruning the states that have no possibility to reach an accepting state.
    
    Proposition~\ref{prop:LDBAandFANBAareIncomparable} indeed still holds:
    the fact that $\mathrm{LDBA} \not \subseteq \mathrm{FANBA}$ is still justified by the LDBA~$\aut[L]$, shown in Figure~\ref{fig:nonCompleteLDBAandFANBAareIncomparable}, that is obtained by pruning from the homonym LDBA depicted in Figure~\ref{fig:LDBAandFANBAareIncomparable} the state $l_{2}$ and relative transitions.
    Clearly it is still limit deterministic and nothing changes with regards to the accepting runs, given that we removed states that can never lead to accepting states in $\acc$.
    
    The fact that $\mathrm{FANBA} \not \subseteq \mathrm{LDBA}$ is witnessed by the NBA~$\aut[F]$ shown in Figure~\ref{fig:nonCompleteLDBAandFANBAareIncomparable}.
    This automaton has language $\lang{\aut[F]} = \setnocond{abb, ab, bbb}^{\omega}$ and it is easy to recognize that it is 2-ambiguous, thus $\aut[F] \in \mathrm{FANBA}$:
    the word $b^{\omega}$ has only the run $f_{0} f_{1} f_{2} f_{0} f_{1} f_{2} \cdots$ over it that is accepting;
    similarly, there is only one accepting run for the words in $\setnocond{abb, ab, bbb}^{*} \cdot \setnocond{abb, ab}^{\omega}$;
    there are instead two accepting runs on the word in $\setnocond{abb, ab, bbb}^{*} \cdot \setnocond{abb, ab} \cdot b^{\omega}$:
    for instance, the word $ab^{\omega}$ has $f_{0} f_{1} f_{2} f_{2} \cdots$ and $f_{0} f_{2} f_{2} \cdots$ as accepting runs over it.
    $\aut[F]$ is trivially not an LDBA since $f_{0} \in \acc$ has two $a$-successors, i.e., $f_{1}$ and $f_{2}$.
\end{markedremark}

\section{Codeterministic Run DAGs for FANBAs}
\label{sec:reduced-run-dag}

In this section, we first describe the concept of run DAG of an NBA over a word $w$, introduced in~\cite{kupferman2001weak}.
We then present the codeterministic run DAGs for FANBAs as a unified tool for both \rkc and \slc constructions by making use of the finite ambiguity in FANBAs.
In the remainder of the paper, we use DAGs as the shorthand for run DAGs.

\begin{definition}[Run DAG]
\label{def:runDAG}
    Let $\aut = (\states, \initialStates, \trans, \acc)$ be an NBA and $w \in \infwords$ be an $\omega$-word. 
    The DAG~$\dagAW{\aut}{w} = \graph{\vertices}{\edges}$ of $\aut$ over $w$ is defined as follows:
    \begin{description}
    \item[Vertices $\vertices$:] 
        the set of vertices $\vertices \subseteq \states \times \naturals$ is defined as $\bigcup_{l \geq 0} (\vertices_{l} \times \setnocond{l})$ where $\vertices_{0} = \initialStates$ and $\vertices_{l + 1} = \trans(\vertices_{l}, \wordletter{w}{l})$ for every $l \in \naturals$;
    \item[Edges $\edges$:] 
        there is an edge $(\vertex{q}{l}, \vertex{q'}{l'}) \in \edges$ from $\vertex{q}{l}$ to $\vertex{q'}{l'}$ if $l' = l + 1$ and $q' \in \trans(q, \wordletter{w}{l})$.
    \end{description}
    We call a DAG \emph{codeterministic} if each vertex only has at most one predecessor, that is, for each $v' \in \vertices$, we have that $\size{\setcond{v \in \vertices}{(v, v') \in \edges}} \leq 1$.
\end{definition}
Intuitively, we have a vertex $\vertex{q_{l}}{l} \in \vertices$ for each state $q_{l}$ that can be reached by $\aut$ after reading the first $l$ letters of $w$.
In other words, we have the vertex $\vertex{q_{l}}{l}$ if there exists a run $\run$ of $\aut$ over $w$ such that $q_{l} = \wordletter{\run}{l}$.
An edge $(\vertex{q_{l}}{l}, \vertex{q_{l+1}}{l+1})$ simply connects the vertex $\vertex{q_{l}}{l}$ to the vertex $\vertex{q_{l+1}}{l+1}$ provided that $\transition{q_{l}}{\wordletter{w}{l}}{q_{l+1}}$.

Let $\dagAW{\aut}{w} = \graph{\vertices}{\edges}$ be a DAG.
We say that $\dagAW{\aut}{w}$ is a sub-DAG of the DAG~$\dagAW{\aut}{w}' = \graph{\vertices'}{\edges'}$, denoted $\dagAW{\aut}{w} \subseteq \dagAW{\aut}{w}'$, if $\vertices \subseteq \vertices'$ and $\edges \subseteq \edges'$.

A vertex $\vertex{q}{l}$ is said to be on level $l$; 
by definition, there are at most $\size{\states}$ vertices on each level. 
A vertex $\vertex{q}{l}$ is an \emph{$\acc$-vertex} if $q \in \acc$.
A finite/infinite sequence of vertices $\branch = \vertex{q_{\ell}}{\ell} \vertex{q_{\ell + 1}}{\ell + 1} \cdots$ is called a \emph{branch} of $\dagAW{\aut}{w}$ if for each $l \geq \ell$, there is an edge from $\vertex{q_{l}}{l}$ to $\vertex{q_{l + 1}}{l+1}$.
We call $\branch$ an \emph{initial branch} if $\ell = 0$ and $q_{0} \in \initialStates$.
An \emph{$\omega$-branch} of $\dagAW{\aut}{w}$ is a branch of infinite length.
A \emph{fragment} $\vertex{q_{l}}{l} \vertex{q_{l+1}}{l+1} \cdots$ of $\branch$ is said to be a branch from the vertex $\vertex{q_{l}}{l}$; 
a fragment $\vertex{q_{l}}{l} \cdots \vertex{q_{l+k}}{l+k}$ of $\branch$ is said to be a \emph{path} from $\vertex{q_{l}}{l}$ to $\vertex{q_{l+k}}{l+k}$, where $k \in \naturals$. 
A vertex $\vertex{q_{j}}{j}$ is \emph{reachable} from $\vertex{q_{l}}{l}$ if there is a path from $\vertex{q_{l}}{l}$ to $\vertex{q_{j}}{j}$.
In particular, a vertex by default can reach itself.
We call a vertex $\vertex{q}{l}$ \emph{finite} if there are no $\omega$-branches in $\dagAW{\aut}{w}$ starting from $\vertex{q}{l}$; 
we call a vertex $\vertex{q}{l}$ \emph{$\acc$-free} if it is not finite and no $\acc$-vertices are reachable from $\vertex{q}{l}$ in $\dagAW{\aut}{w}$.
Given two $\omega$-branches $\branch = \vertex{q_{\ell}}{\ell} \vertex{q_{\ell + 1}}{\ell + 1} \cdots$ and $\branch' = \vertex{q'_{\ell'}}{\ell'} \vertex{q'_{\ell' + 1}}{\ell' + 1} \cdots$, we say that $\branch$ and $\branch'$ \emph{join} (or, \emph{merge}) if there is $n \geq \max\setnocond{\ell, \ell'}$ such that for $q_{j} = q'_{j}$ each $j \geq n$.

It is easy to recognize that it is possible to establish a bijection between the set of runs of $\aut$ on $w$ and the set of initial $\omega$-branches in $\dagAW{\aut}{w}$. 
In fact, to a run $\run = q_{0} q_{1} \cdots$ of $\aut$ over $w$ corresponds the initial $\omega$-branch $\hat{\run} = \vertex{q_{0}}{0} \vertex{q_{1}}{1} \cdots$ and, symmetrically, to an initial $\omega$-branch $\branch = \vertex{q_{0}}{0} \vertex{q_{1}}{1} \cdots$ corresponds the run $\hat{\branch} = q_{0} q_{1} \cdots$.
In both cases, the definitions of initial branch and run are respected since both require that for each $l \in \naturals$, $q_{l+1} \in \trans(q_{l}, \wordletter{w}{l})$.
As an immediate result, we have that $w$ is accepted by $\aut$ if and only if there exists an initial $\omega$-branch in $\dagAW{\aut}{w}$ that visits $\acc$-vertices infinitely often; 
we say that such an initial $\omega$-branch is \emph{accepting}. 
$\dagAW{\aut}{w}$ is called accepting if and only if there exists an accepting initial $\omega$-branch in $\dagAW{\aut}{w}$.
In the remainder of the paper we assume that all branches are initial, unless otherwise specified.

Assume that $\aut$ is an FANBA. 
Then an accepting $\omega$-branch in $\dagAW{\aut}{w}$, if it exists, only merges with other (accepting) $\omega$-branches for finitely many times.
That is, there exists a level $\separatingLevel \geq 1$ such that all vertices $\vertex{q_{l}}{l}$ with $l > \separatingLevel$ on an accepting $\omega$-branch have exactly one predecessor, i.e., there is only one vertex $\vertex{q_{l-1}}{l-1}$ such that $(\vertex{q_{l-1}}{l-1}, \vertex{q_{l}}{l}) \in E$;
we call the level $\separatingLevel$ a \emph{separating level}.
Intuitively, given an accepting $\omega$-branch $\branch$ and a separating level $\separatingLevel$, we can use $\separatingLevel$ to split $\branch$ into two parts: 
the finite initial part, of the levels below $\separatingLevel$, where other branches can merge with $\branch$; 
and the infinite part, of the levels above $\separatingLevel$, where $\branch$ visits $\acc$-vertices infinitely often and where $\branch$ is not affected by the behavior of other $\omega$-branches.
For example, consider the accepting DAG~$\dagAW{\aut}{w}$ of $\aut$ over $b^{\omega}$ in Figure~\ref{fig:example2ANBA}: 
the separating level is $2$, because each vertex $\vertex{q_{1}}{i}$ with $i > 2$ only has $\vertex{q_{1}}{i - 1}$ as its unique predecessor.
The separating level cannot be smaller than $2$ since the vertex $\vertex{q_{1}}{2}$ has two predecessors, namely $\vertex{q_{1}}{1}$ and $\vertex{q_{2}}{1}$.
We formalize the existence of separating levels of $\dagAW{\aut}{w}$ for an FANBA~$\aut$ in the following lemma.

\begin{restatable}[Separating Levels of Accepting DAGs of FANBAs]{lemma}{separatingLevelDAG}
\label{lem:separatingLevelDAG}
    Let $\aut$ be an FANBA and $\dagAW{\aut}{w}$ be the accepting DAG of $\aut$ over $w \in \lang{\aut}$.
    Then there must exist a separating level $\separatingLevel \geq 1$ in $\dagAW{\aut}{w}$.
\end{restatable}
The proof for this lemma is given in~\ref{app:otherProofs} and it relates the fact that $\aut$ has only finitely many accepting runs over $w$ with the fact that accepting $\omega$-branches eventually must stop to share vertices; 
the separating level is then the level of the last vertex that has been shared by any two accepting $\omega$-branches.

An immediate consequence of Lemma~\ref{lem:separatingLevelDAG} is that for each vertex $v$ in $\dagAW{\aut}{w}$ with more than one incoming edge, keeping only one of the incoming edges of $v$ does not change whether $\dagAW{\aut}{w}$ is accepting.
In fact, the definition of branch joining prevents us to merge an accepting and a non-accepting branch, since the former visits infinitely often $\acc$-vertices, while the latter at some point stops to visit an $\acc$-vertex.
Thus, it is not possible to find $n \in \naturals$ such that the two branches coincide after position $n$, hence at least one (non-)accepting branch must remain in $\dagAW{\aut}{w}$ after merging, if there were some before the merging operation.
This means that we can reduce the number of branches in a DAG without affecting its acceptance property.

Given an NBA~$\aut = (\states, \initialStates, \trans, \acc)$, let $\stateOrder$ be some total order over $\states$.
We construct the \emph{edge-reduced} DAG~$\reduceddagAW{\aut}{w} = \graph{\vertices}{\reducededges}$ of $\dagAW{\aut}{w}$, in which each vertex only has at most one predecessor, so that $\reduceddagAW{\aut}{w}$ is codeterministic, by means of the following edge removal policy:
if there is a vertex in $\dagAW{\aut}{w}$ with multiple incoming edges, we keep only the edge from the minimal predecessor according to the enumeration of the states induced by $\stateOrder$.
Formally, the definition of edges in $\reduceddagAW{\aut}{w}$ is given as follows. 
\begin{definition}[Reduced Run DAG]
\label{def:reducedRunDAG}
    Given an NBA~$\aut = (\states, \initialStates, \trans, \acc)$,
    a total order $\stateOrder$ over $\states$,
    and an $\omega$-word $w \in \infwords$, 
    let $\dagAW{\aut}{w} = \graph{\vertices}{\edges}$ be the run DAG of $\aut$ over $w$. 
    The \emph{reduced run DAG} $\reduceddagAW{\aut}{w} = \graph{\vertices}{\reducededges}$ of $\aut$ over $w$ has as set of edges
    \[
        \reducededges = \setcond{(\vertex{q}{l}, \vertex{q'}{l + 1}) \in \edges}{q = \min_{\stateOrder} \setcond{p \in \states}{q' \in \trans(p, \wordletter{w}{l}), \vertex{p}{l} \in \vertices}}\text{.} 
    \]
\end{definition}

The following lemma ensures that $\reduceddagAW{\aut}{w}$ is accepting only when $\aut$ accepts~$w$.
\begin{restatable}{lemma}{FANBAacceptanceCodeterministicDAG}
\label{lem:FANBAacceptanceCodeterministicDAG}
    Given an FANBA~$\aut$ and a word $w \in \infwords$, let $\reduceddagAW{\aut}{w}$ be the codeterministic DAG of $\aut$ over $w$.
    Then $w$ is accepted by $\aut$ if and only if $\reduceddagAW{\aut}{w}$ is accepting.
\end{restatable}
The proof for this lemma is given in~\ref{app:otherProofs} and makes use of $\dagAW{\aut}{w}$ to relate the visits to accepting states to the visits to accepting vertices.

\begin{markedexample}
\label{ex:2ANBAreducedDAG}
    As an example of reduced codeterministic DAG, consider the $\dagAW{\aut}{b^{w}}$ shown in Figure~\ref{fig:example2ANBA} and the natural enumeration $\enumerationnocond{q_{0}, q_{1}, q_{2}, q_{3}}$ for the states; 
    its reduced codeterministic counterpart is obtained by deleting the dashed edge from $\vertex{q_{2}}{1}$ to $\vertex{q_{1}}{2}$, since there is the edge from $\vertex{q_{1}}{1}$ to $\vertex{q_{1}}{2}$ and $q_{1}$ is smaller than $q_{2}$ according to the states enumeration. 
    It is easy to see that it is still accepting.
\end{markedexample}

\begin{markedremark}
\label{rem:NBAreducedDAGcounterexampleToLemmafanba-run-acc}
    The condition in Lemma~\ref{lem:FANBAacceptanceCodeterministicDAG} about $\aut$ being an FANBA is fundamental. 
    In fact, consider the NBA~$\aut[N]$, the word $b^{\omega}$, and the corresponding DAG~$\dagAW{\aut[N]}{b^{\omega}}$, depicted in Figure~\ref{fig:exampleBA}; 
    recall that $\aut[N]$ is not an FANBA.
    If the enumeration of $\states$ is $\enumerationnocond{q_{3}, q_{2}, q_{1}, q_{0}}$, then the edges to be removed are the ones shown in Figure~\ref{fig:exampleBA} as dashed arrows.
    Then every accepting vertex $\vertex{q_{1}}{l}$ has no successor and the only $\omega$-branch is the bottom one composed only by solid arrows, that has no accepting vertex at all.
\end{markedremark}

As mentioned at the beginning of this section, we will use DAGs and codeterministic DAGs as the building blocks for the complementation algorithms \rkc and \slc.
If we consider the set $S_{\ell}$ of the states $q$ occurring in the vertices $\vertex{q}{\ell}$ at level $\ell$ of a given DAG~$\dagAW{\aut}{w}$, we can see that $S_{\ell}$ can be obtained by standard subset construction.
If we consider the edge-reduced $\reduceddagAW{\aut}{w}$ counterpart of $\dagAW{\aut}{w}$, we can simplify $S_{\ell}$ to $S'_{\ell}$ by keeping only the minimal states.
This means that we can use the information provided by $\reduceddagAW{\aut}{w}$ to define how to construct $S'_{\ell + 1}$ from $S'_{\ell}$, that is, we can define a transition function for these sets.

\begin{definition}[Transition Function for Codeterministic DAGs]
\label{def:edge-relation-e}
    Given an FANBA~$\aut$ and $w \in \infwords$, let $S \subseteq \states$ be the set of states at level $\ell$ of $\reduceddagAW{\aut}{w}$; 
    let $S' = \trans(S, \wordletter{w}{\ell})$ be the set of states at level $\ell + 1$ of $\reduceddagAW{\aut}{w}$ and 
    $S_{\mathit{min}} = \setcond{q_{m} \in S}{q_{m} = \min_{\stateOrder} \setcond{p \in S}{q' \in \trans(p, \wordletter{w}{\ell})}, q' \in S'}$ 
    be the set of minimal predecessors of $S'$.
    Then, for a set of states $S_{1} \subseteq S$, we define $\cotrans_{w, \ell}(S_{1}, \wordletter{w}{\ell}) = \trans(S_{1} \cap S_{\mathit{min}}, \wordletter{w}{\ell})$. We call $\cotrans_{w, \ell}$ the \emph{reduced transition function at level~$\ell$} in $\reduceddagAW{\aut}{w}$. 
\end{definition}

\begin{figure}
    \centering
    \resizebox{\linewidth}{!}{
    \includegraphics{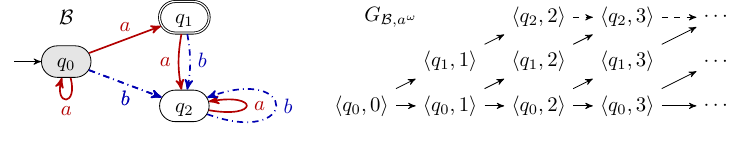}
    }
    \caption{Another FANBA~$\aut[B]$ with $\initialStates = \setnocond{q_{0}}$ and $\acc = \setnocond{q_{1}}$ and its run DAG~$\dagAW{\aut[B]}{a^{\omega}}$.}
    \label{fig:num-of-omega-branches}
\end{figure}
\begin{figure}
    \centering
    \includegraphics{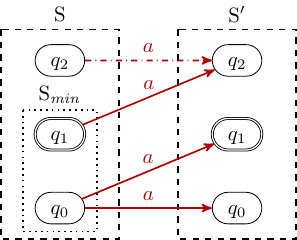}
    \caption{Reduced transition $S' = \cotrans(S,a)$ for the FANBA~$\aut[B]$ shown in Figure~\ref{fig:num-of-omega-branches}.}
    \label{fig:reduced-trans}
\end{figure}
\begin{markedexample}
\label{example:trans-not-on-level}
    Consider the FANBA~$\aut[B]$ and its run DAG~$\dagAW{\aut[B]}{a^{\omega}}$ on $a^{\omega}$, both shown in Figure~\ref{fig:num-of-omega-branches}.
    As order for the states, consider $q_{0} \stateOrder q_{1} \stateOrder q_{2}$.
    The corresponding reduced run DAG~$\reduceddagAW{\aut[B]}{a^{\omega}}$ is obtained by removing the dashed arrows.
    Regarding the level $\ell = 0$, we have $S = \setnocond{q_{0}}$ and we derive the sets $S' = \setnocond{q_{0}, q_{1}}$ and $S_{\mathit{min}} = \setnocond{q_{0}}$; 
    the latter follows from the fact that both $q_{0}$ and $q_{1}$ have $q_{0}$ as minimum (and only) $a$-predecessor.
    Thus, for each $S_{1} \subseteq \states$, we get $\cotrans_{w, 0} (S_{1}, a) = \emptyset$ when $q_{0} \notin S_{1}$ and $\cotrans_{w, 0} (S_{1}, a) = \setnocond{q_{0}, q_{1}}$ when $q_{0} \in S_{1}$.
    For level $\ell = 1$, we have $S = \setnocond{q_{0}, q_{1}}$ from which we get the sets $S' = \setnocond{q_{0}, q_{1}, q_{2}}$ and $S_{\mathit{min}} = \setnocond{q_{0}, q_{1}}$, where $q_{1}$ is now due to $q_{2} \in S'$.
    This means that $\cotrans_{w, 1} (S_{1}, a)$ contains both $q_{0}$ and $q_{1}$ if $q_{0} \in S_{1}$; 
    it also contains $q_{2}$ if $q_{1} \in S_{1}$.
    Lastly, for any level $\ell \geq 2$, we have $S = \states = \setnocond{q_{0}, q_{1}, q_{2}}$, $S' = \setnocond{q_{0}, q_{1}, q_{2}}$, and $S_{\mathit{min}} = \setnocond{q_{0}, q_{1}}$; 
    note that $q_{2} \notin S_{\mathit{min}}$ since we have $\setcond{p \in S}{q_{2} \in \trans(p, \wordletter{w}{\ell})} = \setnocond{q_{1}, q_{2}}$ and the $\stateOrder$-minimum is $q_{1}$.
    See Figure~\ref{fig:reduced-trans} for a graphical representation of the sets.
    This implies that $\cotrans_{w, \ell} = \cotrans_{w, 1}$.
\end{markedexample}

In general, the reduced transition function $\cotrans_{w, \ell}$ seems to depend on the level $\ell$ and the word $w$ yielding the edge connections between vertices at levels $\ell$ and $\ell + 1$ in $\reduceddagAW{\aut}{w}$. 
We claim, instead, that $\cotrans_{w, \ell}$ is not dependent on the level $\ell$ and the word $w$, due to our specific choice of the set $S_{\mathit{min}}$. 
Thus, we can just drop $\ell$ and $w$ from $\cotrans_{w, \ell}$, yielding $\cotrans$.

\begin{restatable}{lemma}{FANBAindependentReducedTrans}
\label{lem:FANBAindependentReducedTrans}
    Given an FANBA~$\aut$, $w_{1}, w_{2} \in \infwords$, $S \subseteq \states$, and $b \in \alphabet$, assume that there are two levels $\ell_{1}$ in $\reduceddagAW{\aut}{w_{1}}$ and $\ell_{2}$ in $\reduceddagAW{\aut}{w_{2}}$ such that the set of states and the input letter at level $\ell_{1}$ in $\reduceddagAW{\aut}{w_{1}}$ and at level $\ell_{2}$ in $\reduceddagAW{\aut}{w_{2}}$ are $S$ and $b$, respectively, for both levels.
    Then $\cotrans_{w_{1}, \ell_{1}}$ of $\reduceddagAW{\aut}{w_{1}}$ and $\cotrans_{w_{2}, \ell_{2}}$ of $\reduceddagAW{\aut}{w_{2}}$ are equal.
\end{restatable}
The proof for this lemma is given in~\ref{app:otherProofs}; 
it is a simple consequence of the assumptions about the set of states $S$ and the letter $b$.

\begin{corollary}
\label{cor:FANBAuniqueReducedTransitionFunctionDAG}
    Given an FANBA~$\aut$, there is a unique reduced transition function $\cotrans \colon 2^{\states} \times \alphabet \to 2^{\states}$ such that for each $w \in \infwords$, $\ell \in \naturals$, $S \in 2^{\states}$, and $a \in \alphabet$, if $\cotrans_{w, \ell}$ of $\reduceddagAW{\aut}{w}$ is defined on $(S, a)$, then $\cotrans(S, a) = \cotrans_{w, \ell}(S, a)$.
    We call $\cotrans$ the \emph{unique reduced transition function} associated with the reduced DAGs of $\aut$.
\end{corollary}

As one can expect, in general $\cotrans$ provides fewer transitions than $\trans$;
the missing transitions, however, are not necessary for the complementation construction, as long as the state order $\stateOrder$ is ``good''.
\begin{markedremark}[$\cotrans$ is different from $\trans$]
\label{ex:reduceTransition}
    Consider again the FANBA~$\aut[B]$ shown in Figure~\ref{fig:num-of-omega-branches} and $w = a^{\omega}$.
    As shown in Figure~\ref{fig:reduced-trans} and explained  below Definition~\ref{def:edge-relation-e}, the definition of $S_{\mathit{min}}$ makes us ignore the contribution of the transition $(q_{2},a,q_{2}) \in \trans$ to $\cotrans$.
    While this does not affect the sequence of sets of states reached while running over $w$, it does affect the behavior when we consider only some of the states. 
    We have seen that after reading the first two $a$ inputs, we reach $S = \setnocond{q_{0}, q_{1}, q_{2}}$; 
    by reading another $a$, we have $\trans(S, a) = S = \cotrans(S, a)$.
    However, if we focus on the singleton set $\setnocond{q_{2}}$, we get $\trans(\setnocond{q_{2}}, a) = \setnocond{q_{2}}$ while $\cotrans(\setnocond{q_{2}}, a) = \emptyset$.
    As we will see in the following sections presenting the complementation algorithms, removing unnecessary transitions can be useful when tracking rejected words and constructing the complementary automaton.
\end{markedremark}

Due to Corollary~\ref{cor:FANBAuniqueReducedTransitionFunctionDAG}, we can just use the reduced transition function $\cotrans$ with respect to a given set of states $S$ and an input letter $b$, independently from the actual level $\ell$ and word $w$ (of $\reduceddagAW{\aut}{w}$). 
This is particularly helpful in the construction of complementary NBAs of FANBAs (see Definitions~\ref{def:FANBArankBasedComplementation} and~\ref{def:FANBAsliceBasedComplementation}), since it allows us to ``ignore the past'', i.e., how we reached the current set of states $S \subseteq \states$. 
We remark that the codeterministic DAGs we construct in this paper from a run DAG are not the only possible codeterministic DAGs one can generate;
in the following example we provide another construction, so that $\cotrans_{w, \ell}$ indeed depends on the level $\ell$.

\begin{figure}
    \centering
    \resizebox{\linewidth}{!}{
    \includegraphics{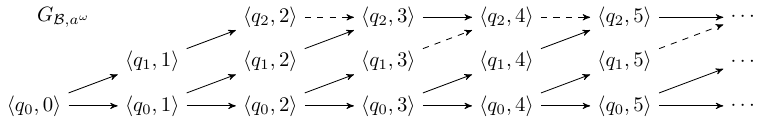}
    }
    \caption{A different reduced run DAG~$\dagAW{\aut[B]}{a^{\omega}}$ for the FANBA~$\aut[B]$ shown in Figure~\ref{fig:num-of-omega-branches}.} 
    \label{fig:cotransDependingOnLevel}
\end{figure}

\begin{markedexample}[$\cotrans_{w, \ell}$ depending on $\ell$]
\label{ex:cotransDependingOnLevel}
    Consider the DAG~$\dagAW{\aut[B]}{a^{\omega}}$ shown in Figure~\ref{fig:num-of-omega-branches} and take any level $\ell \geq 2$; 
    we have that $S^{\ell} = \setnocond{q_{0}, q_{1}, q_{2}}$.
    Moreover, we have that $\trans(S^{\ell}, a) = S^{\ell}$ as the set of successor states on each level $\ell \geq 2$.
    
    Rather than keeping the predecessor with the minimal index of a state in $S_{\mathit{min}}$ (see Definition~\ref{def:edge-relation-e}), we now define $S_{\mathit{min}}$ as $S^{\ell}_{\mathit{min}}$, which depends on the level $\ell$ as follows.
    When $\ell$ is an odd number, we define $S^{\ell}_{\mathit{min}} = \setnocond{q_{0}, q_{1}}$, otherwise we take $S^{\ell}_{\mathit{min}} = \setnocond{q_{0}, q_{2}}$.
    That is, we keep the predecessor $q_{1}$ of $q_{2}$ at odd levels and predecessor $q_{2}$ at even levels.
    Let $\cotrans_{a^{\omega}, \ell}$ be the resulting reduced transition function at level $\ell$. 
    For instance, we have that $\cotrans_{a^{\omega}, \ell}(\setnocond{q_{1}}, a) = \trans(\setnocond{q_{1}} \cap S^{\ell}_{\mathit{min}}, a) = \setnocond{q_{2}}$ when $\ell$ is odd and $\cotrans_{a^{\omega}, \ell}(\setnocond{q_{1}}, a) = \emptyset$ otherwise.
    Clearly, the definition of $\cotrans_{a^{\omega}, \ell}$ is dependent on the level $\ell$ and the resulting codeterministic DAG, shown in Figure~\ref{fig:cotransDependingOnLevel}, is different from the one depicted in Figure~\ref{fig:num-of-omega-branches}, once the dashed arrows have been removed. 
\end{markedexample}

In the remainder of this paper, we may just write $\cotrans(q, b)$ instead of $\cotrans(\setnocond{q}, b)$, for a singleton set $\setnocond{q}$.
The transition function $\cotrans$ will be used in the complementation of FANBAs since the complementation essentially constructs DAGs and then identifies accepting DAGs.

One can verify that each vertex in the codeterministic DAG~$\reduceddagAW{\aut}{w}$ of $\aut$ over $w$ has at most one predecessor. 
It follows that the number of $\omega$-branches in a non-accepting/accepting $\reduceddagAW{\aut}{w}$ is at most $\size{\states}$, as formalized in Lemma~\ref{lem:FANBAfiniteOmegaBranchesInCodeterministicDAG}.
\begin{restatable}[Finite Number of $\omega$-Branches in Codeterministic DAGs for FANBAs]{lemma}{FANBAfiniteOmegaBranchesInCodeterministicDAG}
\label{lem:FANBAfiniteOmegaBranchesInCodeterministicDAG}
    Given an FANBA~$\aut$, let $\reduceddagAW{\aut}{w}$ be a codeterministic DAG of $\aut$ over $w \in \infwords$.
    Then the number $m$ of $\omega$-branches in $\reduceddagAW{\aut}{w}$ is at most $\size{\states}$.
\end{restatable}
The proof for this lemma is given in~\ref{app:otherProofs}.
Intuitively, the fact that $\reduceddagAW{\aut}{w}$ is codeterministic prevents its $\omega$-branches to share a vertex (except for the ones at level $0$), so the Pigeonhole principle prevents such branches to exceed $\size{\states}$, the number of vertices at each level.

This result does not extend to finite branches (cf.\@ the next example);
however, this is not a problem because finite branches can only visit a finite number of accepting vertices, so they play no role in the acceptance of the word $w$. 

\begin{markedexample}
    Consider the DAG~$\dagAW{\aut[B]}{a^{\omega}}$ shown in Figure~\ref{fig:num-of-omega-branches}:
    it is easy to verify that there are infinitely many $\omega$-branches in the non-reduced DAG~$\dagAW{\aut[B]}{a^{\omega}}$ over $a^{\omega}$:
    we have an $\omega$-branch $\vertex{q_{0}}{0} \cdots \vertex{q_{0}}{i} \vertex{q_{1}}{i + 1} \vertex{q_{2}}{i + 2} \vertex{q_{2}}{i + 3} \cdots$ for each $i \in \naturals$.
    On the other hand, its codeterministic counterpart $\reduceddagAW{\aut[B]}{a^{\omega}}$, obtained by removing the dashed edges, has only one $\omega$-branch: 
    $\vertex{q_{0}}{0} \vertex{q_{0}}{1} \vertex{q_{0}}{2} \cdots$.
    
    Regarding the number of finite branches, $\reduceddagAW{\aut[B]}{a^{\omega}}$ still has infinitely many of them: 
    there is a finite branch $\vertex{q_{0}}{0} \cdots \vertex{q_{0}}{i} \vertex{q_{1}}{i + 1} \vertex{q_{2}}{i + 2}$ for each $i \in \naturals$, corresponding to the $\omega$-branch $\vertex{q_{0}}{0} \cdots \vertex{q_{0}}{i} \vertex{q_{1}}{i + 1} \vertex{q_{2}}{i + 2} \vertex{q_{2}}{i + 3} \cdots$ in $\dagAW{\aut[B]}{a^{\omega}}$ that has been truncated after level $i + 2$ by the removal of the edge $(\vertex{q_{2}}{i + 2}, \vertex{q_{2}}{i + 3})$.
\end{markedexample}

After redundant edges have been cut off, we obtain a DAG~$\reduceddagAW{\aut}{w}$ with a finite number of $\omega$-branches.
Thus if $w \notin \lang{\aut}$, there must exist a maximum level $l > 0$ among those $\omega$-branches such that each $\acc$-vertex $\vertex{q}{l'}$ with $l' \geq l$ is finite, which can be used for identifying whether $\reduceddagAW{\aut}{w}$ is accepting in the complementation of FANBAs.

\begin{definition}[Stable level in DAGs]
\label{def:stableLevel}
    Given an NBA~$\aut$ and an $\omega$-word $w \in \infwords$, we call a level $\stableLevel > 0$ \emph{stable} in $\dagAW{\aut}{w}$ if each $\acc$-vertex $\vertex{q}{i}$ with $i \geq \stableLevel$ in $\dagAW{\aut}{w}$ is finite.
\end{definition}

\begin{restatable}[Stable Level in Non-accepting Codeterministic DAGs for FANBAs]{lemma}{FANBAwordRejectedStableLevelInCodeterministicDAG}
\label{lem:FANBAwordRejectedStableLevelInCodeterministicDAG}
    Given an FANBA~$\aut$ and an $\omega$-word $w \in \infwords$, let $\reduceddagAW{\aut}{w}$ be the codeterministic DAG of $\aut$ over $w$.
    Then $w \notin \lang{\aut}$ if and only if there exists a stable level $\stableLevel > 0$ in $\reduceddagAW{\aut}{w}$.
\end{restatable}
The proof for this lemma is given in~\ref{app:otherProofs};
it relates the existence of the stable label of $\reduceddagAW{\aut}{w}$ with the fact that each $\omega$-branch of $\reduceddagAW{\aut}{w}$ at some point terminates to visit accepting vertices.

Consider again the DAG~$\dagAW{\aut[B]}{w}$ depicted in Figure~\ref{fig:num-of-omega-branches}:
there does not exist a stable level in this non-reduced DAG since each $\acc$-vertex $\vertex{q_{1}}{l}$ with $l \geq 1$ is not finite, given the fact that it occurs in the $\omega$-branch $\vertex{q_{1}}{l} \vertex{q_{2}}{l+1} \vertex{q_{2}}{l+2} \cdots$.
On the other hand, we can obtain its codeterministic counterpart $\reduceddagAW{\aut[B]}{w}$ by removing the dashed edges from $\vertex{q_{2}}{l}$ to $\vertex{q_{2}}{l+1}$.
It follows that the level $\stableLevel = 1$ is stable:
each $\acc$-vertex $\vertex{q_{1}}{l}$ with $l \geq 1$ occurs only in a single finite branch.

\section{Rank-Based Complementation}
\label{sec:rank-based}

We first recall in Subsection~\ref{ssec:rank-algo} the rank-based complementation (\rkc) algorithm proposed in~\cite{kupferman2001weak}, which constructs, for a given NBA~$\aut$, a complementary NBA~$\aut^{c}$ with $2^{\bigO(n \log n)}$ states.
Then in Subsection~\ref{ssec:FANBArankBasedComplementation} we show that, when restricted to FANBAs, \rkc based on the construction of codeterministic DAGs produces a complementary NBA~$\aut^{c}$ with $2^{\bigO(n)}$ states.

\subsection{Rank-Based Complementation Algorithm for NBAs}
\label{ssec:rank-algo}

\rkc was introduced by Kupferman and Vardi in~\cite{kupferman2001weak} to construct a complementary NBA~$\aut^{c}$ for a given NBA~$\aut$; 
this was done by identifying the DAGs of $\aut$ over non-accepted words $w \notin \lang{\aut}$. 
Intuitively, given a word $w \notin \lang{\aut}$,  all $\omega$-branches of the DAG of $\aut$ over $w$ will eventually stop visiting $\acc$-vertices, otherwise the word would be accepted. 
Based on this observation, in order to identify the non-accepting DAG of $\aut$ over $w$, Kupferman and Vardi introduced the notion of \emph{level rankings} for the DAG~$\dagAW{\aut}{w}$. 
By assigning only even ranks to $\acc$-vertices, they showed that there exists a unique ranking function that assigns ranks in $\range{2n}$ to the vertices of $\dagAW{\aut}{w}$ such that $w \notin \lang{\aut}$ if and only if all $\omega$-branches of $\dagAW{\aut}{w}$ eventually get trapped in odd ranks.

We now define the notion of level rankings for a run DAG. 
\begin{definition}[Level Ranking of a Run DAG]
\label{def:DAGlevelRanking}
    Given an NBA~$\aut$ with $n$ states and an $\omega$-word $w \in \infwords$, the level ranking of $\dagAW{\aut}{w} = \graph{\vertices}{\edges}$ is given by a ranking function $\levelRanking \colon \vertices \to \range{2n}$ that satisfies the following conditions:
    \begin{enumerate}[(i)]
    \item 
        for each vertex $\vertex{q}{i} \in \vertices$, if $\levelRanking(\vertex{q}{i}) \in \oddsetn{2n}$, then $q \notin \acc$,
    \item 
        for each edge $(\vertex{q}{i}, \vertex{q'}{i+1}) \in \edges$, $\levelRanking(\vertex{q'}{i+1}) \leq \levelRanking(\vertex{q}{i})$
    \end{enumerate}
\end{definition}
By Definition~\ref{def:DAGlevelRanking}, the ranks along a branch in the run DAG~$\dagAW{\aut}{w}$ decrease monotonically;
one can see that non-$\acc$-vertices get either even ranks or odd ranks while $\acc$-vertices get only even ranks.

Unfortunately, not all valid ranking functions defined by Definition~\ref{def:DAGlevelRanking} have the ability to decide whether $w$ is not accepted.
In fact, since non-$\acc$-vertices can also get even ranks, we simply cannot tell whether the infinite even ranks are due to the infinite number of $\acc$-vertices or just by $\levelRanking$.

Therefore, we want to define a specific ranking function $\ranking$ for $\dagAW{\aut}{w}$ that is able to tell whether a given run DAG~$\dagAW{\aut}{w}$ is accepting.
To this end, we will first introduce how to identify non-accepting DAGs by a finite number of pruning operations on them.

\begin{definition}[Sequence of DAGs for Identifying Rejected Words]
\label{def:sequenceDAGsRejectedWord}
    Given an NBA~$\aut$ and $w \in \infwords$, we define a sequence of DAGs $\idagAW{0}{\aut}{w} \supseteq \idagAW{1}{\aut}{w} \supseteq \cdots$, where $\idagAW{0}{\aut}{w} = \dagAW{\aut}{w}$, as follows.
    For each $i \geq 0$, 
    \begin{itemize}
    \item 
        $\idagAW{2i+1}{\aut}{w}$ is the DAG constructed from $\idagAW{2i}{\aut}{w}$ by removing all its finite vertices and the corresponding edges involving them, and
    \item 
        $\idagAW{2i+2}{\aut}{w}$ is the DAG constructed from $\idagAW{2i+1}{\aut}{w}$ by removing all $\acc$-free vertices in $\idagAW{2i+1}{\aut}{w}$ and the corresponding edges.
    \end{itemize}
\end{definition}
In practice, from $\dagAW{\aut}{w}$, we alternate between removing all vertices and edges that belong to a finite branch and removing all vertices and edges that belong to branches that will never reach an $\acc$-vertex.
This alternation is needed because by removing a vertex, we might make other vertices fulfill one of the two conditions.
The overall effect of these two repeated steps is that we prune $\dagAW{\aut}{w}$ by keeping only the vertices and edges occurring in some accepting $\omega$-branch.

\begin{markedexample}
    As an example of sequence of DAGs as given in Definition~\ref{def:sequenceDAGsRejectedWord}, consider the DAG~$\dagAW{\aut[B]}{a^{\omega}}$ shown in Figure~\ref{fig:num-of-omega-branches}.
    By definition, $\idagAW{0}{\aut[B]}{a^{\omega}} = \dagAW{\aut[B]}{a^{\omega}}$.
    
    For $i = 0$, we get the DAG~$\idagAW{1}{\aut[B]}{a^{\omega}}$ that is obtained from $\idagAW{0}{\aut[B]}{a^{\omega}}$ by removing all finite vertices (and corresponding edges).
    Since there are no finite vertices, we have that $\idagAW{1}{\aut[B]}{a^{\omega}} = \idagAW{0}{\aut[B]}{a^{\omega}} = \dagAW{\aut[B]}{a^{\omega}}$.
    We also get the DAG~$\idagAW{2}{\aut[B]}{a^{\omega}}$ as result of removing from $\idagAW{1}{\aut[B]}{a^{\omega}}$ all $\acc$-free vertices (and corresponding edges), that is, the vertices on $\omega$-branches that cannot reach an $\acc$-vertex.
    By definition, this means that we need to remove the vertices $\vertex{q_{2}}{j}$ for $j \geq 2$ and their incoming and outgoing edges.
    
    For $i = 1$, we get the DAG~$\idagAW{3}{\aut[B]}{a^{\omega}}$ that is obtained from $\idagAW{2}{\aut[B]}{a^{\omega}}$ by removing all finite vertices (and corresponding edges).
    This means that we have to remove all vertices $\vertex{q_{1}}{j}$ for $j \geq 1$ and their incoming and outgoing edges since each vertex $\vertex{q_{1}}{j}$ is not on an $\omega$-branch. 
    The resulting DAG~$\idagAW{3}{\aut[B]}{a^{\omega}}$ has only vertices $\vertex{q_{0}}{j}$ and the edges between them.
    
    We also get the DAG~$\idagAW{4}{\aut[B]}{a^{\omega}}$ as result of removing from $\idagAW{3}{\aut[B]}{a^{\omega}}$ all $\acc$-free vertices (and corresponding edges).
    Since there is no $\acc$-vertex in $\idagAW{3}{\aut[B]}{a^{\omega}}$, all vertices in $\idagAW{3}{\aut[B]}{a^{\omega}}$ are $\acc$-free, thus $\idagAW{4}{\aut[B]}{a^{\omega}}$ is empty.
    
    The following DAGs $\idagAW{k}{\aut[B]}{a^{\omega}}$ with $k > 4$ are trivially equal to $\idagAW{4}{\aut[B]}{a^{\omega}}$ since they are also all empty.
\end{markedexample}

Recall that $\acc$-free vertices cannot reach $\acc$-vertices.
We can see that each vertex $\vertex{q}{l}$ is either finite in $\idagAW{2i}{\aut}{w}$ or $\acc$-free in $\idagAW{2i+1}{\aut}{w}$ for all $i \geq 0$.
It was shown in~\cite{kupferman2001weak} that the pruning operations given in Definition~\ref{def:sequenceDAGsRejectedWord} will converge at $\idagAW{2n+1}{\aut}{w}$, as formalized below.

\begin{lemma}[\cite{kupferman2001weak}]
\label{lem:convergence-pruning}
    Given an NBA~$\aut$ with $n$ states and an $\omega$-word $w \in \infwords$, let $\idagAW{i}{\aut}{w}$, $i \in \naturals$, form the sequence of run DAGs defined in Definition~\ref{def:sequenceDAGsRejectedWord}.
    Then we have $\idagAW{k}{\aut}{w} = \idagAW{2n+1}{\aut}{w}$ for each $k > 2n$.
    In particular, $w$ is not accepted by $\aut$ if and only if $\idagAW{2n+1}{\aut}{w}$ is empty.
\end{lemma}

According to Lemma~\ref{lem:convergence-pruning}, we know that $\dagAW{\aut}{w}$ is accepting if $\idagAW{2n+1}{\aut}{w}$ is not empty.
We show below (cf.~\cite{kupferman2001weak}) how the sequence of DAGs generated from Definition~\ref{def:sequenceDAGsRejectedWord} actually yields a \emph{unique} ranking function $\ranking$ over the set of vertices in $\dagAW{\aut}{w}$.
We refer to this ranking function $\ranking$ as the \emph{classical} ranking function.
\begin{definition}[Classical Ranking Function for DAG~$\dagAW{\aut}{w}$]
\label{def:rankingFunctionForDAGrejectedWord}
    Given the DAG~$\dagAW{\aut}{w}$, let $\idagAW{0}{\aut}{w} \supseteq \idagAW{1}{\aut}{w} \supseteq \cdots \supseteq \idagAW{2n}{\aut}{w} \supseteq \idagAW{2n+1}{\aut}{w}$ be the finite sequence of DAGs constructed according to Definition~\ref{def:sequenceDAGsRejectedWord}.

    The ranking function $\ranking \colon \vertices \to \naturals$ for $\dagAW{\aut}{w}$ is defined as follows:
    for each $0 \leq i \leq 2 n$,
    \begin{enumerate}
    \item 
        if $i$ is even, then $\ranking(\vertex{q}{l}) = i$ for each finite vertex $\vertex{q}{l}$ in $\idagAW{i}{\aut}{w}$;
    \item 
        if $i$ is odd, then $\ranking(\vertex{q}{l}) = i$ for each $\acc$-free vertex $\vertex{q}{l}$ in $\idagAW{i}{\aut}{w}$.
    \end{enumerate}
    We also set $\ranking(\vertex{q}{l}) = 2n$ for every vertex $\vertex{q}{l}$ in $\idagAW{2n+1}{\aut}{w}$ provided that $\idagAW{2n+1}{\aut}{w}$ is not empty.
\end{definition}

As mentioned before, every vertex $\vertex{q}{l}$ in $\dagAW{\aut}{w}$ must satisfy either one of the following conditions (cf.~\cite{kupferman2001weak}):
\begin{itemize}
\item
    $\vertex{q}{l}$ is present in $\idagAW{2n+1}{\aut}{w}$, which indicates that $w$ is accepted by $\aut$ and the rank of $\vertex{q}{l}$ is $2n$;
\item
    $\vertex{q}{l}$ is a finite vertex in $\idagAW{i}{\aut}{w}$, with $i$ being even, which means that the rank of $\vertex{q}{l}$ is $i$; or
\item
    $\vertex{q}{l}$ is a $\acc$-free vertex in $\idagAW{i}{\aut}{w}$, with $i$ being odd, which indicates that the rank of $\vertex{q}{l}$ is $i$,
\end{itemize}
where we have $0 \leq i \leq 2n$.
By analyzing Definition~\ref{def:sequenceDAGsRejectedWord}, we get that from $\idagAW{0}{\aut}{w}$ to $\idagAW{1}{\aut}{w}$ we remove all finite vertices of $\idagAW{0}{\aut}{w}$, thus their ranking is $0$; 
all other vertices are not finite, so they don't get assigned a rank, for now.
$\idagAW{2}{\aut}{w}$ is obtained from $\idagAW{1}{\aut}{w}$ by removing all $\acc$-free vertices, that get $1$ as ranking;
all remaining vertices are not $\acc$-free vertices, so they do not get now their ranking.
By removing $\acc$-free vertices, it can happen that some of the non-finite vertices $v$ $\idagAW{1}{\aut}{w}$ become finite in $\idagAW{2}{\aut}{w}$: 
this is the case when $v$ is only on infinite branches that visit $\acc$-vertices finitely often.
By the removal of the $\acc$-free vertices, all these infinite branches are cut after the last occurrence of an $\acc$-vertex;
this means that $v$ has become finite, so it will be removed from $\idagAW{2}{\aut}{w}$ and get ranking $2$.
We can see that each vertex $v$ gets its ranking only once, since it fulfills the conditions in Definition~\ref{def:rankingFunctionForDAGrejectedWord} only once, just before being removed by the pruning operations given in Definition~\ref{def:sequenceDAGsRejectedWord}.
Thus $\ranking$ defined in Definition~\ref{def:rankingFunctionForDAGrejectedWord} is a valid ranking function (cf.\@ Definition~\ref{def:DAGlevelRanking}).
As shown in~\cite{kupferman2001weak}, we can use the classical ranking function $\ranking$ defined in Definition~\ref{def:rankingFunctionForDAGrejectedWord} to identify whether a given run DAG~$\dagAW{\aut}{w}$ is accepting, so to decide whether $w$ is accepted by $\aut$.

\begin{lemma}[Identification of Nonaccepting Run DAGs~{\cite[Lemma~5.2]{kupferman2001weak}}]
\label{lem:DAGrejectedWordOmegaBranchTrappedInOddRanks}
    Let $\dagAW{\aut}{w}$ be a run DAG.
    $\aut$ rejects the word $w$ if and only if the unique classical ranking function $\ranking$ given in Definition~\ref{def:rankingFunctionForDAGrejectedWord} for $\dagAW{\aut}{w}$ is such that all $\omega$-branches of $\dagAW{\aut}{w}$ eventually get trapped in odd ranks.
\end{lemma}

By means of Definition~\ref{def:rankingFunctionForDAGrejectedWord} and Lemma~\ref{lem:DAGrejectedWordOmegaBranchTrappedInOddRanks}, we have constructed a unique ranking function for identifying non-accepting DAGs.
However, this construction assumes that we are given the run DAG~$\dagAW{\aut}{w}$ beforehand.
In practice, it is hard to precompute the run DAG for every $\omega$-word $w$, not to mention constructing a ranking function to check whether $w$ is accepted by $\aut$.
According to~\cite{kupferman2001weak}, one can construct a run DAG level by level in an on-the-fly manner along an input $\omega$-word $w$.
In fact, in order to construct the next level, we just need to know the vertices at the current level and the letter $a \in \alphabet$ to read;
this allows us to construct run DAGs for every $\omega$-word.
Actually, it is enough to know the current set of states~$S$, since it does not matter how many steps (that is, the current level) we took to reach~$S$;
thus it is not necessary to keep track of the level number of the vertices when constructing the complementary NBA of $\aut$. 
Moreover, as we will see later in Definition~\ref{def:NBArankBasedComplementation} and Proposition~\ref{prop:correctness-rank-kv}, the number of such ``levels'' for all $\omega$-words, which correspond to states in the complementary NBA, is finite since the number of states in $\aut$ is finite.
Another important observation is that it is impossible for us to construct the classical ranking function: 
we do not know which run DAG~$\idagAW{i}{\aut}{w}$ in Definition~\ref{def:rankingFunctionForDAGrejectedWord} a vertex $v$ belongs to, thus we are not sure whether we must assign to $v$ the rank $2i$ or $2i+1$.

Therefore, in order to construct the complementary NBA~$\aut^{c}$ with the help of the classical ranking functions to identify non-accepting run DAGs, we have to guess the ranks of the vertices at every constructed level.
The maximum rank of a vertex is $2n$ according to Lemma~\ref{lem:convergence-pruning} and Definition~\ref{def:rankingFunctionForDAGrejectedWord}.
We note that Definition~\ref{def:rankingFunctionForDAGrejectedWord} just gives us one way to define a ranking function that works for identifying such non-accepting DAGs.
In fact, as long as we guess every possible ranking function for a run DAG, we are guaranteed to have one function that is the classical ranking function as by Definition~\ref{def:rankingFunctionForDAGrejectedWord}.
That is, if $w \notin \lang{\aut}$, the constructed NBA~$\aut^{c}$ is guaranteed to identify and accept $w$ because, by guessing ranking functions, we also guess correctly the classical ranking function given in Definition~\ref{def:rankingFunctionForDAGrejectedWord}.
On the other hand, as long as every ranking function we have guessed is a valid ranking function as in Definition~\ref{def:DAGlevelRanking}, $\aut^{c}$ will not accept an $\omega$-word $w \in \lang{\aut}$.
Thus, the constructed NBA~$\aut^{c}$ accepts exactly $\infwords \setminus \lang{\aut}$.

In the following, we will provide the details for the above complementation construction procedure based on guessing ranking functions.
As we have seen by Lemma~\ref{lem:convergence-pruning} and Definition~\ref{def:rankingFunctionForDAGrejectedWord}, the maximum rank of a classical ranking function is $2n$; this means that along an input word $w$ we can encode a ranking function for $\dagAW{\aut}{w}$ by ranking each state $q$ in the set of states $S$ we visit level by level in the DAG~$\dagAW{\aut}{w}$.
To do this, we define the level ranking functions as follows.
\begin{definition}
\label{def:levelRankingFunction}
    Given an NBA~$\aut$ with $n$ states, a \emph{level ranking function} for a given set of states $S \subseteq \states$ is a total function $\levelRanking_{S} \colon \states \to \range{2n} \cup \setnocond{\bot}$ such that 
    \begin{itemize}
    \item
        if $q \in S \cap \acc$, then $\levelRanking_{S}(q)$ is even; 
    \item
        if $q \in \states \setminus S$, then $\levelRanking_{S}(q) = \bot$.
    \end{itemize}
    We denote by $\levelRankingFunctions$ the set of all possible level ranking functions.
\end{definition}
Note that in Definition~\ref{def:levelRankingFunction} there is no explicit constraint on $\levelRanking_{S}(q)$ when $q \in S \setminus \acc$; 
however, by the way level ranking functions will be used in the construction of $\aut^{c}$ (cf.\@ Definitions~\ref{def:coverage-level-rnk}, \ref{def:NBArankBasedComplementation}, and~\ref{def:FANBArankBasedComplementation}) we ensure that $\levelRanking_{S}(q) \in \range{2n}$ for each $q \in S \setminus \acc$.
Moreover, we will use level ranking functions whose set $S$ is constructed by means of the transition function $\trans$, starting from the set~$\initialStates$ of initial states of $\aut$.
This means that a level ranking function $\levelRanking_{S}$ implicitly encodes $S$ by means of the values assigned to the states; 
for this motivation, we can drop the subscript $S$ from $\levelRanking_{S}$ since it is easy to derive it.

By the discussion above, the level ranking functions for the states visited at each level in the DAG~$\dagAW{\aut}{w}$ will constitute a ranking function $\levelRanking$ for $\dagAW{\aut}{w}$.
To make sure the resulting ranking function $\levelRanking$ agrees with Definition~\ref{def:DAGlevelRanking}, we need the ranks not increase in a branch, which is formalized below as a coverage relation between two level ranking functions.
Let $\statesWithValidValue{\levelRanking} = \setcond{q \in \states}{\levelRanking(q) \neq \bot}$.
\begin{definition}[Coverage Relation for Level Rankings]
\label{def:coverage-level-rnk}
    Given two level ranking functions $\levelRanking, \levelRanking' \colon \states \to \range{2n} \cup \setnocond{\bot}$, a transition function $\trans$, and a letter $a \in \alphabet$, we say that \emph{$\levelRanking$ covers $\levelRanking'$ under letter $a$ and transition function $\trans$}, written $\levelRanking' \coveredBy{\trans}{a} \levelRanking$, if for all $q' \in \states$, it holds that 
    \begin{itemize}
    \item
         if $q' \in \trans(\statesWithValidValue{\levelRanking}, a)$, then for each $q \in \statesWithValidValue{\levelRanking}$ such that $q' \in \trans(q, a)$, we have that $0 \leq \levelRanking'(q') \leq \levelRanking(q)$, and 
    \item
        if $q' \notin \trans(\statesWithValidValue{\levelRanking}, a)$, then $\levelRanking'(q') = \bot$.
    \end{itemize}
\end{definition}

Intuitively, if a state $q'$ can be reached from a state $q$ at the previous level (i.e., $q \in \statesWithValidValue{\levelRanking}$) via the transition function $\trans$ and letter $a$, then the rank of $q'$ has to be at most that of $q$, i.e., $\levelRanking'(q') \leq \levelRanking(q)$.
This condition in the coverage relation indicates that the level rankings $\levelRanking_{S}$ and $\levelRanking'_{S'}$ of two consecutive levels of $\dagAW{\aut}{w}$ with states $S$ and $S'$, respectively, do not increase in ranks.
Moreover, it ensures that none of the states $q'$ in $S'$ is assigned to $\bot$, since $0 \leq \levelRanking'(q')$.
If $q'$ cannot be reached from any state in $\statesWithValidValue{\levelRanking}$, then $q'$ will not be present in the next level, thus we impose $\levelRanking'(q') = \bot$.
Note here that $\coveredBy{\trans}{a}$ depends on the transition function $\trans$.

In the coverage relation given in Definition~\ref{def:coverage-level-rnk}, we can derive the sets $S'$ and $S$ for $\levelRanking' \coveredBy{\trans}{a} \levelRanking$, respectively, as $S = \statesWithValidValue{\levelRanking}$ and $S' = \trans(S, a) = \trans(\statesWithValidValue{\levelRanking}, a)$.

In order to verify that the guess about the ranking of $\dagAW{\aut}{w}$ is correct, \rkc uses a \emph{breakpoint construction} as proposed in~\cite{miyano1984alternating}. 
This construction employs a set of states $O \subseteq \states$ to check whether the vertices assigned with even ranks are finite.
Similarly to Lemma~\ref{lem:FANBAwordRejectedStableLevelInCodeterministicDAG}, the non-accepting DAG~$\dagAW{\aut}{w}$ with the classical ranking function given in Definition~\ref{def:rankingFunctionForDAGrejectedWord} eventually reaches a level where all $\omega$-branches get trapped in odd ranks, i.e., all vertices after the level with even ranks are finite.
Therefore, after reaching such a level, the breakpoint set $O$ can be used to keep track of the states that are assigned with even ranks, whose descendants eventually will be empty, i.e., $O$ will become empty.
When $O$ is empty, we also track other states with newly assigned even ranks by putting them in the set $O$.
If all states being checked have only finite descendants, $O$ will become empty for infinitely many times, which means that the guess is guaranteed to be correct.

The formal definition of the complementary NBA~$\aut^{c}$ of the input NBA~$\aut$ is given below.
To make the presentation simpler, in the remainder of the paper we call the states and runs of $\aut^{c}$ \emph{macrostates} and \emph{macroruns}, respectively.

\begin{definition}[\cite{kupferman2001weak}]
\label{def:NBArankBasedComplementation}
    Given an NBA~$\aut = (\states, \initialStates, \trans, \acc)$, we define the NBA~$\aut^{c} = (\states^{c}, \initialStates^{c}, \trans^{c}, \acc^{c})$ as the complementary NBA of $\aut$ as follows.
    \begin{itemize}
    \item 
        $\states^{c} \subseteq \levelRankingFunctions \times 2^{\states}$ is the smallest set such that for $\initialStates^{c}$ and $\trans^{c}$ as defined below, we have $\initialStates^{c} \subseteq \states^{c}$ and $\trans^{c}(\states^{c}, a) \subseteq \states^{c}$ for each $a \in \alphabet$;
    \item 
        $\initialStates^{c} = \setnocond{(\levelRanking, \emptyset)}$, where $\levelRanking(q) = 2n$ if $q \in \initialStates$ and $\levelRanking(q) = \bot$ otherwise;
    \item 
        for each $(\levelRanking, O) \in \states^{c}$ and $a \in \alphabet$, 
        \begin{enumerate}
        \item 
            if $O \neq \emptyset$, then $\trans^{c}((\levelRanking, O), a) = \setcond{(\levelRanking', \trans(O, a) \setminus \odd{\levelRanking'})}{\levelRanking' \coveredBy{\trans}{a} \levelRanking}$ (intuition: the breakpoint set $O$ only tracks vertices with even rank), 
        \item 
            if $O = \emptyset$, then $\trans^{c}((\levelRanking, \emptyset), a) = \setcond{(\levelRanking', \even{\levelRanking'})}{\levelRanking' \coveredBy{\trans}{a} \levelRanking}$ (intuition: $O = \emptyset$ means that all previous $\acc$-vertices with even ranks are finite, so we just need to verify the new vertices with even rank);
        \end{enumerate}
    \item 
        $\acc^{c} = \setnocond{(\levelRanking, \emptyset) \in \states^{c}}$,
    \end{itemize}
    where $\odd{\levelRanking} = \setcond{q \in \states}{\text{$\levelRanking(q)$ is odd}}$ and $\even{\levelRanking} = \setcond{q \in \states}{\text{$\levelRanking(q)$ is even}}$.
\end{definition}
For an example of the rank-based NBA complementation construction, see Example~\ref{ex:FANBArankBasedComplementation} in Section~\ref{ssec:FANBArankBasedComplementation}.
By inspecting Definition~\ref{def:NBArankBasedComplementation}, we can recognize that the ranking function for the initial macrostate is $\levelRanking_{\initialStates}$, for which we have $\statesWithValidValue{\levelRanking_{\initialStates}} = \initialStates$.
This, combined with Definition~\ref{def:coverage-level-rnk}, gives that the $a$-successors $(\levelRanking_{S}, O)$ of the initial macrostate have $S = \trans(\initialStates, a)$;
as before, we have that $\statesWithValidValue{\levelRanking_{S}} = S$.
This reasoning can be repeated for the successors $(\levelRanking'_{S'}, O')$ of $(\levelRanking_{S}, O)$ and for their successors, and so on and so forth.
In practice, we have that every macrostate $(\levelRanking, O)$ in $\aut^{c}$ corresponds to one or more levels of the DAG~$\dagAW{\aut}{w}$ for some $\omega$-word $w$, namely, the levels $\ell$ having $\statesWithValidValue{\levelRanking}$ as the set of all states occurring in  the vertices $\vertex{q}{\ell}$. 
So a macrorun of $\aut^{c}$ over $w$ can be related to the run DAG~$\dagAW{\aut}{w}$.
As mentioned before, if $w$ is accepted by $\aut^{c}$, i.e., $O$ becomes empty infinitely often, then we know that our guess about the ranking function of $\dagAW{\aut}{w}$ is correct.
We then conclude that all $\omega$-branches of $\dagAW{\aut}{w}$ eventually get trapped in odd ranks, since the vertices assigned with even ranks being checked in $O$ are all finite. 
It follows that no branch is accepting in $\dagAW{\aut}{w}$, i.e., $w \notin \lang{\aut}$.

When $w \in \lang{\aut}$, it is easy to see that some accepting $\omega$-branch will get infinitely many even ranks, which indicates that $O$ will be nonempty forever once an accepting branch enters $O$.
Therefore, $\aut^{c}$ will not accept an $\omega$-word $w \in \lang{\aut}$.
Thus we conclude that $\lang{\aut^{c}} = \infwords \setminus \lang{\aut}$. 
Since $\levelRanking \in \levelRankingFunctions$ is a function from $\states$ to $\range{2n} \cup \setnocond{\bot}$, the number of possible such functions is $(2n+2)^{n} \in \bigO((3n)^{n})$;
this implies that the number of macrostates of $\aut^{c}$ is in $2^{n} \times \bigO((3n)^{n}) \in \bigO((6n)^{n})$.

\begin{proposition}[The Language and Size of $\aut^{c}$~\cite{kupferman2001weak}]
\label{prop:correctness-rank-kv}
    Let $\aut$ be an NBA with $n$ states and $\aut^{c}$ be the NBA constructed according to Definition~\ref{def:NBArankBasedComplementation}. 
    Then $\inflang{\aut^{c}} = \infwords \setminus \lang{\aut}$ and $\aut^{c}$ has $\bigO((6n)^{n})$ macrostates.
\end{proposition}

\subsection{Rank-Based Complementation Algorithm for FANBAs}
\label{ssec:FANBArankBasedComplementation}

The \rkc construction given in Definition~\ref{def:NBArankBasedComplementation} applies to general NBAs;
if we restrict ourselves to the class of FANBAs, we can use a fixed set, not depending on the number of states of $\aut$, as the range of the level ranking functions appearing in the macrostates of $\aut^{c}$; 
in this way we can reduce considerably the size of $\aut^{c}$.
In the remainder of this section, we present how to obtain this result;
in particular, we show 
in Lemma~\ref{lem:maximumRankOfCodeterministicDAGs} that if $\aut$ is an FANBA, then the maximum rank of the vertices in a codeterministic DAG of $\aut$ is at most $2$, instead of $2n$ as for NBAs (cf.\@ Lemma~\ref{lem:DAGrejectedWordOmegaBranchTrappedInOddRanks}). 
This means that the range of a classical ranking function $\levelRanking \in \levelRankingFunctions$ defined by Definition~\ref{def:rankingFunctionForDAGrejectedWord} can be restricted to just $\setnocond{0, 1, 2} \cup \setnocond{\bot}$. 
Therefore, the number of macrostates in $\aut^{c}$ is in $2^{n} \times 4^{n} \in \bigO(8^{n})$; 
as we will see in Theorem~\ref{thm:FANBAlanguageSizeRankBasedComplement}, the actual upperbound on the number of macrostates can be reduced to $\bigO(6^{n})$, by a more careful analysis of the interaction between the possible value of the rank and the presence of the state in the breakpoint set.

\begin{restatable}[Maximum Rank of Codeterministic DAGs]{lemma}{maximumRankOfCodeterministicDAGs}
\label{lem:maximumRankOfCodeterministicDAGs}
    Given an FANBA~$\aut$ and $w \in \infwords$, let $\reduceddagAW{\aut}{w}$ be the codeterministic DAG of $\aut$ over $w$.
    Then $w \notin \lang{\aut}$ if and only if $\ireduceddagAW{3}{\aut}{w}$ is empty. 
\end{restatable}
The proof for this lemma is given in~\ref{app:otherProofs}.
It mainly links the existence of a stable level with what states are removed step by step from $\ireduceddagAW{0}{\aut}{w} = \reduceddagAW{\aut}{w}$.

Similarly to Definition~\ref{def:rankingFunctionForDAGrejectedWord}, we can construct a classical ranking function for $\reduceddagAW{\aut}{w}$.

\begin{definition}[Classical Ranking Function for DAG~$\reduceddagAW{\aut}{w}$]
\label{def:rankingFunctionForreducedDAGrejectedWord}
    Given the DAG~$\reduceddagAW{\aut}{w}$, let $\ireduceddagAW{0}{\aut}{w} \supseteq \ireduceddagAW{1}{\aut}{w} \supseteq \ireduceddagAW{2}{\aut}{w} \supseteq \ireduceddagAW{3}{\aut}{w}$ be the finite sequence of DAGs constructed according to Definition~\ref{def:sequenceDAGsRejectedWord}.
    
    The ranking function $\ranking \colon \vertices \to \setnocond{0, 1, 2}$ for $\reduceddagAW{\aut}{w}$ is defined as follows:
    \begin{enumerate}
    \item 
        $\ranking(\vertex{q}{l}) = 0$ for each finite vertex $\vertex{q}{l}$ in $\ireduceddagAW{0}{\aut}{w}$;
    \item 
        $\ranking(\vertex{q}{l}) = 1$ for each $\acc$-free vertex $\vertex{q}{l}$ in $\ireduceddagAW{1}{\aut}{w}$;
    \item
        $\ranking(\vertex{q}{l}) = 2$ for each finite vertex $\vertex{q}{l}$ in $\ireduceddagAW{2}{\aut}{w}$;
    \item
        $\ranking(\vertex{q}{l}) = 2$ for each vertex $\vertex{q}{l}$ in $\ireduceddagAW{3}{\aut}{w}$ if $\ireduceddagAW{3}{\aut}{w}$ is not empty.
    \end{enumerate}
\end{definition}
By Lemma~\ref{lem:maximumRankOfCodeterministicDAGs}, if there exists a vertex in $\ireduceddagAW{3}{\aut}{w}$, then $\reduceddagAW{\aut}{w}$ is accepting.
It follows that $\ireduceddagAW{3}{\aut}{w}$ has an $\omega$-branch whose vertices are all assigned with even rank $2$ according to Definition~\ref{def:rankingFunctionForreducedDAGrejectedWord}.
So the $\omega$-branch of $\reduceddagAW{\aut}{w}$ will not get trapped in odd ranks.
If $\ireduceddagAW{3}{\aut}{w}$ is empty, similarly to Definition~\ref{def:rankingFunctionForDAGrejectedWord}, a vertex of $\reduceddagAW{\aut}{w}$ can be a finite vertex either in $\ireduceddagAW{0}{\aut}{w}$ or in $\ireduceddagAW{2}{\aut}{w}$, or an $\acc$-free vertex in $\ireduceddagAW{1}{\aut}{w}$.
Moreover, since all $\acc$-free vertices get only rank $1$, all $\acc$-vertices are assigned with the even rank $2$.
Recall that a vertex will not be removed later than its predecessors when constructing the finite sequence of DAGs $\ireduceddagAW{i}{\aut}{w}$.
So the ranks in a branch do not increase and the ranking function $\ranking$ is a valid ranking function according to Definition~\ref{def:DAGlevelRanking}.
We show below that the classical ranking function defined in Definition~\ref{def:rankingFunctionForreducedDAGrejectedWord} can also be used to check whether a given $\omega$-word $w$ is accepted by $\aut$. 
\begin{restatable}[Identification of Nonaccepting Run DAGs for FANBAs]{lemma}{identificationOfNonacceptingRunDAGsForFANBAs}
\label{lem:identificationOfNonacceptingRunDAGsForFANBAs}
    Given an FANBA~$\aut$ and $w \in \infwords$, let $\reduceddagAW{\aut}{w}$ be the codeterministic DAG of $\aut$ over $w$.
    $\aut$ rejects the word $w$ if and only if the unique classical ranking function $\ranking$ given in Definition~\ref{def:rankingFunctionForreducedDAGrejectedWord} for $\reduceddagAW{\aut}{w}$ is such that all $\omega$-branches of $\reduceddagAW{\aut}{w}$ eventually get trapped in odd ranks.
\end{restatable}
The proof for this lemma is given in~\ref{app:otherProofs}; 
it is just a direct consequence of Definition~\ref{def:rankingFunctionForreducedDAGrejectedWord} and Lemma~\ref{lem:FANBAacceptanceCodeterministicDAG}.

In order to set the maximum rank to $2$ in Definition~\ref{def:NBArankBasedComplementation}, the underlying DAG~$\dagAW{\aut}{w}$ constructed for complementing FANBAs has to be codeterministic. 
We cannot use $\trans$, since it does not ensure that $\dagAW{\aut}{w}$ is codeterministic; 
this is instead guaranteed by the reduced transition function $\cotrans$, as we discussed in Section~\ref{sec:reduced-run-dag}.
Since \rkc generates rankings level by level, we then use the reduced transition function $\cotrans$ for computing successors at the next level. 
For FANBAs, the complementation construction  in Definition~\ref{def:NBArankBasedComplementation} can be improved accordingly:
\begin{definition}
\label{def:FANBArankBasedComplementation}
    Given an FANBA~$\aut = (\states, \initialStates, \trans, \acc)$, we define the NBA~$\aut^{c} = (\states^{c}, \initialStates^{c}, \trans^{c}, \acc^{c})$, where 
    \begin{itemize}
    \item
        $\states^{c} \subseteq \levelRankingFunctions \times 2^{\states}$ is the smallest set such that for $\initialStates^{c}$ and $\trans^{c}$ as defined below, we have $\initialStates^{c} \subseteq \states^{c}$, $\trans^{c}(\states^{c}, a) \subseteq \states^{c}$ for each $a \in \alphabet$, and for each $(\levelRanking, O) \in \states^{c}$, $\levelRanking$ has codomain $\setnocond{0, 1, 2, \bot}$;
    \item 
        $\initialStates^{c} = \setnocond{(\levelRanking, \emptyset)}$ where $\levelRanking(q) = 2$ if $q \in \initialStates$ and $\levelRanking(q) = \bot$ otherwise; 
    \item 
        for each $(\levelRanking, O) \in \states^{c}$ and $a \in \alphabet$,
        \begin{enumerate}
        \item 
            if $O \neq \emptyset$, then $\trans^{c}((\levelRanking, O), a) = \setcond{(\levelRanking', \cotrans(O, a) \setminus \odd{\levelRanking'})}{\levelRanking' \coveredBy{\cotrans}{a} \levelRanking}$, 
        \item 
            if $O = \emptyset$, then $\trans^{c}((\levelRanking, \emptyset), a) =  \setcond{(\levelRanking', \even{\levelRanking'})}{\levelRanking' \coveredBy{\cotrans}{a} \levelRanking}$,
        \end{enumerate}
        where $\cotrans$ is the unique reduced transition function associated with the codeterministic DAGs of $\aut$; 
        and
    \item
        $\acc^{c} = \setnocond{(\levelRanking, \emptyset) \in \states^{c}}$.
    \end{itemize}
\end{definition}

\begin{figure}[t]
    \centering
    \resizebox{\textwidth}{!}{
    \includegraphics{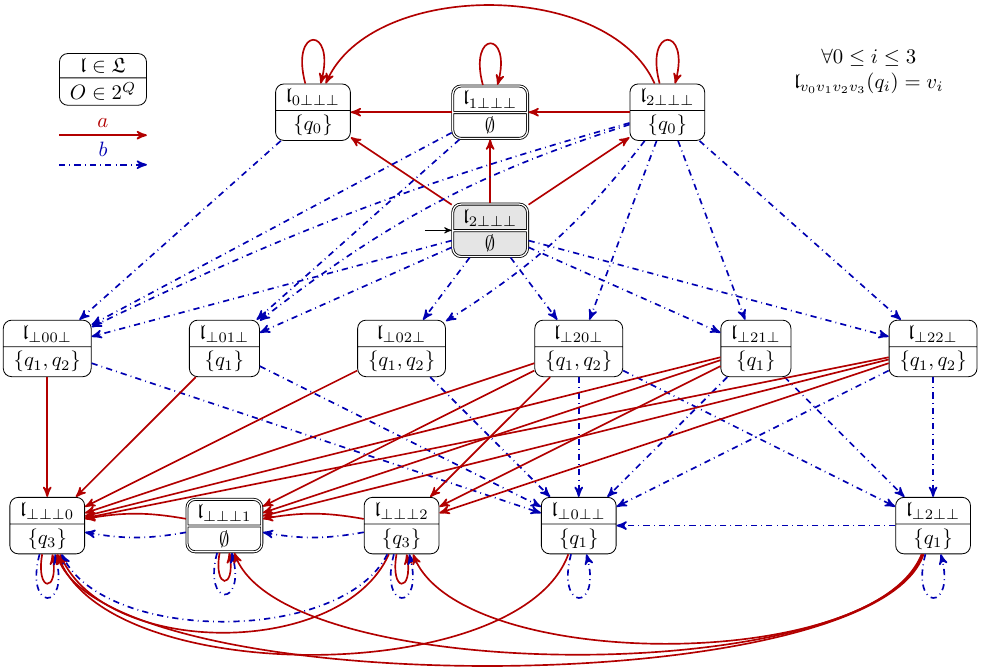}
    }
    \caption{The complementary NBA~$\aut^{c}$ of $\aut$ depicted in Figure~\ref{fig:example2ANBA} constructed by Definition~\ref{def:FANBArankBasedComplementation}.} 
    \label{fig:FANBArankBasedComplementation}
\end{figure}

\begin{markedexample}
\label{ex:FANBArankBasedComplementation}
    As an example of application of the rank-based complementation construction given in Definition~\ref{def:FANBArankBasedComplementation}, consider the FANBA~$\aut$ depicted in Figure~\ref{fig:example2ANBA}. 
    By definition, the initial macrostates are $\initialStates^{c} = \setnocond{(\levelRanking, \emptyset)}$, i.e., there is only one initial macrostate $(\levelRanking, \emptyset)$, where $\levelRanking(q_{0}) = 2$ and $\levelRanking(q_{1}) = \levelRanking(q_{2}) = \levelRanking(q_{3}) = \bot$.
    To make the notation more compact, we encode the assignments of the ranking made by the level ranking function in its subscript, so we just write $\levelRanking$ as $\levelRanking_{2 \bot \bot \bot}$; 
    the initial macrostate $(\levelRanking, \emptyset)$ thus becomes $(\levelRanking_{2 \bot \bot \bot}, \emptyset)$.
    
    From this initial macrostate, we have the following transitions:
    on reading the letter $a$, we have the successors
    $(\levelRanking_{0 \bot \bot \bot}, \setnocond{q_{0}})$,
    $(\levelRanking_{1 \bot \bot \bot}, \emptyset)$, and 
    $(\levelRanking_{2 \bot \bot \bot}, \setnocond{q_{0}})$.
    This is the case because the $a$-successor of $q_{0}$, the only state in $\statesWithValidValue{\levelRanking_{2 \bot \bot \bot}}$, is $q_{0}$ itself, so we need to have $\statesWithValidValue{\levelRanking'} = \setnocond{q_{0}}$.
    About its value, it can be $0$, $1$, and $2$ since each of them is between $0$ and $2$; 
    the set $O$ is constructed accordingly.
    
    In a similar way, on reading the letter $b$, we have the successors macrostates
    $(\levelRanking_{\bot 0 0 \bot}, \setnocond{q_{1}, q_{2}})$,
    $(\levelRanking_{\bot 0 1 \bot}, \setnocond{q_{1}})$,
    $(\levelRanking_{\bot 0 2 \bot}, \setnocond{q_{1}, q_{2}})$,
    $(\levelRanking_{\bot 2 0 \bot}, \setnocond{q_{1}, q_{2}})$,
    $(\levelRanking_{\bot 2 1 \bot}, \setnocond{q_{1}})$, and
    $(\levelRanking_{\bot 2 2 \bot}, \setnocond{q_{1}, q_{2}})$.
    Note that here $q_{1}$ only gets even ranks since it is an accepting state, so e.g.\@ the function $\levelRanking_{\bot 1 0 \bot}$ such that $\levelRanking_{\bot 1 0 \bot}(q_{1}) = 1$, $\levelRanking_{\bot 1 0 \bot}(q_{2}) = 0$, and $\levelRanking_{\bot 1 0 \bot}(q_{0}) = \levelRanking_{\bot 1 0 \bot}(q_{3}) = \bot$ is not considered since it is not a valid level ranking function, according to Definition~\ref{def:levelRankingFunction}.
    All other macrostates are constructed in a similar way; 
    the resulting NBA~$\aut^{c}$ is shown in Figure~\ref{fig:FANBArankBasedComplementation}.
    
    Note that for two macrostates $(\levelRanking_{1}, O_{1})$ and $(\levelRanking_{2}, O_{2})$ of $\aut^{c}$, we have that $(\levelRanking_{1}, O_{1})$ subsumes $(\levelRanking_{2}, O_{2})$, that is, $\lang{\aut^{c}[(\levelRanking_{2}, O_{2})]} \subseteq \lang{\aut^{c}[(\levelRanking_{1}, O_{1})]}$, if the following syntactic conditions hold (cf.~\cite{DBLP:journals/corr/abs-0902-3958}): 
    \begin{itemize}
    \item
        $\statesWithValidValue{\levelRanking_{1}} = \statesWithValidValue{\levelRanking_{2}} = S$ for some $S \subseteq \states$,
    \item
        $\levelRanking_{1}(q) \geq \levelRanking_{2}(q) $ for each $q \in S$, 
        and
    \item
        $O_{1} \subseteq O_{2}$.
    \end{itemize}
    In the NBA~$\aut^{c}$ depicted in Figure~\ref{fig:FANBArankBasedComplementation} we have several cases of subsumption; 
    for instance, we can see that the state $(\levelRanking_{\bot 0 1 \bot}, \setnocond{q_{1}})$ subsumes $(\levelRanking_{\bot 0 0 \bot}, \setnocond{q_{1}, q_{2}})$.
    
    Note that there are pairs of macrostates $q$ and $q'$ such that the language of $\autNewInitial{\aut^{c}}{q}$ is a superset of the language of $\autNewInitial{\aut^{c}}{q'}$, without fulfilling the syntactic conditions given above; 
    this means that such conditions are sufficient but not necessary to identify the subsumption between states.
    For instance, we have that $\lang{\autNewInitial{\aut^{c}}{(\levelRanking_{\bot 0 2 \bot}, \setnocond{q_{1}, q_{2}})}} \subseteq \lang{\autNewInitial{\aut^{c}}{(\levelRanking_{\bot 2 0 \bot}, \setnocond{q_{1}, q_{2}})}}$, but the second condition is violated, given that for $q_{2} \in \setnocond{q_{1}, q_{2}} = \statesWithValidValue{\levelRanking_{\bot 0 2 \bot}} = \statesWithValidValue{\levelRanking_{\bot 2 0 \bot}}$ we have that $\levelRanking_{\bot 2 0 \bot}(q_{2}) = 0 \not \geq 2 = \levelRanking_{\bot 0 2 \bot}(q_{2})$.
\end{markedexample}

Recall that the coverage relation between two level ranking functions $\levelRanking$ and $\levelRanking'$, parameterized with  $\cotrans$, is given in Definition~\ref{def:coverage-level-rnk}.
Similarly to Definition~\ref{def:edge-relation-e}, in order to compute $\cotrans(S_{1}, a)$, we have to first compute the set $S_{\mathit{min}}$ of minimal predecessors of $S' = \trans(S, a)$, where $S = \statesWithValidValue{\levelRanking}$, i.e., the set of states at the current level, and $a$ is the input letter at the current level. 
Thus we have $\cotrans(S_{1}, a) = \trans(S_{1} \cap S_{\mathit{min}}, a)$.
Intuitively, for $w \in \infwords$, $\cotrans$ is used to construct a codeterministic DAG~$\reduceddagAW{\aut}{w}$ over $w$ level by level.
By Lemma~\ref{lem:maximumRankOfCodeterministicDAGs}, the maximum rank of $\reduceddagAW{\aut}{w}$ is at most $2$, which is sufficient in Definition~\ref{def:FANBArankBasedComplementation} for constructing a classical ranking function to identify whether $\reduceddagAW{\aut}{w}$ eventually gets all its $\omega$-branches trapped in odd ranks, according to Lemma~\ref{lem:identificationOfNonacceptingRunDAGsForFANBAs}.
Therefore, with Definition~\ref{def:FANBArankBasedComplementation}, we can construct a complementary NBA~$\aut^{c}$ with $(2 + 2 + 1 + 1)^{n} \in {\bigO(6^{n})}$ states since a state is assigned with (1) even ranks ($0$ and $2$) and in the breakpoint $O$, (2) with even ranks and not in $O$, (3) with the odd rank $1$ or (4) the label $\bot$.
Note that in the last two cases, by definition we have that $q \notin O$.

\begin{restatable}[The Language and Size of $\aut^{c}$ for FANBAs]{theorem}{FANBAlanguageSizeRankBasedComplement}
\label{thm:FANBAlanguageSizeRankBasedComplement}
    Let $\aut$ be an FANBA with $n$ states and $\aut^{c}$ be the NBA constructed according to Definition~\ref{def:FANBArankBasedComplementation}. 
    Then $\lang{\aut^{c}} = \infwords \setminus \lang{\aut}$ and $\aut^{c}$ has $\bigO(6^{n})$ macrostates.
\end{restatable}
The detailed proof can be found in~\ref{app:FANBAlanguageSizeRankBasedComplement}.
Intuitively, to prove that $\lang{\aut^{c}} \subseteq \infwords \setminus \lang{\aut}$ we use the fact that an accepting run over $w \in \lang{\aut}$ visits accepting states infinitely often, thus the rank must stabilize to an even value, so eventually the component $O$ never becomes empty again.
For the other inclusion $\infwords \setminus \lang{\aut} \subseteq \lang{\aut^{c}}$, among all possible choices of the ranking function, one is the good one that makes all branches of $\reduceddagAW{\aut}{w}$ be trapped in odd ranks, hence the set $O$ becomes empty again and again, making the corresponding run of $\aut^{c}$ accepting.

In~\cite{DBLP:journals/corr/abs-1110-6183}, Fogarty and Vardi proved that complementing reverse deterministic NBAs with \rkc is doable in $2^{\bigO(n)}$ as the non-reduced DAGs $\dagAW{\aut}{w}$ are already codeterministic.
The motivation of this is that when $\aut$ is reverse deterministic, then each vertex $\vertex{q}{\ell + 1}$ in $\dagAW{\aut}{w}$ has at most one predecessor, a consequence of the fact that $q$ has only one $\wordletter{w}{\ell}$-predecessor;
it follows that $\dagAW{\aut}{w}$ is codeterministic. 
This implies that, similarly to Lemma~\ref{lem:FANBAfiniteOmegaBranchesInCodeterministicDAG}, the number of (accepting) $\omega$-branches in $\dagAW{\aut}{w}$ is at most $\size{\states}$.
It follows that there are at most $\size{\states}$ accepting runs of a reverse deterministic NBA~$\aut$ since there is an one-to-one correspondence between infinite runs of $\aut$ and $\omega$-branches in $\dagAW{\aut}{w}$.
Therefore, we have that reverse deterministic NBAs are a special class of FANBAs.
\begin{proposition}
\label{prop:reverseDeterministicNBAsAreFANBAs}
    Let $\aut$ be a reverse deterministic NBA.
    Then $\aut$ is also an FANBA.
\end{proposition}
In contrast, an FANBA is not necessarily a reverse deterministic NBA. 
For instance, the FANBA~$\aut$ depicted in Figure~\ref{fig:example2ANBA} is not reverse deterministic since $q_{1}$ has three $b$-predecessors, namely, $q_{0}$, $q_{1}$, and $q_{2}$. 

We remark that the construction presented in~\cite{DBLP:journals/corr/abs-1110-6183} just sets the maximum rank to $2$ in Definition~\ref{def:NBArankBasedComplementation} without modifying the transition function $\trans^{c}$, which turns out to be a special case of our construction, according to Proposition~\ref{prop:reverseDeterministicNBAsAreFANBAs}.

\section{Slice-Based Complementation}
\label{sec:slice-based}

In this section we consider a different algorithm for the complementation of NBAs.
It is based on slices, a different way to organize the information about runs on $\omega$-words.
In Subsection~\ref{ssec:slice-run-dag}, we first recall the \emph{slice-based} complementation construction (\slc) that was introduced in~\cite{DBLP:conf/birthday/VardiW08,kahler2008complementation,TsaiFVT14}, adapted to use our notation; 
given an NBA~$\aut$ with $n$ states, this construction produces a complementary NBA~$\aut^{c}$ with $\bigO((3n)^{n})$ states. 
Then, in Subsection~\ref{ssec:FANBAsliceBasedComplementation}, we show that this construction, when restricted to FANBAs, can be simplified to yield a complementary NBA with $\bigO(4^{n})$ states.

\subsection{Slice-Based Complementation Algorithm for NBAs}
\label{ssec:slice-run-dag}

Let $\aut$ be an NBA and $w \in \infwords$ be an $\omega$-word.
\slc uses a data structure called \emph{slice} instead of level ranking functions to encode the set of vertices in $\dagAW{\aut}{w}$ that are at the same level.
A slice is defined in~\cite{DBLP:conf/birthday/VardiW08} as an ordered sequence of disjoint sets of states with the same level. 
To simplify the notation, in this section we assume that the order is given by the index of the sets;
for instance, for vertices $\vertex{S_{i}}{\ell}$ and $\vertex{S_{j}}{\ell}$, we assume that $S_{i}$ precedes $S_{j}$ in the order whenever $i < j$.

Since the codeterministic DAGs constructed by \slc are different from those built by \rkc in Section~\ref{sec:reduced-run-dag}, we use the superscript $s$ to distinguish them, so we write $\sdagAW{\aut}{w}$ (instead of $\reduceddagAW{\aut}{w}$) for the codeterministic DAG generated by \slc as it proceeds along the word $w$.

We now describe \slc from the perspective of building codeterministic DAGs $\sdagAW{\aut}{w}$.
Each vertex is a pair $\vertex{S_{i}}{\ell}$ where $\emptyset \neq S_{i} \subseteq \states$ and $\ell \in \naturals$;
we say that $\vertex{S_{i}}{\ell}$ is an $\acc$-vertex if $S_{i} \subseteq \acc$.
We note here that a vertex in $\sdagAW{\aut}{w}$ is labelled with a set of states, in contrast to a single state in $\reduceddagAW{\aut}{w}$.

If $\initialStates \setminus \acc \neq \emptyset$, then we have at most two vertices at the level $0$ of $\sdagAW{\aut}{w}$: 
the vertex $\vertex{S_{1}}{0} = \vertex{\initialStates \setminus \acc}{0}$ and the $\acc$-vertex $\vertex{S_{2}}{0} = \vertex{\initialStates \cap \acc}{0}$;
if $\initialStates \cap \acc = \emptyset$, then we omit $\vertex{S_{2}}{0}$.
Otherwise, when $\initialStates \setminus \acc = \emptyset$, we have only the vertex $\vertex{S_{1}}{0} = \vertex{\initialStates \cap \acc}{0}$ at level $0$ of $\sdagAW{\aut}{w}$. 
In practice, we slice the initial states $\initialStates$ of $\aut$ as non-accepting and accepting states; 
clearly the resulting sets $S_{1}$ and $S_{2}$ are disjoint. 

Regarding the other levels, recall that the vertices $\vertex{S_{j}}{l}$ on some level $l$ in $\sdagAW{\aut}{w}$ are ordered from left to right by their indices. 
As already seen for level $0$, during the construction all empty sets $S_{j}$ are going to be removed and the indices of the remaining sets are shifted and compacted according to the increasing order of their original indices.
More precisely, assume that on level $l$, the sequence of vertices in $\sdagAW{\aut}{w}$ is $\vertex{S_{1}}{l}, \cdots, \vertex{S_{k_{l}}}{l}$. 
The vertices at level $l+1$ are generated as follows. 
Given a set $S_{j}$, on reading the letter $\wordletter{w}{l}$, the set $\trans(S_{j}, \wordletter{w}{l})$ of the $\wordletter{w}{l}$-successors of $S_{j}$ is computed;
it is then partitioned into the non-$\acc$ set $S'_{2j-1} = \trans(S_{j}, \wordletter{w}{l}) \setminus \acc$ and the $\acc$-set $S'_{2j} = \trans(S_{j}, \wordletter{w}{l}) \cap \acc$.

This gives us a sequence of sets $S'_{1}, S'_{2}, \cdots, S'_{2k_{l}-1}, S'_{2k_{l}}$.
We now have to ensure that the sets of states are not empty and disjoint.
Note that there can be states of $\aut$ that occur in several sets $S'_{j}$. 
The first operation we perform on the sequence of sets is to keep only the rightmost occurrence of a state: 
different runs of $\aut$ may merge with each other at some level and we only need to keep one of them and cut off the others, since they share the same infinite suffix and finite prefixes do not affect whether the DAG is accepting.
This means that, if a state $s$ occurs in multiple sets, we only keep $s$ in the rightmost set, that is, we choose to keep the run $\run$ represented in the rightmost set and cut off the other runs that join with $\run$ at $s$.
This operation does not change whether the codeterministic DAG~$\sdagAW{\aut}{w}$ is accepting, since at least one accepting run of $\aut$ remains and will not be cut off~\cite{DBLP:conf/birthday/VardiW08}. 
Formally, for each set $S'_{j}$, we define the set $S''_{j} = S'_{j} \setminus \bigcup_{j < p \leq 2k_{l}} S'_{p}$. 
This yields a sequence of disjoint sets $S''_{1}, S''_{2}, \cdots, S''_{2k_{l}-1}, S''_{2k_{l}}$; 
some of them may be empty.
After removing the empty sets in this sequence and reassigning the index of each set according to their positions, we finally obtain the sequence of sets on level $l+1$, denoted by $S_{1}, \cdots, S_{k_{l+1}}$, which will become the vertices $\vertex{S_{1}}{l+1}, \cdots, \vertex{S_{k_{l+1}}}{l+1}$ at level $l+1$.
Obviously, the resulting sets at the same level are again pairwise disjoint, so they form a slice.

Based on this construction, we define the codeterministic DAG~$\sdagAW{\aut}{w} = \graph{\vertices}{\edges}$ of $\aut$ over $w$ for an NBA~$\aut$ as follows:
\begin{description}
\item[Vertices $\vertices$:] 
    the set of vertices is $V = \bigcup_{\ell \in \naturals, 1 \leq j \leq k_{\ell}} \setnocond{\vertex{S_{j}}{\ell}}$;
\item[Edges $\edges$:]
    there is an edge from $\vertex{S_{j}}{l}$ to $\vertex{S_{h}}{l+1}$ if and only if $S_{h}$ is either $S''_{2j-1}$ or $S''_{2j}$ as defined above where $1 \leq j \leq k_{l}$ and $1 \leq h \leq k_{l+1}$. 
\end{description}

\begin{markedexample}
\label{ex:slicedDAGwordInLanguage}
    As an example of the construction of $\sdagAW{\aut}{w}$, consider the NBA~$\aut$ shown in Figure~\ref{fig:example2ANBA} and the word $w = b^{\omega}$.
    \begin{itemize}
    \item 
        There is only one vertex at level $0$, namely, $\vertex{S^{0}_{1}}{0} = \vertex{\setnocond{q_{0}}}{0}$, since $\initialStates \setminus \acc = \setnocond{q_{0}} \neq \emptyset$ and $\initialStates \cap \acc = \emptyset$.
    \item
        Regarding level $1$, from $\vertex{S^{0}_{1}}{0}$ we construct $S'_{1} = \trans(S^{0}_{1}, b) \setminus \acc = \setnocond{q_{2}}$ and $S'_{2} = \trans(S^{0}_{1}, b) \cap \acc = \setnocond{q_{1}}$.
        These sets are disjoint and not empty, so we just obtain the vertices $\vertex{S^{1}_{1}}{1} = \vertex{\setnocond{q_{2}}}{1}$ and $\vertex{S^{1}_{2}}{1} = \vertex{\setnocond{q_{1}}}{1}$, the only $\acc$-vertex at this level.
    \item
        Moving on to level $2$, we now construct four sets: 
        $S'_{1} = \trans(S^{1}_{1}, b) \setminus \acc = \emptyset$ and $S'_{2} = \trans(S^{1}_{1}, b) \cap \acc = \setnocond{q_{1}}$ out of $S^{1}_{1}$ as well as $S'_{3} = \trans(S^{1}_{2}, b) \setminus \acc = \emptyset$ and $S'_{4} = \trans(S^{1}_{2}, b) \cap \acc = \setnocond{q_{1}}$ out of $S^{1}_{2}$.
        After removing the common states, we obtain the four sets $S''_{1} = S''_{2} = S''_{3} = \emptyset$ and $S''_{4} = \setnocond{q_{1}}$.
        By purging empty sets and reassigning indices, we get the vertex $\vertex{S^{2}_{1}}{2} = \vertex{\setnocond{q_{1}}}{2}$, that is also an $\acc$-vertex.
    \item
        In level $3$, we construct the sets $S'_{1} = \trans(S^{2}_{1}, b) \setminus \acc = \emptyset$ and $S'_{2} = \trans(S^{2}_{1}, b) \cap \acc = \setnocond{q_{1}}$;  
        the only resulting vertex is $\vertex{S^{3}_{1}}{3} = \vertex{\setnocond{q_{1}}}{3}$, again an $\acc$-vertex.
    \end{itemize}
    The vertices in higher levels $\ell > 3$ are similar to $\vertex{S^{3}_{1}}{3} = \vertex{\setnocond{q_{1}}}{3}$: 
    they are just $\vertex{S^{\ell}_{1}}{\ell} = \vertex{\setnocond{q_{1}}}{\ell}$.
\end{markedexample}

\begin{markedexample}
\label{ex:slicedDAGwordNotInLanguage}
    As another example of the construction of $\sdagAW{\aut}{w'}$, consider again the NBA~$\aut$ shown in Figure~\ref{fig:example2ANBA} and the word $w' = bba^{\omega}$.
    The vertices at level $0$, $1$, and $2$ are the same as in Example~\ref{ex:slicedDAGwordInLanguage}, given the fact that the two words $w$ and $w'$ share the same first two letters.
    Consider now level $3$ and the letter $a$: from the vertex $\vertex{S^{2}_{1}}{2} = \vertex{\setnocond{q_{1}}}{2}$ we construct the sets $S'_{1} = \trans(S^{2}_{1}, a) \setminus \acc = \setnocond{q_{3}}$ and $S'_{2} = \trans(S^{2}_{1}, a) \cap \acc = \emptyset$;  
    the only resulting vertex is $\vertex{S^{3}_{1}}{3} = \vertex{\setnocond{q_{3}}}{3}$, a non-$\acc$-vertex.
    All vertices at level $\ell > 3$ are just similar, i.e., $\vertex{S^{\ell}_{1}}{\ell} = \vertex{\setnocond{q_{3}}}{\ell}$.
\end{markedexample}

By the construction of $\sdagAW{\aut}{w}$, each vertex $\vertex{S_{h}}{l+1}$, in which $S_{h}$ is either $S''_{2j-1}$ or $S''_{2j}$ as computed from $S_{j}$, has at most one predecessor $\vertex{S_{j}}{l}$.
Thus $\sdagAW{\aut}{w}$ is codeterministic.
Similarly to Lemma~\ref{lem:FANBAwordRejectedStableLevelInCodeterministicDAG}, we have the following lemma for the codeterministic DAG constructed for a given word.

\begin{lemma}[Codeterministic DAGs for NBAs~\cite{DBLP:conf/birthday/VardiW08}]
\label{lem:finite-ambiguity-stable-levels}
    Given an NBA~$\aut$ and a word $w \in \infwords$, let $\sdagAW{\aut}{w}$ be the codeterministic DAG as constructed above.
    Then we have that
    \begin{enumerate}
    \item 
        the number of (accepting) $\omega$-branches in $\sdagAW{\aut}{w}$ is at most $\size{\states}$;
    \item
        $w$ is accepted by $\aut$ if and only if $\sdagAW{\aut}{w}$ is accepting; and
    \item
        there exists a stable level $\stableLevel \geq 1$ in $\sdagAW{\aut}{w}$ such that all $\acc$-vertices after level $\stableLevel$ are finite if and only if $w \notin \lang{\aut}$.
    \end{enumerate}
\end{lemma}

Similarly to the \rkc method, given an NBA~$\aut$, the \slc algorithm~\cite{DBLP:conf/birthday/VardiW08} constructs a complementary NBA~$\aut^{c}$ whose runs over an $\omega$-word $w$  correspond to a codeterministic DAG~$\sdagAW{\aut}{w}$.
The crucial part of the \slc construction is also the identification of non-accepting codeterministic DAGs, similarly to the rank-based one.
The rank-based complementation construction uses the ranking function for the identification: 
the fact that the ranks of all $\omega$-branches get trapped in odd ranks indicates the non-acceptance of the DAGs.
The \slc construction, instead, exploits a different approach to identify non-accepting codeterministic DAGs, based on the stable level from Lemma~\ref{lem:finite-ambiguity-stable-levels}, Item (3).

Given an $\omega$-word $w$ and the codeterministic DAG~$\sdagAW{\aut}{w}$, the general idea in the \slc construction~\cite{DBLP:conf/birthday/VardiW08} to check whether $w \notin \lang{\aut}$ is that it will first guess a stable level $\stableLevel$ and then check whether all $\acc$-vertices after the level $\stableLevel$ are finite.
To identify the finite vertices after $\stableLevel$, we can label a vertex with $\bot$ if it is a descendant of an $\acc$-vertex and with $\top$ otherwise.
For the levels before $\stableLevel$, we can just label all vertices with $\top$.
It follows that $\sdagAW{\aut}{w}$ is non-accepting if and only if all $\omega$-branches in $\sdagAW{\aut}{w}$ will have only $\top$-labels.

Formally, this is obtained by defining the labelling function $\slcLabellingFunction \colon \vertices \to \setnocond{\top, \bot}$, where $\vertices$ is the set of vertices in $\sdagAW{\aut}{w}$, as follows:
for each $i \in \naturals$, 
let $\vertices_{i}$ be the set of vertices on level $i$.
\begin{itemize}
\item 
    If $i \leq \stableLevel$, then for every $v \in \vertices_{i}$ we define $\slcLabellingFunction(v) = \top$.
\item 
    If $i > \stableLevel$, then for every $v \in \vertices_{i}$:
    \begin{itemize}
    \item 
        if $v$ is an $\acc$-vertex, then $\slcLabellingFunction(v) = \bot$;
    \item 
        otherwise, $\slcLabellingFunction(v) = \slcLabellingFunction(u)$ with $(u, v) \in \edges$ where $\edges$ is the set of edges in $\sdagAW{\aut}{w}$.
        This means that the vertex $v$ gets its label from its only parent $u \in \vertices_{i-1}$. 
    \end{itemize}
\end{itemize}

It follows immediately from Lemma~\ref{lem:finite-ambiguity-stable-levels} that:
\begin{lemma}[\cite{DBLP:conf/birthday/VardiW08,DBLP:journals/corr/FogartyKWV13,DBLP:journals/iandc/FogartyKVW15}]
\label{lem:identify-nonacc-lambda-func}
    Let $\sdagAW{\aut}{w}$ be the codeterministic DAG of the NBA~$\aut$ over $w \in \infwords$.
    $\aut$ rejects $w$ if and only if there exists an integer $\stableLevel > 0$ such that all $\omega$-branches in $\sdagAW{\aut}{w}$ will have only $\top$-labels by the labelling function $\slcLabellingFunction$.
\end{lemma}

We note that every guess of the stable level $\stableLevel$ yields a run in the complementary NBA~$\aut^{c}$ constructed by \slc, so there may be infinitely many runs of $\aut^{c}$ over $w$.
The last part of the \slc construction is about how to verify whether $\slcLabellingFunction$ is correct given a guessed stable level $\stableLevel$, which is explained below.

Similarly to \rkc, also the \slc construction makes use of a breakpoint construction for the verification of $\slcLabellingFunction$. 
Assume that we are at a level $\ell \geq \stableLevel$ and the ordered vertices are $\vertex{S_{1}}{\ell}, \cdots, \vertex{S_{k_{\ell}}}{\ell}$.
A macrostate of $\aut^{c}$, corresponding to the level $\ell \geq \stableLevel$, not only contains the information about the ordered vertices $\vertex{S_{j}}{\ell}$, but also decorates each vertex $\vertex{S_{j}}{\ell}$ with a label $l_{j} \in \setnocond{\labDie, \labInf, \labNew}$.
These labels are used to verify whether $\slcLabellingFunction$ is correct.
Intuitively,
\begin{itemize}
\item  
    \labDie-labelled vertices $\vertex{S_{j}}{\ell}$ are  $\bot$-vertices, descendants of an $\acc$-vertex.
    The set of \labDie-labelled vertices is currently being inspected in the construction and in fact forms a breakpoint set similarly to the rank-based complementation construction.
    For $w$ to be accepted (i.e., $w \notin \lang{\aut}$), the sets $S_{j}$ in \labDie-labelled vertices should eventually become empty after finitely many steps, thus making $\vertex{S_{j}}{\ell}$ a finite vertex.
    Recall that empty sets will be removed in the construction of $\sdagAW{\aut}{w}$.
\item 
    \labInf-labelled vertices $\vertex{S_{j}}{\ell}$ are $\top$-vertices; 
    they are not yet a descendant of an $\acc$-vertex.
\item $
    \labNew$-labelled vertices are also descendants of an $\acc$-vertex and thus $\bot$-vertices in $\slcLabellingFunction$;
    the \labNew-labelled vertices should be inspected later by changing their labels to \labDie once all current \labDie-labelled vertices disappear. 
\end{itemize}
We omit the detailed \slc construction for general NBAs since
it can be found in~\cite{DBLP:conf/birthday/VardiW08,TsaiFVT14};
in Section~\ref{ssec:FANBAsliceBasedComplementation} we will present a similar but simplified construction for FANBAs.
The acceptance condition in \slc also requires the breakpoint set to be empty infinitely often. 
Assume that $w \in \lang{\aut^{c}}$:
this means that the breakpoint set becomes empty for infinitely many times, i.e., \labDie-labelled vertices disappear infinitely often.
It follows that all $\omega$-branches in $\sdagAW{\aut}{w}$ will have only \labInf-labelled vertices;
that is, all $\omega$-branches in $\sdagAW{\aut}{w}$ will have only $\top$-labels by the labelling function $\slcLabellingFunction$.
Then $w$ is not accepted by $\aut$ due to Lemma~\ref{lem:identify-nonacc-lambda-func}.

Assume that $w$ is not accepted by $\aut$.
We will have an integer $\stableLevel$ such that $\sdagAW{\aut}{w}$ will have only $\top$-labels by the labelling function $\slcLabellingFunction$.
Since we will guess each possible $\stableLevel$, we are guaranteed to obtain the labelling function $\slcLabellingFunction$ and the set of \labDie-labelled vertices will disappear infinitely often, i.e, the breakpoint becomes empty for infinitely many times.
Thus $w$ will be accepted by $\aut^{c}$.

The guess of the level $\stableLevel$ is done by a nondeterministic transition from a state of $\aut^{c}$ without the decorations in $\setnocond{\labDie, \labInf, \labNew}$ to a state with such a decoration.
A state of $\aut^{c}$ without a decoration can be encoded as a ordered sequence of states $S_{1}, \cdots, S_{k}$ since the level numbers can be omitted.
Obviously, the number $k$ of sets in a slice is at most the number $n$ of states in $\aut$.
According to~\cite{DBLP:journals/corr/FogartyKWV13}, the number of all possible states is approximately $(0.53n)^{n}$.
A state with the decoration will induce a blow-up of $4^{n}$ since a state may be labelled with \labDie, \labInf, \labNew, or none of them, thus leading to $4^{n} \times (0.53n)^{n} \in \bigO((3n)^{n})$.
Now we recall the correctness and the complexity results of the above slice-based construction:
\begin{proposition}[The Language and Size of $\aut^{c}$ for NBAs~\cite{DBLP:conf/birthday/VardiW08}]
\label{lem:size-language-slice}
    Given an NBA~$\aut$ with $n$ states, let $\aut^{c}$ be the NBA constructed by \slc.
    Then $\lang{\aut^{c}} = \infwords \setminus \lang{\aut}$ and $\aut^{c}$ has $\bigO((3n)^{n})$ states.
\end{proposition}

As a final note about the construction of codeterministic DAGs $\sdagAW{\aut}{w}$ over words given above, one can translate any given NBA with $n$ states to an FANBA with at most $3^{n}$ states and $n$ accepting runs per word~\cite{LodingP18}, where every run of the constructed FANBA over $w$ corresponds to a branch in the constructed $\sdagAW{\aut}{w}$.
This means that codeterministic DAGs are useful not only for complementing automata, but also in the translation of NBAs to more restrictive, yet equivalent, subclasses.

\subsection{Slice-Based Complementation Algorithm for FANBAs}
\label{ssec:FANBAsliceBasedComplementation}

We now propose a specialized slice-based complementation construction for FANBAs.
We first provide an overview of the algorithm and then we  present the technicalities.
According to Lemma~\ref{lem:FANBAwordRejectedStableLevelInCodeterministicDAG},  given a word $w \notin \lang{\aut}$, there exists a stable level $\stableLevel$ in the codeterministic DAG~$\reduceddagAW{\aut}{w}$ such that each $\acc$-vertex on a level after $\stableLevel$ is finite. 
Therefore, in the construction of $\aut^{c}$, we can nondeterministically guess such a level $\stableLevel$ and then use a breakpoint construction to verify that our guess is the correct one, in analogy with \rkc. 
This can be done since, when constructing the complementary NBA~$\aut^{c}$, we can identify two phases: 
the \emph{initial phase} and the \emph{accepting phase}.
While the term ``initial'' is common, the second phase has different names in literature: 
it is called ``accepting'' in~\cite{DBLP:journals/iandc/LiCZL20} and ``repetition'' in~\cite{DBLP:conf/birthday/VardiW08}.
We adopt the former since it recalls that the constructed NBA is a limit deterministic NBA (cf.\@ Remark~\ref{rem:slcGivesLDBAs}), with the initial phase in the nondeterministic part and the accepting phase in the deterministic part, corresponding to the constraints $\initialStates \subseteq \nstates$ and $\acc \subseteq \dstates$, respectively (cf.\@ Definition~\ref{def:typesOfBAsbyTransitions}).

In the initial phase, we trace the evolution of $\aut$ over $w$ by a pure subset construction that keeps track of the states of $\aut$ reached while reading $w$; 
note that such states are also occurring at the corresponding  level of the codeterministic DAG~$\reduceddagAW{\aut}{w}$. 
During the initial phase, on reading a letter while being in a macrostate of $\aut^{c}$, the macrorun of $\aut^{c}$ over $w$ decides to either remain in the initial phase or to jump to the accepting phase. 
Once entering the accepting phase, we bet that the macrorun of $\aut^{c}$, which represents multiple runs of $\aut$, has reached the stable level $\stableLevel$. 
To verify whether the bet is winning, we adopt a breakpoint construction that allows us to check whether all $\acc$-vertices after level $\stableLevel$ are finite.

In the accepting phase, we use as macrostate a triplet $(N, C, B)$ to encode the set of vertices and their labels as they appear on a level $\ell > \stableLevel$ in the codeterministic DAG~$\reduceddagAW{\aut}{w}$ (or $\sdagAW{\aut}{w}$ for general NBAs accordingly). 
Since we do not care about the actual level $\ell$, but only about the states we are visiting in such a level, we omit $\ell$ and identify vertices with their states:
suppose that we reach the set of states $S$ at two levels $\ell$ and $\ell' > \ell$;
the fact that a word is accepted depends only on whether from $S$ we will visit infinitely often accepting states, not on how we reached $S$ by reading the initial, finite fragment of $w$. 
This means that in the triplet $(N, C, B)$ we have that 
\begin{itemize}
\item 
    the set $N$ keeps all states of the reachable vertices on level $\ell$, corresponding to the set of all vertices labelled with \labDie, \labInf, and \labNew;
\item 
    the set $C$ keeps all states of the finite vertices on the level $\ell$. That means, it contains states of both \labNew-labelled vertices recording new encountered states and \labDie-labelled vertices being inspected now;
\item 
    the set $B \subseteq C$ as a breakpoint construction is used to verify that the guess on the set $C$ of finite vertices is correct, corresponding to the set of vertices labelled with \labDie,
\end{itemize} 
where $\ell$ is some level.
Recall that \labDie, \labInf, and \labNew are three labels of vertices used in \slc for complementing general NBAs, as described in Subsection~\ref{ssec:slice-run-dag}.
The specialized complementation algorithm for FANBAs is formalized below.
Recall that $\cotrans$ denotes the unique reduced transition function associated with the reduced DAGs of $\aut$ (cf.\@ Corollary~\ref{cor:FANBAuniqueReducedTransitionFunctionDAG} in Section~\ref{sec:reduced-run-dag}).
\begin{definition}[Slice-based complementation for FANBAs]
\label{def:FANBAsliceBasedComplementation}
    Given an FANBA~$\aut = (\states, \initialStates, \trans, \acc)$, let $\cotrans$ be the reduced transition function associated with the reduced DAGs for $\aut$.

    We define the NBA~$\aut^{c}= (\states^{c}, \initialStates^{c}, \trans^{c}, \acc^{c})$ as follows.
    \begin{itemize}
    \item 
        $\states^{c} \subseteq 2^{\states} \cup (2^{\states} \times 2^{\states} \times 2^{\states})$ is the smallest set such that for $\initialStates^{c}$ and $\trans^{c}$ as defined below, we have $\initialStates^{c} \subseteq \states^{c}$ and $\trans^{c}(\states^{c}, a) \subseteq \states^{c}$ for each $a \in \alphabet$;
    \item 
        $\initialStates^{c} = \setnocond{\initialStates}$;
    \item 
        $\trans^{c} = \ntrans^{c} \cup \jtrans^{c} \cup \dtrans^{c}$ where
        \begin{itemize}
        \item 
            $\ntrans^{c} \colon 2^{\states} \times \alphabet \to 2^{2^{\states}}$ is such that for each $S \subseteq \states$ and $a \in \alphabet$, we set $\ntrans^{c}(S, a) = \setnocond{\cotrans(S, a)}$;
        \item 
            $\dtrans^{c} \colon (2^{\states} \times 2^{\states} \times 2^{\states}) \times \alphabet \to (2^{\states}\times 2^{\states} \times 2^{\states})$ is such that for each $N, C, B \subseteq \states$ and $a \in \alphabet$, we define $\dtrans^{c}((N, C, B), a) = (N', C', B')$ where 
            \begin{itemize}
            \item 
                $N' = \cotrans(N, a)$,
            \item 
                $C' = \cotrans(C, a) \cup (N' \cap \acc)$, and
            \item 
                if $B \neq \emptyset$, then $B' = \cotrans(B, a)$, otherwise $B' = C'$;
            \end{itemize}
        \item 
            $\jtrans^{c} \colon 2^{\states} \times \alphabet \to 2^{2^{\states} \times 2^{\states} \times 2^{\states}}$ is such that for each $S \subseteq \states$ and $a \in \alphabet$, we set $\jtrans^{c}(S, a) = \setnocond{\dtrans^{c}((S, S \cap \acc, S \cap \acc), a)}$;
        \end{itemize}
    \item 
        $\acc^{c} = \setnocond{(N, C, \emptyset) \in \states^{c}}$.
    \end{itemize}
\end{definition}
In practice, the component $\ntrans^{c}$ and the macrostates $S \subseteq 2^{\states}$ follow a subset construction so to keep track of all possible current states of $\aut$ after having read a finite prefix of an $\omega$-word $w \in \infwords$.
The component $\jtrans^{c}$ makes a guess that it is time to establish whether all runs of $\aut$ are finite (thus, $w \notin \lang{\aut}$), so to let $\aut^{c}$ accept $w$; 
we do this by jumping from the current macrostate $S \subseteq 2^{\states}$ in the initial phase to the corresponding macrostate $(N', C', B') \in 2^{\states} \times 2^{\states} \times 2^{\states}$ in the accepting phase.
The component $\dtrans^{c}$ is responsible for checking whether the guess was correct:
the component $N$ is again the result of a subset construction, like for $\ntrans^{c}$;
the component $C$ traces the runs that visited accepting states after the guess point;
the component $B$ becomes empty once all runs have disappeared because the visited accepting states were finite.
Note that, for each $S \subseteq \states$ and $a \in \alphabet$, both $\ntrans^{c}(S,a)$ and $\jtrans^{c}(S,a)$ are singleton sets, so we can simplify their notation and for instance just write $\ntrans^{c}(S, a) = \cotrans(S, a)$ instead of $\ntrans^{c}(S, a) = \setnocond{\cotrans(S, a)}$.

\begin{figure}
    \centering
    \resizebox{\linewidth}{!}{
    \includegraphics{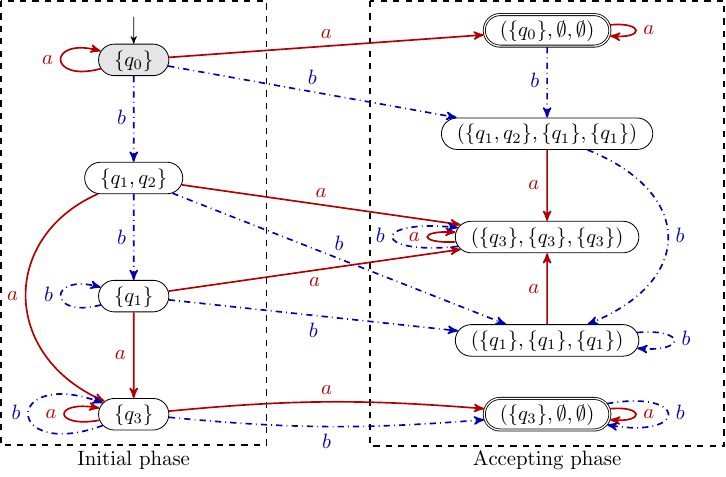}
    }
    \caption{The complementary NBA~$\aut^{c}$ of $\aut$ depicted in Figure~\ref{fig:example2ANBA} constructed by Definition~\ref{def:FANBAsliceBasedComplementation}.} 
    \label{fig:FANBAsliceBasedComplementation}
\end{figure}
\begin{markedexample}
\label{ex:FANBAsliceBasedComplementation}
    As an example of application of the construction given in Definition~\ref{def:FANBAsliceBasedComplementation}, consider the FANBA~$\aut$ depicted in Figure~\ref{fig:example2ANBA}. 
    By definition, the initial macrostates are $\initialStates^{c} = \setnocond{\initialStates}$, i.e., the only initial macrostate is $\setnocond{q_{0}} = \initialStates$.
    From this macrostate, we have the following transitions:
    on reading the letter $a$, we have the successors
    \begin{itemize}
    \item 
        $\setnocond{q_{0}} = \ntrans^{c}(\setnocond{q_{0}}, a)$ and 
    \item 
        $(\setnocond{q_{0}}, \emptyset, \emptyset) = \jtrans^{c}(\setnocond{q_{0}}, a) = \dtrans^{c}((\setnocond{q_{0}}, \emptyset, \emptyset), a)$. 
    \end{itemize}  
    Note that $(\setnocond{q_{0}}, \emptyset, \emptyset)$ is an accepting macrostate since its component $B$ is empty. 

    Still from the initial macrostate, after reading a letter $b$, we have the two successors $\setnocond{q_{1}, q_{2}}$ and $(\setnocond{q_{1}, q_{2}}, \setnocond{q_{1}}, \setnocond{q_{1}})$:
    \begin{itemize}
    \item 
        in the initial phase, $\setnocond{q_{1}, q_{2}} = \ntrans^{c}(\setnocond{q_{0}}, b)$; and
    \item 
        the transition from the initial phase to the accepting phase is given by $(\setnocond{q_{1}, q_{2}}, \setnocond{q_{1}}, \setnocond{q_{1}}) = \jtrans^{c}(\setnocond{q_{0}}, b) = \dtrans^{c}((\setnocond{q_{0}}, \emptyset, \emptyset), b)$.
    \end{itemize}
    All other (reachable) macrostates can be derived similarly;
    the resulting complementary NBA~$\aut^{c}$ is shown in Figure~\ref{fig:FANBAsliceBasedComplementation}. 
    It is easy to see that, for $\alphabet = \setnocond{a, b}$, $\lang{\aut} = \setcond{a^{i} b^{\omega} \in \alphabet^{\omega}}{i \in \naturals}$ and $\lang{\aut^{c}} = \setnocond{a^{\omega} \in \alphabet^{\omega}} \cup \setcond{a^{j} b b^{k} a \setnocond{a,b}^{\omega} \in \alphabet^{\omega}}{j, k \in \naturals}$, that is, $\lang{\aut^{c}} = \alphabet^{\omega} \setminus \lang{\aut}$.
\end{markedexample}

\begin{figure}
    \centering
    \resizebox{\linewidth}{!}{
    \includegraphics{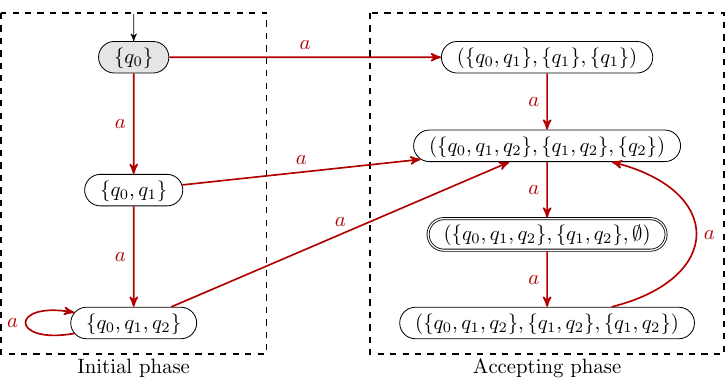}
    }
    \caption{The `$a$' fragment of $\aut[B]^{c}$ constructed by Definition~\ref{def:FANBAsliceBasedComplementation}, for the LDBA~$\aut[B]$ from Figure~\ref{fig:num-of-omega-branches}.} 
    \label{fig:FANBAsliceBasedComplementationTransCompare}
\end{figure}
\begin{markedremark}
\label{rem:transCompare}
    As we have seen in Remark~\ref{ex:reduceTransition}, by using $\trans$ and $\cotrans$ we get different sets of successors for the subsets of the currently reached states, i.e., the components $C$ and $B$ in Definition~\ref{def:FANBAsliceBasedComplementation}. 
    We now show, by an example, that in general $\cotrans$ cannot be replaced by $\trans$ in Definition~\ref{def:FANBAsliceBasedComplementation}: the automaton shown in Figure~\ref{fig:FANBAsliceBasedComplementationTransCompare} is the $a$-fragment of the complementary NBA~$\aut[B]^{c}$ of the LDBA~$\aut[B]$ depicted in Figure~\ref{fig:num-of-omega-branches}, constructed by Definition~\ref{def:FANBAsliceBasedComplementation}.
    By using the reduced transition $\cotrans$, since we have $\cotrans(\setnocond{q_{2}}, a) = \emptyset$, by reading the letter $a$ from the macrostate $(\setnocond{q_{0}, q_{1}, q_{2}}, \setnocond{q_{1}, q_{2}}, \setnocond{q_{2}})$, we reach the successor macrostate $(\setnocond{q_{0}, q_{1}, q_{2}}, \setnocond{q_{1}, q_{2}}, \emptyset)$, that is accepting since the $B$ component is empty.
    As we can see, the word $a^{\omega}$ is accepted by the resulting complement automaton~$\aut[B]^{c}_{\cotrans}$.
    
    If we would use the original transition $\trans$ instead of $\cotrans$ to construct the complement automaton $\aut[B]^{c}_{\trans}$, then the $a$-successor of $(\setnocond{q_{0}, q_{1}, q_{2}}, \setnocond{q_{1}, q_{2}}, \setnocond{q_{2}})$ would be $(\setnocond{q_{0}, q_{1}, q_{2}}, \setnocond{q_{1}, q_{2}}, \setnocond{q_{2}})$ itself, which becomes a non-accepting sink state.
    This implies that the $a$-fragment of the resulting  automaton $\aut[B]^{c}_{\trans}$ would have no accepting states, thus $a^{\omega}$ cannot be accepted when using $\trans$ instead of $\cotrans$.
    This means that $\aut[B]^{c}_{\trans}$ is not the complementary automaton of $\aut[B]$, since $a^{\omega}$ is not accepted by $\aut[B]$ but it is also not accepted by~$\aut[B]^{c}_{\trans}$.
\end{markedremark}

\begin{markedremark}
\label{rem:slcGivesLDBAs}
    It is worthwhile to note that the NBA~$\aut^{c}$ constructed according to Definition~\ref{def:FANBAsliceBasedComplementation} is \emph{limit deterministic}. 
    In fact, the set of macrostates $\states^{c}$ of $\aut^{c}$ can be partitioned into two disjoint sets $\nstates^{c} \subseteq 2^{\states}$ and $\dstates^{c} \subseteq 2^{\states} \times 2^{\states} \times 2^{\states}$ such that $\acc^{c} \subseteq \dstates^{c}$ and for each macrostate $q \in \dstates^{c}$ and $a \in \alphabet$, we have that $\size{\trans^{c}(q, a)} = 1$ and $\trans^{c}(q, a) \subseteq \dstates^{c}$.
\end{markedremark}

The NBA~$\aut^{c}$ is indeed a complementary automaton of the given FANBA~$\aut$, so the use of the superscript $c$ is correct, as formalized by the following theorem.
\begin{restatable}[The Language and Size of $\aut^{c}$ for FANBAs]{theorem}{FANBAlanguageSizeSliceBasedComplement}
\label{thm:FANBAlanguageSizeSliceBasedComplement}
    Let $\aut$ be an FANBA with $n$ states and $\aut^{c}$ be the NBA constructed according to Definition~\ref{def:FANBAsliceBasedComplementation}.
    Then $\lang{\aut^{c}} = \infwords \setminus \lang{\aut}$ and  $\aut^{c}$ has $2^{n} + 4^{n} \in \bigO(4^{n})$ macrostates. 
\end{restatable}
The complete proof is available in~\ref{app:FANBAlanguageSizeSliceBasedComplement}.
The proof of the direction $w \in \lang{\aut} \implies w \notin \lang{\aut^{c}}$ is based on the fact that the finitely many accepting runs of $\aut$ over $w \in \lang{\aut}$ eventually get trapped in the component $B$, so $\aut^{c}$ is not able to visit macrostates with $B = \emptyset$ infinitely often.
For the other direction $w \notin \lang{\aut} \implies w \in \lang{\aut^{c}}$, we use the fact that eventually no accepting state is visited anymore, thus they disappear from the components $B$ and $C$, hence $B$ becomes definitely empty making the corresponding run of $\aut^{c}$ accepting.

Based on Definition~\ref{def:FANBAsliceBasedComplementation}, we can establish a subsumption relation between the macrostates of $\aut^{c}$.
Recall that, given an NBA~$\aut[B] = (\states, \initialStates, \trans, \acc)$ and a set of states $S \subseteq \states$, we denote by $\autNewInitial{\aut[B]}{S}$ be NBA~$\autNewInitial{\aut[B]}{S} = (\states, S, \trans, \acc)$;
we may just write $\autNewInitial{\aut[B]}{q}$ instead of $\autNewInitial{\aut[B]}{\setnocond{q}}$ for $q \in \states$.
\begin{restatable}[Subsumption Relation between Macrostates]{proposition}{sumbsumptionRelationBetweenMacrostatesSLCforFANBAs}
\label{prop:sumbsumptionRelationBetweenMacrostatesSLCforFANBAs}
    Given an FANBA~$\aut$, let $\aut^{c}$ be the complementary NBA of $\aut$ constructed according to Definition~\ref{def:FANBAsliceBasedComplementation}. 
    For each pair of macrostates $m = (N, C, B)$ and $m' = (N', C', B')$ of $\aut^{c}$ such that $N = N'$ and $C \subseteq C'$, it holds that $\lang{\autNewInitial{\aut^{c}}{m'}} \subseteq \lang{\autNewInitial{\aut^{c}}{m}}$;
    that is, $m$ subsumes $m'$.
\end{restatable}
The proof for this proposition is given in~\ref{app:otherProofs}.
The idea underlying this result is that  $\autNewInitial{\aut}{N}$ and $\autNewInitial{\aut}{N'}$ are the same automaton, so they have the same codeterministic DAG and the same stable level. 
Since in $\autNewInitial{\aut^{c}}{m'}$ the set $B'$ becomes empty infinitely often and $B'$ is reset to $C'$, then also $B$ becomes empty infinitely often since it is reset to $C \subseteq C'$.

The subsumption relation defined in Proposition~\ref{prop:sumbsumptionRelationBetweenMacrostatesSLCforFANBAs} can help in reducing the number of macrostates of $\aut^{c}$ we need to explore when constructing $\aut^{c}$ itself:
before generating the successors of the current macrostate $m'$, we can look for an already generated macrostate $m$ subsuming $m'$ such that $\lang{\autNewInitial{\aut^{c}}{m}} = \emptyset$.
If we can find such a macrostate $m$, then we already know that $\lang{\autNewInitial{\aut^{c}}{m'}} = \emptyset$ so we can avoid to generate all possible successors of $m'$;
this may help in reducing the number of macrostates in $\aut^{c}$ from a practical perspective.
Similarly, we can look for such macrostates $m$ while checking language inclusion between an NBA~$\aut[B]$ and the FANBA~$\aut$, since we can avoid to explore a macrostate $m'$ subsumed by some macrostate $m$ when we already know that the language from such a macrostate will be empty.
In fact, checking whether $\lang{\aut[B]} \subseteq \lang{\aut}$ reduces to check whether $\lang{\aut[B]} \cap \lang{\aut^{c}} = \emptyset$; 
this is done by constructing the product automaton $\aut[B] \times \aut^{c}$ and checking whether its language is empty (see, e.g.,~\cite{DBLP:reference/mc/2018} for details).
During the construction of $\aut[B] \times \aut^{c}$, every time we need to explore the product states reachable from a product state $(q,m')$ of $\aut[B] \times \aut^{c}$, we can first try to find a macrostate $m$ such that $m$ subsumes $m'$ and $\lang{\autNewInitial{\aut^{c}}{m}} = \emptyset$ or, more generally, a product state $(q,m)$ such that $m$ subsumes $m'$ and $\lang{\autNewInitial{\aut[B] \times \aut^{c}}{(q,m)}} = \emptyset$.
If this is the case, then we do not have to explore the macrostates reachable from $m'$ since by Proposition~\ref{prop:sumbsumptionRelationBetweenMacrostatesSLCforFANBAs} we know that $\lang{\autNewInitial{\aut^{c}}{m'}} \subseteq \lang{\autNewInitial{\aut^{c}}{m}} = \emptyset$ (and similarly $\lang{\autNewInitial{\aut[B] \times \aut^{c}}{(q,m')}} = \emptyset$ if $\lang{\autNewInitial{\aut[B] \times \aut^{c}}{(q,m)}} = \emptyset$), thus also the language of the product automaton will be empty from the product state $(q, m')$.
We refer to~\cite{DBLP:journals/corr/abs-0902-3958,DBLP:conf/cav/AbdullaCCHHMV10,DBLP:conf/concur/AbdullaCCHHMV11} for an efficient discovery method for such macrostates $m$.

\section{Application to Limit Deterministic \buchi Automata}
\label{sec:applicationToLDBAs}

In the previous two sections, we have shown how to use codeterministic DAGs for the complementation of general NBAs and how it can be optimized for FANBAs. 
The same idea can be applied also in complementing limit deterministic \buchi automata (LDBAs).
In Subsection~\ref{ssec:LDBAcodeterministicDAGconstruction}, we describe the codeterministic DAGs for LDBAs, which later will be used to produce a complementary NBA~$\aut^{c}$ of an LDBA~$\aut$.
Then, in Subsection~\ref{ssec:LDBAcomplementationConstruction}, we show the complementation construction for LDBAs.

\subsection{Codeterministic DAGs for LDBAs}
\label{ssec:LDBAcodeterministicDAGconstruction}

In this section, we consider limit deterministic \buchi automata. 
Given a \buchi automaton, it is limit deterministic if it behaves deterministically after the first visit of an accepting state. 
Recall, as implied by Definition~\ref{def:typesOfBAsbyTransitions}, that an NBA~$\aut = (\states, \initialStates, \trans, \acc)$ is a \emph{limit deterministic \buchi automaton} if, for each $q_{f} \in \acc$, the reachable fragment of the automaton $(\states, \setnocond{q_{f}}, \trans, \acc)$ is deterministic. 
Each LDBA can be divided into two parts: 
one is the deterministic part, which consists of states that are reachable from the accepting states; 
the other is the nondeterministic part that consists of all other states. 
As in Definition~\ref{def:typesOfBAsbyTransitions}, we denote the set of states in the deterministic part by $\dstates$ and the set of states in the nondeterministic part by $\nstates$. 
We can also partition the transition function into three parts: $\ntrans$ relative to states in $\nstates$ only, $\dtrans$ relative to states in $\dstates$ only, and $\jtrans$ connecting states in $\nstates$ to states in $\dstates$.
For uniformity of presentation and in order to simplify the notation, we might just consider $\dtrans$ as a function mapping pairs of states and letters to singletons, instead of states.

We now describe the LDBA specific construction of the codeterministic DAG~$\lcodagAW{\aut}{w} = \graph{\vertices}{\edges}$ for a given LDBA~$\aut$ and $\omega$-word $w$. 
Similarly to the partition of the states of the LDBA~$\aut$, we  split the vertices of $\lcodagAW{\aut}{w}$ into two parts $\vertices = \nvertices \cup \dvertices$: 
the nondeterministic part $\nvertices$ contains all vertices $\vertex{S}{\ell} \in \vertices$ such that $S \subseteq \nstates$ is the set of nondeterministic states reached by $\aut$ after reading the first $\ell$ letters of $w$; 
in the deterministic part $\dvertices$ there are all vertices $\vertex{q}{\ell}$ such that $q \in \dstates$ is a deterministic state reached by $\aut$ after jumping from $\nstates$ to $\dstates$ at some point while reading the first $\ell$ letters of $w$.
This means that at level $0$, we have the following vertices in $\lcodagAW{\aut}{w}$: 
a single vertex $\vertex{\initialStates \cap \nstates}{0} \in \nvertices$ if $\initialStates \cap \nstates \neq \emptyset$ and a vertex $\vertex{q}{0} \in \dvertices$ for each $q \in \initialStates \cap \dstates$. 
Of course, there will be no $\dvertices$-vertices if $\initialStates \cap \dstates = \emptyset$.

Regarding the edges of $\lcodagAW{\aut}{w}$, that is, how to construct the vertices on level $i+1$, we adopt different approaches depending on the type of vertices we start from and we are going to produce.
Given the vertex $\vertex{S_{i}}{i} \in \nvertices$, we generate only one successor vertex $\vertex{S_{i+1}}{i+1} \in \nvertices$ by the standard subset construction on $\ntrans$, that is, $S_{i+1} = \ntrans(S_{i}, \wordletter{w}{i})$.
(Note that the definition of $\ntrans$ ensures that $S_{i+1} \subseteq \nstates$.)
We also generate one successor vertex $\vertex{q}{i+1}$ for each $q \in \jtrans(S_{i}, \wordletter{w}{i}) \subseteq \dstates$.
Given a vertex $\vertex{q}{i} \in \dvertices$, we generate the single successor vertex $\vertex{\dtrans(q, \wordletter{w}{i})}{i+1} \in \dvertices$.

It is worthwhile to note that the deterministic nature of $\dtrans$ ensures that $\dvertices$ is closed under reachability, that is, only vertices $\vertex{q'}{\ell'} \in \dvertices$ can be reached from vertices $\vertex{q}{\ell} \in \dvertices$.
Symmetrically, the subset nature of the construction of nondeterministic vertices ensures that each vertex $\vertex{S'}{\ell'} \in \nvertices$ is only reachable from vertices $\vertex{S}{\ell} \in \nvertices$;
moreover, $\vertex{S'}{\ell'}$ has only one predecessor, namely $\vertex{S}{\ell' - 1} \in \nvertices$ with $S' = \ntrans(S, \wordletter{w}{\ell' - 1})$, provided that $\ell' > 0$.
The way we produce vertices at level $i+1$ does not ensure that the resulting DAG~$\lcodagAW{\aut}{w}$ is codeterministic:
there may be duplicate deterministic states in the successors of states from level $i$, for instance because $\dtrans(q, \wordletter{w}{i}) \in \jtrans(S_{i}, \wordletter{w}{i})$ or because $\dtrans(q, \wordletter{w}{i}) = \dtrans(q', \wordletter{w}{i})$ for vertices $\vertex{q}{i}, \vertex{q'}{i} \in \dvertices$. 
This means that we have to remove all these edges but one. 
We do this by assigning priorities to every candidate edge resulting from the construction of the vertices at level $i+1$; we then keep only the candidate edge with the lowest priority. 

In practice, we use a priority function $\priority \colon \vertices \to \naturals \setminus \setnocond{0}$ to assign priorities to the vertices.
The main idea is that deterministic vertices have priority over nondeterministic vertices and that vertices inherit the minimum priority of their predecessors.
This means that a nondeterministic vertex can only inherit its priority from another nondeterministic vertex while a deterministic vertex gets its priority from its least deterministic predecessor, provided there is one. 
If there is no deterministic predecessor, then its priority depends on the one of its unique nondeterministic predecessor and its deterministic siblings.

As usual, let $\enumerationnocond{q_{1}, q_{2}, \dotsc, q_{n}}$ be the enumeration of the states $\states$ according to some total order $\stateOrder$ over states such that $d \stateOrder n$ whenever $d \in \dstates$ and $n \in \nstates$, that is, deterministic states ``come earlier than'' nondeterministic states in $\stateOrder$.
Recall that at level $0$ we have a deterministic vertex $\vertex{\initial_{j}}{0}$ for each $\initial_{j} \in \initialStates \cap \dstates$ and the nondeterministic vertex $\vertex{\initialStates \cap \nstates}{0}$, provided that $\initialStates \cap \nstates \neq \emptyset$.
Let $\enumerationnocond{\initial_{1}, \cdots, \initial_{k}}$ be the enumeration of $\initialStates \cap \dstates$ under $\stateOrder$.
We define the priority of the vertices at level $0$ as $\priority(\vertex{\initial_{j}}{0}) = j$ for each $\initial_{j} \in \enumerationnocond{\initial_{1}, \cdots, \initial_{k}}$ and $\priority(\vertex{\initialStates \cap \nstates}{0}) = 1 + \size{\initialStates \cap \dstates}$.

We now show how to assign priorities to the vertices at level $\ell + 1$ provided that we already have the priority of the vertices at level $\ell$.
Recall that we generate the single vertex $\vertex{S_{\ell + 1}}{\ell + 1} \in \nvertices$ from the unique $\vertex{S_{\ell}}{\ell} \in \nvertices$, where $S_{\ell + 1} = \ntrans(S_{\ell}, \wordletter{w}{\ell})$, as well as the vertices $\vertex{q}{\ell + 1} \in \dvertices$ for each $q \in \jtrans(S_{\ell}, \wordletter{w}{\ell})$.
Lastly, from each $\vertex{q}{\ell} \in \dvertices$, we generate the vertex $\vertex{\dtrans(q, \wordletter{w}{\ell})}{\ell + 1} \in \dvertices$.
Let 
\begin{itemize}
\item 
    $P_{\ell} = \enumerationcond{p_{k} \in \dstates}{\vertex{p_{k}}{\ell} \in \dvertices}$ be the enumeration under $\stateOrder$ of the deterministic states occurring in the deterministic vertices at level $\ell$;
\item
    $D_{\ell + 1} = \enumerationcond{d_{k} \in \dstates}{d_{k} \in \dtrans(P_{\ell}, \wordletter{w}{\ell})}$ be the enumeration under $\stateOrder$ of the $\wordletter{w}{\ell}$-successor states of $P_{\ell}$ under $\dtrans$; 
    and
\item
    $J_{\ell + 1} = \enumerationcond{j_{k} \in \dstates \setminus D_{\ell + 1}}{j_{k} \in \jtrans(S_{\ell}, \wordletter{w}{\ell})}$ be the enumeration under $\stateOrder$ of the deterministic states not in $D_{\ell + 1}$ that are $\wordletter{w}{\ell}$-successor states of $S_{\ell}$ under $\jtrans$.
\end{itemize}
Recall that the index $k$ of the enumerations $P_{\ell}$, $D_{\ell + 1}$, and $J_{\ell + 1}$ starts from $1$. 
We define the priority function $\priority$ on the vertices at level $\ell + 1$ as 
\begin{itemize}
\item 
    $\priority(\vertex{d}{\ell + 1}) = \min \setcond{\priority(\vertex{p}{\ell})}{p \in P_{\ell} \land d = \dtrans(p, \wordletter{w}{\ell})}$ for each $d \in D_{\ell + 1}$, that is, we assign to each vertex $\vertex{d}{\ell + 1}$ the minimum priority of its deterministic predecessors; 
\item
    $\priority(\vertex{j_{k}}{\ell + 1}) = \priority(\vertex{S_{\ell}}{\ell}) + k - 1$ for each $j_{k} \in J_{\ell + 1}$, that is, we assign to each deterministic vertex that is a successor only of $\vertex{S_{\ell}}{\ell}$ the same priority of $\vertex{S_{\ell}}{\ell}$ increased by $k - 1$ to keep track of the position of $j_{k}$ in $J_{\ell + 1}$;
    and
\item
    $\priority(\vertex{S_{\ell + 1}}{\ell + 1}) = \priority(\vertex{S_{\ell}}{\ell}) + \size{J_{\ell + 1}}$, that is we assign to $\vertex{S_{\ell + 1}}{\ell + 1}$ the priority of its only predecessor $\vertex{S_{\ell}}{\ell}$ increased by $\size{J_{\ell + 1}}$ to ensure that the deterministic vertices whose states are in $J_{\ell + 1}$ get a priority lower than $\priority(\vertex{S_{\ell + 1}}{\ell + 1})$.
\end{itemize}

The above construction of the priority function $\priority$ can be formalized as follows.
\begin{definition}[Codeterministic DAG for LDBAs]
\label{def:LDBAcodeterministicDAGandPriorityFunction}
    Given an LDBA~$\aut = (\states, \initialStates, \trans, \acc)$ with partition of states $\nstates$ and $\dstates$ and corresponding transition relation $\trans = \ntrans \cup \jtrans \cup \dtrans$, an $\omega$-word $w \in \infwords$, and a total order $\stateOrder$ over states such that $d \stateOrder n$ whenever $d \in \dstates$ and $n \in \nstates$, let $\vertices = \nvertices \cup \dvertices$ where $\nvertices \subseteq 2^{\nstates} \times \naturals$ and $\dvertices \subseteq \dstates \times \naturals$ are the smallest sets such that 
    \[
        \arraycolsep=2mm
        \begin{array}{ll}
             \vertex{\initialStates \cap \nstates}{0} \in \nvertices & \text{if $\initialStates \cap \nstates \neq \emptyset$;} \\ 
             \vertex{\ntrans(S_{\ell}, \wordletter{w}{\ell})}{\ell + 1} \in \nvertices & \text{if $\vertex{S_{\ell}}{\ell} \in \nvertices$ and $\ntrans(S_{\ell}, \wordletter{w}{\ell}) \neq \emptyset$;} \\
             \vertex{q}{0} \in \dvertices & \text{for each $q \in \initialStates \cap \dstates$;} \\
             \vertex{q}{\ell + 1} \in \dvertices & \text{for each $q \in \jtrans(S_{\ell}, \wordletter{w}{\ell})$ where $\vertex{S_{\ell}}{\ell} \in \nvertices$;} \\
             \vertex{\dtrans(q, \wordletter{w}{\ell})}{\ell + 1} & \text{for each $\vertex{q}{\ell} \in \dvertices$.}
        \end{array}
    \]
    We call $\priority \colon \vertices \to \naturals \setminus \setnocond{0}$ a \emph{priority function} if $\priority$ satisfies the following constraints:
    \[
        \arraycolsep=2mm
        \begin{array}{ll}
             \priority(\vertex{\initial_{k}}{0}) = k & \text{for each $\initial_{k} \in D_{0}$;} \\
             \priority(\vertex{\initialStates \cap \nstates}{0}) = 1 + \size{\initialStates \cap \dstates} & \text{if $\initialStates \cap \nstates \neq \emptyset$;} \\
             \priority(\vertex{d}{\ell + 1}) = \min \setcond{\priority(\vertex{p}{\ell})}{p \in P_{\ell} \land d = \dtrans(p, \wordletter{w}{\ell})} & \text{for each $d \in D_{\ell + 1}$;}\\
             \priority(\vertex{j_{k}}{\ell + 1}) = \priority(\vertex{S_{\ell}}{\ell}) + k - 1 & \text{for each $j_{k} \in J_{\ell + 1}$;} \\
             \priority(\vertex{S_{\ell + 1}}{\ell + 1}) = \priority(\vertex{S_{\ell}}{\ell}) + \size{J_{\ell + 1}} & \text{if $\vertex{S_{\ell}}{\ell} \in \nvertices$,}
        \end{array}
    \]
    where $\ell \in \naturals$ and 
    \begin{align*}
        D_{0} & = \enumerationcond{\initial_{k} \in \dstates}{\vertex{\initial_{k}}{0} \in \dvertices} \\
        P_{\ell} & = \enumerationcond{p_{k} \in \dstates}{\vertex{p_{k}}{\ell} \in \dvertices} \\
        D_{\ell + 1} & = \enumerationcond{d_{k} \in \dstates}{d_{k} \in \dtrans(P_{\ell}, \wordletter{w}{\ell})} \\
        J_{\ell + 1} & = \enumerationcond{j_{k} \in \dstates \setminus D_{\ell + 1}}{j_{k} \in \jtrans(S_{\ell}, \wordletter{w}{\ell})}
    \end{align*}
    
    The \emph{codeterministic DAG} $\lcodagAW{\aut}{w} = \graph{\vertices}{\edges}$ of $\aut$ over $w$ has the set $\vertices$ defined above as vertices, and the set of edges $\edges \subseteq (\nvertices^{2}) \cup (\nvertices \times \dvertices) \cup (\dvertices^{2})$ is the smallest set such that
    \[
        \arraycolsep=2mm
        \begin{array}{ll}
            (\vertex{S_{\ell}}{\ell}, \vertex{S_{\ell +1}}{\ell + 1}) \in \edges & \text{if $\vertex{S_{\ell}}{\ell}, \vertex{S_{\ell + 1}}{\ell + 1} \in \nvertices$ and $S_{\ell + 1} = \ntrans(S_{\ell}, \wordletter{w}{\ell})$;} \\
            (\vertex{S_{\ell}}{\ell}, \vertex{q}{\ell + 1}) \in \edges & \text{if $\vertex{S_{\ell}}{\ell} \in \nvertices$, $\vertex{q}{\ell + 1} \in \dvertices$, and $q \in J_{\ell + 1}$;} \\
            (\vertex{p}{\ell}, \vertex{d}{\ell + 1}) \in \edges & \text{if $\vertex{p}{\ell}, \vertex{d}{\ell + 1} \in \dvertices$, $p \in P_{\ell}$, $d = \dtrans(p, \wordletter{w}{\ell})$, and} \\ & \text{\hphantom{if} $\priority(\vertex{d}{\ell + 1}) = \priority(\vertex{p}{\ell})$.}
        \end{array}
    \]
\end{definition}

\begin{markedremark}
\label{rem:LDBAcodeterministicDAGisIndeedCodeterministic}
    Note that in the above definition, we say that the run DAG~$\lcodagAW{\aut}{w}$ is codeterministic.
    This is indeed the case because the only situation where a vertex might have two or more predecessors is when it is a deterministic vertex $\vertex{d}{\ell+1}$ with $d$ being in the image under $\jtrans$ or $\dtrans$ of at least two states of $\aut$.
    By construction, it must be the case that at least one of such states is deterministic, since nondeterministic states contribute only one candidate edge given that they are grouped in a single set.
    By the definition of the priority $\priority$, nondeterministic vertices get a priority always larger than deterministic ones, so they will never be used for defining an edge to $\vertex{d}{\ell+1}$, since there is at least one candidate edge from a deterministic vertex.
    Since all deterministic vertices in the level $\ell+1$ get a priority that is increasing according to their enumeration, there is only one vertex $\vertex{m}{\ell}$ that has the minimal priority and has $\vertex{d}{\ell+1}$ as successor.
    This means that only $\vertex{m}{\ell}$ fulfills the condition to have the edge $(\vertex{m}{\ell},\vertex{d}{\ell+1})$.
\end{markedremark}

\begin{figure}
    \centering
    \includegraphics{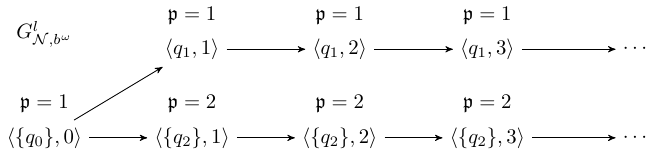}
    \caption{The codeterministic run DAG~$\lcodagAW{\aut[N]}{b^{\omega}}$ of the LDBA~$\aut[N]$ depicted in Figure~\ref{fig:exampleBA}.} 
    \label{fig:runDAGlDBA}
\end{figure}
\begin{markedexample}
\label{ex:runDAGlDBA}
    As an example of priority function, we consider again the NBA~$\aut[N]$ depicted in Figure~\ref{fig:exampleBA} and the word $b^{\omega}$. 
    As we have seen in Example~\ref{ex:typesOfBAs}, $\aut[N]$ is also an LDBA, with $\nstates = \setnocond{q_{0}, q_{2}}$ and $\dstates = \setnocond{q_{1}, q_{3}}$ as partition of the states.
    Regarding the vertices of $\lcodagAW{\aut}{w}$, we have the following vertices and priorities. 
    \begin{description}
    \item[Level $0$: ]
        there is only the nondeterministic vertex $\vertex{\setnocond{q_{0}}}{0} \in \nvertices$, since $\initialStates = \setnocond{q_{0}} \subseteq \nstates$ and $\initialStates \cap \dstates = \emptyset$.
        This means that its priority is $\priority(\vertex{\setnocond{q_{0}}}{0}) = 1 + \size{\initialStates \cap \dstates} = 1$.
    \item[Level $1$:]
        for this level, as vertices, we have the nondeterministic vertex $\vertex{\setnocond{q_{2}}}{1} \in \nvertices$ and the deterministic vertex $\vertex{q_{1}}{1} \in \dvertices$ coming from $\jtrans$. 
        So we get $S_{0} = \setnocond{q_{0}}$, $S_{1} = \setnocond{q_{2}}$, $P_{0} = D_{1} = \emptyset$ and $J_{1} = \enumerationnocond{q_{1}}$.
        
        As priority, we have $\priority(\vertex{q_{1}}{1}) = \priority(\vertex{\setnocond{q_{0}}}{0}) + 1 - 1 = 1$ for the deterministic vertices, and priority $\priority(\vertex{\setnocond{q_{2}}}{1}) = \priority(\vertex{\setnocond{q_{0}}}{0}) + \size{J_{1}} = 1 + 1 = 2$ for the nondeterministic vertices.
    \item[Level $2$:]
        for this level, we have the nondeterministic vertex $\vertex{\setnocond{q_{2}}}{2} \in \nvertices$ since $\setnocond{q_{2}} = \ntrans(S_{1}, \wordletter{w}{2}) = \ntrans(\setnocond{q_{2}}, b)$;
        we also have the deterministic vertex $\vertex{q_{1}}{2}$.
        Since $q_{1} \in D_{2}$ as a consequence of $q_{1} \in \dtrans(P_{1}, \wordletter{w}{2})$, we have $S_{1} = S_{2} = \setnocond{q_{2}}$, $P_{1} = D_{2} = \enumerationnocond{q_{1}}$, and $J_{2} = \emptyset$.
        Note that $J_{2} = \emptyset$ because $q_{1} \in \jtrans(S_{1}, b)$ but $q_{1} \in \dtrans(q_{1}, b)$ too, so $q_{1} \in D_{2}$ which prevents $q_{1}$ to belong to $J_{2}$.
        
        As priority, we just have $\priority(\vertex{q_{1}}{2}) = \min \setcond{\priority(\vertex{p}{1})}{p \in P_{1} \land d = \dtrans(p, \wordletter{w}{1})} = \priority(\vertex{q_{1}}{1}) = 1$ for the deterministic vertices, and priority $\priority(\vertex{\setnocond{q_{2}}}{2}) = \priority(\vertex{\setnocond{q_{2}}}{1}) + \size{J_{2}} = 2 + 0 = 2$ for the nondeterministic vertices.
    \item[Level $\ell \geq 3$:]
        for level $3$, there is just the nondeterministic vertex $\vertex{\setnocond{q_{2}}}{3}$, since $S_{2} = \setnocond{q_{2}}$; 
        we get that all successors $S_{\ell}$ are $\setnocond{q_{2}}$ as well, so as nondeterministic vertices only $\vertex{\setnocond{q_{2}}}{\ell}$ will be at level $\ell$.
        About the deterministic vertices, we only get $\vertex{q_{1}}{3}$ from $q_{1} = \dtrans(q_{1}, b)$; it follows that $P_{\ell - 1} = D_{\ell} = \enumerationnocond{q_{1}}$ and $J_{\ell} = \emptyset$, so as deterministic vertices only $\vertex{q_{1}}{\ell}$ will be at level $\ell$.
        
        As priority, we have $\priority(\vertex{q_{1}}{3}) = \min \setcond{\priority(\vertex{p}{2})}{p \in P_{2} \land d = \dtrans(p, \wordletter{w}{2})} = \priority(\vertex{q_{1}}{2}) = 1$, and similarly for all other deterministic vertices $\vertex{q_{1}}{\ell}$.
        Also $\priority(\vertex{\setnocond{q_{2}}}{3}) = \priority(\vertex{\setnocond{q_{2}}}{2}) + \size{J_{3}} = 2 + 0 = 2$, and similarly for all other nondeterministic vertices $\vertex{\setnocond{q_{2}}}{\ell}$.
    \end{description}
    The resulting codeterministic DAG~$\lcodagAW{\aut[N]}{b^{\omega}}$ is shown in Figure~\ref{fig:runDAGlDBA}.
    From this codeterministic DAG, we can see how priorities are propagated through the branches: 
    the priority of a deterministic vertex is equal to the one of its predecessor; 
    the priority of a nondeterministic vertex depends on its predecessor as well as its deterministic siblings.
\end{markedexample}

\begin{markedexample}
    Consider now the FANBA~$\aut$ shown in Figure~\ref{fig:example2ANBA} and again the $\omega$-word $b^{\omega}$.
    By Example~\ref{ex:typesOfBAs}, we know that $\aut$ is also an LDBA, with the same partition $\nstates = \setnocond{q_{0}, q_{2}}$ and $\dstates = \setnocond{q_{1}, q_{3}}$ as $\aut[N]$.
    The codeterministic DAG~$\lcodagAW{\aut}{b^{\omega}}$ is similar to $\lcodagAW{\aut[N]}{b^{\omega}}$ constructed in Example~\ref{ex:runDAGlDBA} and depicted in Figure~\ref{fig:runDAGlDBA};
    the only difference is that the nondeterministic vertex $\vertex{\setnocond{q_{2}}}{1}$ has no successor, so all nondeterministic vertices $\vertex{\setnocond{q_{2}}}{\ell}$ with $\ell > 1$ are not present.
\end{markedexample}

Codeterministic DAGs for LDBAs enjoy the same properties about acceptance, the number of $\omega$-branches, and stable level as those for FANBAs stated in Lemmas~\ref{lem:FANBAacceptanceCodeterministicDAG},~\ref{lem:FANBAfiniteOmegaBranchesInCodeterministicDAG}, and~\ref{lem:FANBAwordRejectedStableLevelInCodeterministicDAG}.
\begin{restatable}[Properties of Codeterministic DAGs for LDBAs]{lemma}{LDBAPropertiesOfCodeterministicDAG}
\label{lem:LDBAPropertiesOfCodeterministicDAG}
    Given an LDBA~$\aut$ and $w \in \infwords$, let $\lcodagAW{\aut}{w}$ be the codeterministic DAG for $\aut$ over $w$ as defined in Definition~\ref{def:LDBAcodeterministicDAGandPriorityFunction}.
    Then the following properties hold:
    \begin{enumerate}[\bf Property~1:]
    \item 
    \label{lem:LDBAPropertiesOfCodeterministicDAG:acceptance}
        $w$ is accepted by $\aut$ if and only if $\lcodagAW{\aut}{w}$ is accepting;
    \item
    \label{lem:LDBAPropertiesOfCodeterministicDAG:numberOmegaBranches}
        the number of (accepting) $\omega$-branches in $\lcodagAW{\aut}{w}$ is at most the number of states in $\aut$;
    \item 
    \label{lem:LDBAPropertiesOfCodeterministicDAG:stableLevel}
        there exists a stable level $\stableLevel \geq 1$ in $\lcodagAW{\aut}{w}$ such that all $\acc$-vertices after level $\stableLevel$ are finite if and only if $w \notin \lang{\aut}$.
    \end{enumerate}
\end{restatable}
We present the proof for this lemma in~\ref{app:LDBAPropertiesOfCodeterministicDAG}.
The idea underlying the proof of Property~\ref{lem:LDBAPropertiesOfCodeterministicDAG:acceptance} is to connect each run of $\aut$ over $w$ with its representative in $\lcodagAW{\aut}{w}$ and showing that the acceptance of one is equivalent to the acceptance of the other;
Property~\ref{lem:LDBAPropertiesOfCodeterministicDAG:numberOmegaBranches} is analogous to Lemma~\ref{lem:FANBAfiniteOmegaBranchesInCodeterministicDAG} and they share similar proofs;
lastly, the proof of Property~\ref{lem:LDBAPropertiesOfCodeterministicDAG:stableLevel} uses the other two properties to relate the visits to accepting states of a run of $\aut$ over $w$ with the level in the corresponding branch in $\lcodagAW{\aut}{w}$.

\subsection{Complementation Construction for LDBAs}
\label{ssec:LDBAcomplementationConstruction}

In this section, we first describe the complementation construction for a given LDBA, based on a classification of the states visited in a run into four sets $N$, $S$, $B$, and $C$;
then we formalize the algorithm in Definition~\ref{def:LDBAnsbcComplementation}.

We first provide some intuition underlying the construction. 
To encode the exact transition relation between the states of two consecutive levels in $\lcodagAW{\aut}{w}$, we would need to store the priority function $\priority$, which results in the complementation complexity $2^{\bigO(n \log n)}$, much higher than $\bigO(4^{n})$ obtained in~\cite{Blahoudek16}.
However, if we are able to make nondeterministic guesses, then we can avoid to store the exact priorities for the runs.
In fact, according to Lemma~\ref{lem:LDBAPropertiesOfCodeterministicDAG}, given a word $w \notin \lang{\aut}$, there exists a stable level $\stableLevel$ in the codeterministic DAG~$\lcodagAW{\aut}{w}$ such that each $\acc$-vertex at a level after $\stableLevel$ is finite. 
Recall that after reaching the stable level, runs on $w \notin \lang{\aut}$ entering the deterministic part $\dstates$ are either finite or safe to be used to accept $w$ in $\aut^{c}$.
Since all runs in $\dstates$ behave deterministically, the number of safe runs eventually stabilizes and all runs newly entering $\dstates$ must be finite.
Otherwise, if a run that can enter $\dstates$ is infinite (thus not merged and cut off later), we can guess a successive level to make it enter $\dstates$, that is, we postpone its entrance to the accepting phase to a later moment.
The number of infinite runs in $\dstates$ is at most $\size{\dstates}$ since $\aut$ behaves deterministically in $\dstates$.
Therefore, we only postpone the entrance of runs in the accepting phase for a finite number of times.
After the time point where the number of safe runs is fixed, we can easily put those runs into a safe set $S$.
For the runs keeping jumping from $\nstates$, we use a collector set $C$ as a buffer before they are inspected in the breakpoint $B$.
Note that we only need to keep the relative priorities between the runs in these sets.
That is, we give priority to $S$ (without visiting accepting states), then $B$, $C$, and $N$ come in order.
This means that, similarly to the slice-based algorithm for FANBAs, in order to construct $\aut^{c}$ we can also nondeterministically guess the stable level $\stableLevel$ and then use a breakpoint construction to verify that our guess was correct. 
This means that during the construction of the complementary NBA~$\aut^{c}$, we still use the initial and accepting phases.

Analogously to the \slc algorithm, the initial phase follows a pure subset construction on all states $\states$ to trace the reachable states of each level of the codeterministic DAG~$\lcodagAW{\aut}{w}$ over $w$.
From each of these macrostates we can choose whether to remain in the initial phase, or to jump to the accepting phase.
If we decide to jump, this means that we guess that the macrorun of $\aut^{c}$, which consists of multiple runs of $\aut$, has reached the stable level $\stableLevel$. 
Thus in the accepting phase, we adopt a breakpoint construction to verify that we guessed correctly, i.e., that all $\acc$-vertices after level $\stableLevel$ are finite.

Differently from the slice-based algorithm, in the accepting phase we use quadruplets $(N, S, B, C)$ as macrostates to organize the set of vertices on a level after $\stableLevel$ in the codeterministic DAG~$\lcodagAW{\aut}{w}$, where the four sets are used for the following purposes.
\begin{itemize}
\item 
    The set $N \subseteq \nstates$ just continues to follow the \emph{nondeterministic} states, similarly to the initial phase. 
    This allows us to manage jumps from the nondeterministic states that we have not left yet.
\item 
    The set $S \subseteq \dstates$ is used to keep track of the runs that will not visit accepting states anymore after the jump. 
    If a run remains in $S$ forever, then such a run of $\aut$ is for sure not accepting over $w$, so it is \emph{safe} to be used to accept the word $w$ in $\aut^{c}$.
    The set $S$ is initialized at the moment of the jump to the accepting phase and then it is updated by means of a subset construction where all accepting states are discharged from $S$;
    such discharged accepting states are collected in the other components $B$ and $C$ presented below for further analysis.
    This means that the runs that visit an accepting state while being tracked by $S$ are going to be truncated at the moment of the visit, since they are not safe to be used for accepting $w$ in $\aut^{c}$;
    the runs that survive are indeed safe, in particular when the jump happens after the stable level $\stableLevel$:
    in this case, all infinite runs are for sure safe since they will never visit accepting states anymore.
\item 
    The set $B \subseteq \dstates$, that follows a \emph{breakpoint} construction, is used to verify whether the guess we made to enter the accepting phase is correct; 
    moreover, the macrostates of $\aut^{c}$ having $B = \emptyset$ are accepting.
    
    The set $B$ is initialized, at the moment of the jump as well as every time it becomes empty, with the current states of the runs that need to be analyzed.
    Step by step, $B$ is purged of all runs that are safely tracked in $S$;
    if a run results to be not safe, it will be collected in the set $C$ below and its analysis postponed to the next reset of $B$.
    
    If a run of $\aut$ is accepting, it will eventually definitively leave $S$ to enter $B$ and stay in $B$ forever, thus making the run of $\aut^{c}$ not accepting.
    On the other hand, if an infinite run of $\aut$ is not accepting, it will eventually be trapped in $S$, so it will eventually stay out of $B$ forever.
    This means that if all infinite runs of $\aut$ over $w$ are not accepting, then $B$ becomes empty again and again, thus $\aut^{c}$ accepts $w$;
    if at least one run of $\aut$ over $w$ is accepting, then $B$ eventually tracks it forever, thus $\aut^{c}$ rejects $w$.
\item 
    The set $C \subseteq \dstates$ keeps all states that still need to be checked, like the ones just visited by a jump from $N$, the accepting ones discharged from $S$, or the successors of those that are currently to be checked, unless they are already managed by $S$ or $B$. 
    In practice, $C$ is a \emph{collector} for all states visited by runs we are not already analyzing.
    As soon as $B$ becomes empty, all these states in $C$ are transferred to $B$ to be analyzed and $C$ is reset, so to be ready to collect the new runs that need to be analyzed later. 
\end{itemize} 

The specialized complementation algorithm for LDBAs is formalized below.
\begin{definition}
\label{def:LDBAnsbcComplementation}
    Given an LDBA~$\aut = (\states, \initialStates, \trans, \acc)$, let $\nstates$ and $\dstates$ form a partition of $\states$ according to Definition~\ref{def:typesOfBAsbyTransitions}; 
    let $\ntrans$, $\jtrans$, and $\dtrans$ be the corresponding partition of $\trans$.
    We define the complement automaton $\aut^{c}= (\states^{c}, \initialStates^{c}, \trans^{c}, \acc^{c})$ as follows.
    \begin{itemize}
    \item 
        $\states^{c} \subseteq 2^{\states} \cup (2^{\nstates} \times 2^{\dstates} \times 2^{\dstates} \times 2^{\dstates})$ is the smallest set such that for $\initialStates^{c}$ and $\trans^{c}$ as defined below, we have $\initialStates^{c} \subseteq \states^{c}$ and $\trans^{c}(\states^{c}, a) \subseteq \states^{c}$ for each $a \in \alphabet$;
    \item 
        $\initialStates^{c} = \setnocond{\initialStates}$;
    \item 
        $\trans^{c} = \ntrans^{c} \cup \jtrans^{c} \cup \dtrans^{c}$ where
        \begin{itemize}
        \item 
            $\ntrans^{c} \colon 2^{\states} \times \alphabet \to 2^{2^{\states}}$ is such that for each $R \subseteq \states$ and $a \in \alphabet$, $\ntrans^{c}(R, a) = \setnocond{\trans(R, a)}$;
        \item 
            $\dtrans^{c} \colon (2^{\nstates} \times 2^{\dstates} \times 2^{\dstates} \times 2^{\dstates}) \times \alphabet \to (2^{\nstates} \times 2^{\dstates} \times 2^{\dstates} \times 2^{\dstates})$ is such that for each $N \subseteq \nstates$, $S, B, C \subseteq \dstates$, and $a \in \alphabet$, $\dtrans^{c}((N, S, B, C), a) = (N', S', B', C')$ where 
            \begin{itemize}
            \item 
                $N' = \ntrans(N, a)$,
            \item 
                $S' = \dtrans(S, a) \setminus \acc$,
            \item 
                if $B \neq \emptyset$, then 
                    $B' = \dtrans(B, a) \setminus S'$ and
                    $C' = ((\dtrans(C, a) \cup \jtrans(N, a) \cup (\dtrans(S, a) \cap \acc)) \setminus S') \setminus B'$;
            \item 

                if $B = \emptyset$, then 
                    $B' = (\dtrans(C, a) \cup \jtrans(N, a) \cup (\dtrans(S, a) \cap \acc)) \setminus S'$ and
                    $C' = \emptyset$;
            \end{itemize}
        \item 
            $\jtrans^{c} \colon 2^{\states} \times \alphabet \to 2^{2^{\nstates} \times 2^{\dstates} \times 2^{\dstates} \times 2^{\dstates}}$ is such that for each $R \subseteq \states$ and $a \in \alphabet$, $\jtrans^{c}(R, a) = \setnocond{\dtrans^{c}((R \cap \nstates, (R \cap \dstates) \setminus \acc, R \cap \acc, \emptyset), a)}$;
        \end{itemize}
    \item 
        $\acc^{c} = \setnocond{(N, S, \emptyset, C) \in \states^{c}}$.
    \end{itemize}
\end{definition}
Similarly to Definition~\ref{def:FANBAsliceBasedComplementation}, we might simplify the notation of $\ntrans^{c}$ and $\jtrans^{c}$ by dropping the curly brackets from the successor singleton sets.

\begin{figure}
    \centering
    \includegraphics{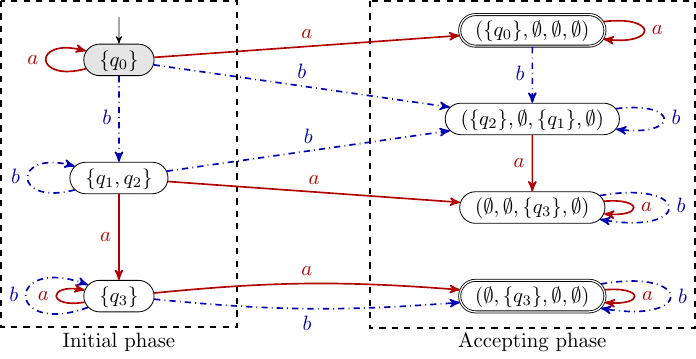}
    \caption{The complementary NBA~$\aut[N]^{c}$ of $\aut[N]$ depicted in Figure~\ref{fig:exampleBA} constructed by Definition~\ref{def:LDBAnsbcComplementation}.} 
    \label{fig:LDBAnsbcComplementation}
\end{figure}
\begin{markedexample}
\label{ex:LDBAcomplementation}
    As an example of the LDBA complementation given in Definition~\ref{def:LDBAnsbcComplementation}, consider again the LDBA~$\aut[N]$ depicted in Figure~\ref{fig:exampleBA};
    recall from Example~\ref{ex:typesOfBAs} that $\nstates = \setnocond{q_{0}, q_{2}}$ and $\dstates = \setnocond{q_{1}, q_{3}}$ form a partition of the states and consider the corresponding partition $\ntrans$, $\jtrans$, and $\dtrans$ of $\trans$.
    
    By Definition~\ref{def:LDBAnsbcComplementation}, $\aut[N]^{c}$ has $\setnocond{\setnocond{q_{0}}}$ as set of initial macrostates.
    From the macrostate $\setnocond{q_{0}} \in 2^{\states}$, $\ntrans$ is defined as:
    \begin{align*}
        \ntrans(\setnocond{q_{0}}, a) & = \setnocond{q_{0}} & 
        \ntrans(\setnocond{q_{0}}, b) & = \setnocond{q_{1}, q_{2}} \\ 
        \ntrans(\setnocond{q_{1}, q_{2}}, a) & = \setnocond{q_{3}} & 
        \ntrans(\setnocond{q_{1}, q_{2}}, b) & = \setnocond{q_{1}, q_{2}} \\
         \ntrans(\setnocond{q_{3}}, a) & = \setnocond{q_{3}} & \ntrans(\setnocond{q_{3}}, b) & = \setnocond{q_{3}}
    \end{align*}
    
    From these macrostates in the initial phase, for $c \in \alphabet$, $\jtrans$ jumps to the macrostates in the accepting phase by:
    \begin{align*}
        \jtrans(\setnocond{q_{0}}, c) & = \dtrans((\setnocond{q_{0}}, \emptyset, \emptyset, \emptyset), c), \\
        \jtrans(\setnocond{q_{1}, q_{2}}, c) & = \dtrans((\setnocond{q_{2}}, \emptyset, \setnocond{q_{1}}, \emptyset), c), \\
        \jtrans(\setnocond{q_{3}}, c) & = \dtrans((\emptyset, \setnocond{q_{3}}, \emptyset, \emptyset), c).
    \end{align*}
    
    Lastly, we move between macrostates in the accepting phase through $\dtrans$ as follows:
    \begin{align*}
        \dtrans((\setnocond{q_{0}}, \emptyset, \emptyset, \emptyset), a) & = (\setnocond{q_{0}}, \emptyset, \emptyset, \emptyset) & \dtrans((\setnocond{q_{0}}, \emptyset, \emptyset, \emptyset), b) & = (\setnocond{q_{2}}, \emptyset, \setnocond{q_{1}}, \emptyset) \\
        \dtrans((\setnocond{q_{2}}, \emptyset, \setnocond{q_{1}}, \emptyset), a) & = (\emptyset, \emptyset, \setnocond{q_{3}}, \emptyset) & \dtrans((\setnocond{q_{2}}, \emptyset, \setnocond{q_{1}}, \emptyset), b) & = (\setnocond{q_{2}}, \emptyset, \setnocond{q_{1}}, \emptyset) \\
        \dtrans((\emptyset, \setnocond{q_{3}}, \emptyset, \emptyset), a) & = (\emptyset, \setnocond{q_{3}}, \emptyset, \emptyset) & \dtrans((\emptyset, \setnocond{q_{3}}, \emptyset, \emptyset), b) & = (\emptyset, \setnocond{q_{3}}, \emptyset, \emptyset) \\
        \dtrans((\emptyset, \emptyset, \setnocond{q_{3}}, \emptyset), a) & = (\emptyset, \emptyset, \setnocond{q_{3}}, \emptyset) & \dtrans((\emptyset, \emptyset, \setnocond{q_{3}}, \emptyset), b) & = (\emptyset, \emptyset, \setnocond{q_{3}}, \emptyset)
    \end{align*}

    The resulting complementary NBA~$\aut[N]^{c}$ is shown in Figure~\ref{fig:LDBAnsbcComplementation}.
\end{markedexample}

\begin{restatable}[The Language and Size of $\aut^{c}$ for LDBA]{theorem}{LDBAlanguageSizeComplementNBA}
\label{thm:LDBAlanguageSizeComplementNBA}
    Given an LDBA~$\aut = (\states, \initialStates, \trans, \acc)$, let $\nstates$ and $\dstates$ form a partition of $\states$ such that $\acc \subseteq \dstates$; 
    let $\aut^{c}$ be the NBA constructed according to Definition~\ref{def:LDBAnsbcComplementation}.
    Then $\lang{\aut^{c}} = \infwords \setminus \lang{\aut}$ and $\aut^{c}$ has at most $2^{\size{\states}} + 2^{\size{\nstates}} \times 3^{\size{\acc}} \times 4^{\size{\dstates \setminus \acc}} \in \bigO(4^{\size{{\states}}})$ states. 
\end{restatable}
Before presenting the intuition of the proof for this theorem, we want to remark that there are other specialized complementation algorithms for LDBAs~\cite{Blahoudek16,DBLP:conf/pldi/ChenHLLTTZ18};
these algorithms differ from ours in that their constructed complementary NBAs are nondeterministic but \emph{not} limit deterministic, as our construction does.
In fact, the algorithms in~\cite{Blahoudek16,DBLP:conf/pldi/ChenHLLTTZ18} can be seen as optimized rank-based algorithms tailored for LDBAs, while ours is an optimized slice-based algorithm specialized for LDBAs.
Since the proof is rather long and involved, we provide it in~\ref{app:LDBAlanguageSizeComplementNBA};
here we give the main ideas it is based on.

The intuition behind the proof for $w \in \lang{\aut} \implies w \notin \lang{\aut^{c}}$ is as follows.
Given $w \in \lang{\aut}$, the macrorun of $\aut^{c}$ over $w$ either remains in the initial phase visiting only macrostates in $2^{\states}$, and thus is trivially non-accepting, or eventually jumps to the accepting phase where it visits only macrostates of the form $(N, S, B, C)$.
Since $w \in \lang{\aut}$, one run $\run$ of $\aut$ over $w$ is accepting; 
the states visited by $\run$ eventually leave $S$ and get trapped into the $B$ component, since $\run$ visits $\acc$ infinitely often, so the $B$ component becomes empty only finitely many times. 
Thus the macrorun of $\aut^{c}$ over $w$ visits only finitely many accepting macrostates, hence $w \notin \lang{\aut^{c}}$.

The intuition behind the proof for $w \notin \lang{\aut} \implies w \in \lang{\aut^{c}}$ is as follows.
Since $w \notin \lang{\aut}$, all runs of $\aut$ over $w$ are not accepting.
So the runs entering $\dstates$ will either eventually become safe or die out.
It is easy to see that the number of runs in $S$ or in $B$ (when nonempty) cannot increase.
Thus there will be some point after the stable level, where we have put all safe runs in $S$ since the other runs that merge with runs in $S$ will be cut off.
Therefore, all deterministic runs entering $C$ or $B$ have to be finite.
It follows that $B$ will become empty infinitely often, thus $\aut^{c}$ has an accepting macrorun over $w$, hence $w \in \lang{\aut^{c}}$.

\section{Related Work}
\label{sec:relatedWork}

Run DAGs were introduced in~\cite{kupferman2001weak} and codeterministic run DAGs were first used in~\cite{DBLP:journals/corr/abs-1110-6183} without an explicit name.
In~\cite{DBLP:journals/corr/abs-1110-6183}, Fogarty and Vardi exploited codeterministic run DAGs to complement \emph{reverse deterministic} \buchi automata with \rkc and the Ramsey-based algorithm, while we consider \rkc and \slc in this work.
Recently, codeterministic run DAGs have been also applied~\cite{DBLP:conf/fm/LiTTVZ21} to the Ramsey-based complementation construction, that is, however, not the focus of this work.
In a reverse deterministic \buchi automaton, each state has only one predecessor for each letter, for which all run DAGs are already codeterministic, as explained in Section~\ref{ssec:FANBArankBasedComplementation}, while the run DAGs of FANBAs may not be codeterministic without our construction described in Section~\ref{sec:reduced-run-dag}.

Later, codeterministic run DAGs were constructed in~\cite{rabinovich18} under the name of \emph{narrow forest} for complementing FANBAs with the \slc construction only. 
Here we present it as codeterministic run DAGs to serve as a unified tool for explaining concepts in both \rkc and \slc constructions. 
A subtle difference between the construction of codeterministic run DAGs in~\cite{rabinovich18} and ours is the following. 
To construct a codeterministic run DAG over $w \in \infwords$, Rabinovich~\cite{rabinovich18} makes use of a transducer $\mathcal{T}$ that chooses one predecessor for each vertex at the current level, while our construction uses a transition function to make the sets of successors of each pair of vertices at the current level disjoint with each other, as given in Definition~\ref{def:edge-relation-e}.

More significantly, for complementation, we applied codeterministic run DAGs to \emph{both} \rkc~\cite{kupferman2001weak} and \slc as presented in~\cite{DBLP:conf/birthday/VardiW08}.
(The complementation construction proposed in~\cite{rabinovich18} is a variant of \slc as introduced in~\cite{kahler2008complementation}.)
The comparison of the construction in~\cite{rabinovich18} and our improvement over \slc is as follows.
First, the complementary NBA constructed in~\cite{rabinovich18} is a UNBA with at most $\bigO(5^{n})$ states;
this complementary NBA is the product automaton of the transducer $\mathcal{T}$, a \buchi automaton $\mathcal{C}$ for expressing unambiguity and a \buchi automaton $\aut[D]$ for accepting all possible ways to construct codeterministic DAGs over $w \notin \lang{\aut}$.
Our complementary NBA is not required to be a UNBA, since we are interested in complementation for containment checking.
Thus, the bound of $\bigO(5^{n})$ given in~\cite{rabinovich18} is exponentially larger than the bound of $\bigO(4^{n})$ we achieve in this work.
Indeed, the product automaton of $\aut[T]$ and $\aut[D]$ in~\cite{rabinovich18} does yield a complementary NBA with $\bigO(4^{n})$ states, but this construction and complexity were not explicitly given in~\cite{rabinovich18}.

Second, the construction in~\cite{rabinovich18} and our \slc-based construction are both based on reduced DAGs in which each vertex has at most one predecessor. 
These two constructions, however, are technically different and have different emphases. 
The one developed in~\cite{rabinovich18} aims at building a complementary NBA~$\aut^{c}$ that is unambiguous, based on building the product of three automata, in which each automaton fulfills part of the desired functionality for $\aut^{c}$. 
For instance, $\mathcal{C}$ takes care of unambiguity while $\mathcal{D}$ obtains the complementary language. 
Instead, our focus is on a complementation construction for containment checking. 
In contrast to building the product automata, the construction we present in Section~\ref{ssec:FANBAsliceBasedComplementation} takes a tuple of sets of states of $\aut$ as a macrostate in the complementary automaton $\aut^{c}$ and works directly on those tuples for computing successors on-the-fly, following the idea of the NCSB complementation for limit deterministic \buchi automata used in~\cite{Blahoudek16}, where various subsumption relations have been proposed for this representation of macrostates in the NCSB complementation;
they help to reduce the number of macrostates in $\aut^{c}$, even improving termination analysis of programs. 
Inspired by~\cite{DBLP:conf/pldi/ChenHLLTTZ18}, we can also define a subsumption relation between macrostates in $\aut^{c}$ (see Proposition~\ref{prop:sumbsumptionRelationBetweenMacrostatesSLCforFANBAs}) by our construction, which can be used to improve the containment checking between an NBA and an (FA)NBA and to reduce the number of macrostates in $\aut^{c}$.
More optimizations for the NCSB algorithm for complementing LDBAs can be found in~\cite{DBLP:conf/cav/HavlenaLS22}.

Specialized complementation algorithms for LDBAs have been proposed~\cite{DBLP:conf/tacas/EsparzaKRS17,DBLP:conf/cav/LiTFVZ22}, based on the determinization of LDBAs;
however, the complexity of these methods is $\bigO(n!)$, much higher than $\bigO(4^{n})$ obtained by this work.
It is shown in \cite{DBLP:conf/concur/HahnLST015,DBLP:conf/cav/SickertEJK16} that LDBAs can be used for quantitative model checking of probabilistic systems, which is, however, not the scope of this work.

We believe that recent optimizations for \rkc constructions, including the tight level rankings~\cite{DBLP:journals/ijfcs/FriedgutKV06,Schewe09} and super-tight runs~\cite{DBLP:conf/concur/HavlenaL21}, are compatible with our specialized \rkc algorithm for FANBAs since these optimizations are all in line with the level rankings given in Definitions~\ref{def:levelRankingFunction} and~\ref{def:coverage-level-rnk}.
Note that we only show how to organize the runs over a word so that their maximal rank becomes $2$, without altering those definitions.
With tight level rankings (thus also super-tight runs), we can omit macrostates with states assigned rank $2$ since the maximal rank of a tight level ranking must be odd;
so, the upper bound of the specialized \rkc construction for FANBAs can further be improved to $\bigO(4^{n})$ since nonaccepting states can either (1) be not reached at the moment, (2) have rank $1$, (3) have rank $0$ and be not present in the breakpoint or (4) have rank $0$ and be present in the breakpoint.
(The situation for accepting states are similar except that they cannot have rank $1$.)
This construction, like our algorithm given in Definition~\ref{def:FANBAsliceBasedComplementation}, needs to guess when to jump from the initial phase to the accepting phase where the level rankings are tight.
From a practical point of view, one may also improve our construction with the optimizations proposed in~\cite{DBLP:conf/tacas/HavlenaLS22} by handling with specific methods the different types of strongly connected components.
However, we believe such optimizations cannot further improve the theoretical upper bound.

In literature there are also many works concerning about the restriction of nondeterminism in the automata.
For instance, both~\cite{DBLP:conf/dlt/MichalewskiS16} and~\cite{DBLP:journals/dagstuhlreports/ColcombetQS21} are excellent works showing how unambiguity plays a role in automata that accept infinite trees;
the latter also considers automata over finite words.
Both, however, do not consider automata on infinite words, the focus of our paper.
When restricting the nondeterminism in automata, one can also identify new subclasses of automata, namely, \emph{Good-For-Game} (GFG) automata~\cite{DBLP:conf/csl/HenzingerP06} and \emph{Good-For-MDP} (GFM) automata~\cite{DBLP:conf/tacas/HahnPSS0W20}.
GFG and GFM automata are the classes of nondeterministic automata that can be used in the context of games and the verification of Markov decision processes, respectively.
GFG \buchi automata have the same expressiveness as deterministic \buchi automata~\cite{DBLP:journals/apal/KupfermanSV06,DBLP:conf/stacs/NiwinskiW98};
this means that one needs more general conditions, such as Rabin and parity, for GFG automata to recognize the full class of $\omega$-regular languages~\cite{DBLP:journals/apal/KupfermanSV06,DBLP:conf/stacs/NiwinskiW98}.
Here we only focus on the \buchi condition.
It is proved in~\cite{DBLP:conf/icalp/KuperbergS15} that one can construct an equivalent deterministic \buchi automaton with only a quadratic blowup of states for a given GFG \buchi automaton.
(The complexity of the construction itself is shown to be in NP.)
It follows that given a GFG \buchi automaton, we can obtain a complementary \buchi automaton with only quadratic blowup of states since deterministic \buchi automata can be complemented in linear time.
The state complexity of the complementary automata is much lower than that we obtained for FANBAs.
Nonetheless, FANBAs can represent the whole class of $\omega$-regular languages, which is beyond the expressive capacity of GFG and deterministic \buchi automata. 
In~\cite{DBLP:conf/tacas/HahnPSS0W20}, the authors showed how to construct from a given \buchi automaton a GFM \buchi automaton with a branching degree of $2$ that may neither be limit deterministic nor GFG.
Based on personal communications with the authors of~\cite{DBLP:journals/corr/abs-2202-07629}, we are convinced that FANBAs and GFM \buchi automata are incomparable, that is, we can find \buchi automata in each part of the symmetric difference between the classes of FANBAs and of GFM \buchi automata.
Therefore, the specialized complementation algorithms for GFG \buchi automata, LDBAs, or FANBAs may not always be suitable to GFM \buchi automata.
To the best of our knowledge, we are not aware of any specialized complementation algorithm for GFM \buchi automata; 
we believe this is an interesting future work.

\section{Conclusion and Future Work}
\label{sec:conclusion}

This work exploits codeterministic run DAGs over infinite words as a unified tool to optimize both \rkc and \slc constructions.
Consequently, we have improved the complexity of the classical \rkc and \slc constructions for FANBAs, respectively, to ${\bigO(6^n)}$ from $2^{\bigO(n \log n)}$ and to $\bigO(4^{n})$ from $\bigO((3n)^{n})$, based on codeterministic DAGs.
As a further contribution, we view the \slc algorithm explicitly as the construction of codeterministic DAGs and a specialized complementation algorithm for FANBAs. 
We then provide a subsumption relation between states in the complementary NBAs of FANBAs in hope of improving the containment checking between an NBA and an (FA)NBA.
Our work proposes the construction of codeterministic DAGs as a way to obtain finite ambiguity, the key ingredient of \buchi complementation.
As an application, we apply codeterministic DAGs to the complementation of LDBAs, demonstrating the generality of the power of finite ambiguity. 

As future work, we plan to study whether $\bigO(4^{n})$ is also the lower bound for the complementation of FANBAs.
We also plan to evaluate empirically the different complementation algorithms to check how much their theoretical complexity is reflected in practice; 
the implementations however depend also on software engineering aspects and heuristics (like, automata representation, data structures, order of state exploration) which can affect their actual performance.
One may improve the practical performance of existing complementation constructions by first applying minimization techniques~\cite{DBLP:journals/lmcs/ClementeM19,DBLP:conf/popl/MayrC13} and then complementing the minimized automata.
We plan to study how our specialized construction for FANBAs compares to this approach in practice.
Also an empirical evaluation on how the subsumption relation between macrostates proposed in Proposition~\ref{prop:sumbsumptionRelationBetweenMacrostatesSLCforFANBAs} will benefit the containment checking problem is worthy of exploring.
Another line of future work is to study determinization constructions for FANBAs.
Finally, it is possible to use our work to improve the program-termination checking framework proposed in~\cite{DBLP:conf/cav/HeizmannHP14} if one generalizes a terminating path to an FANBA.

\subsubsection*{Acknowledgment}
We thank the anonymous reviewers for their useful remarks that
helped us improve the quality of the paper and Qiyi Tang for sharing insights about GFG automata.
Work supported in part by 
the National Natural Science Foundation of China (grants no.\@ 62102407 and 61836005);
the Strategic Priority Research Program of the Chinese Academy of Sciences (grant no.\@ XDA0320000);
the CAS Project for Young Scientists in Basic Research (grant no.\@ YSBR-040);
the Engineering and Physical Sciences Research Council (grant no.\@ EP/X021513/1);
the Guangdong Science and Technology Department (Grant No. 2018B010107004);
NSF grants IIS-1527668, CCF-1704883, and IIS-1830549; 
and 
an award from the Maryland Procurement Office.

\noindent
\protect\includegraphics[height=8pt]{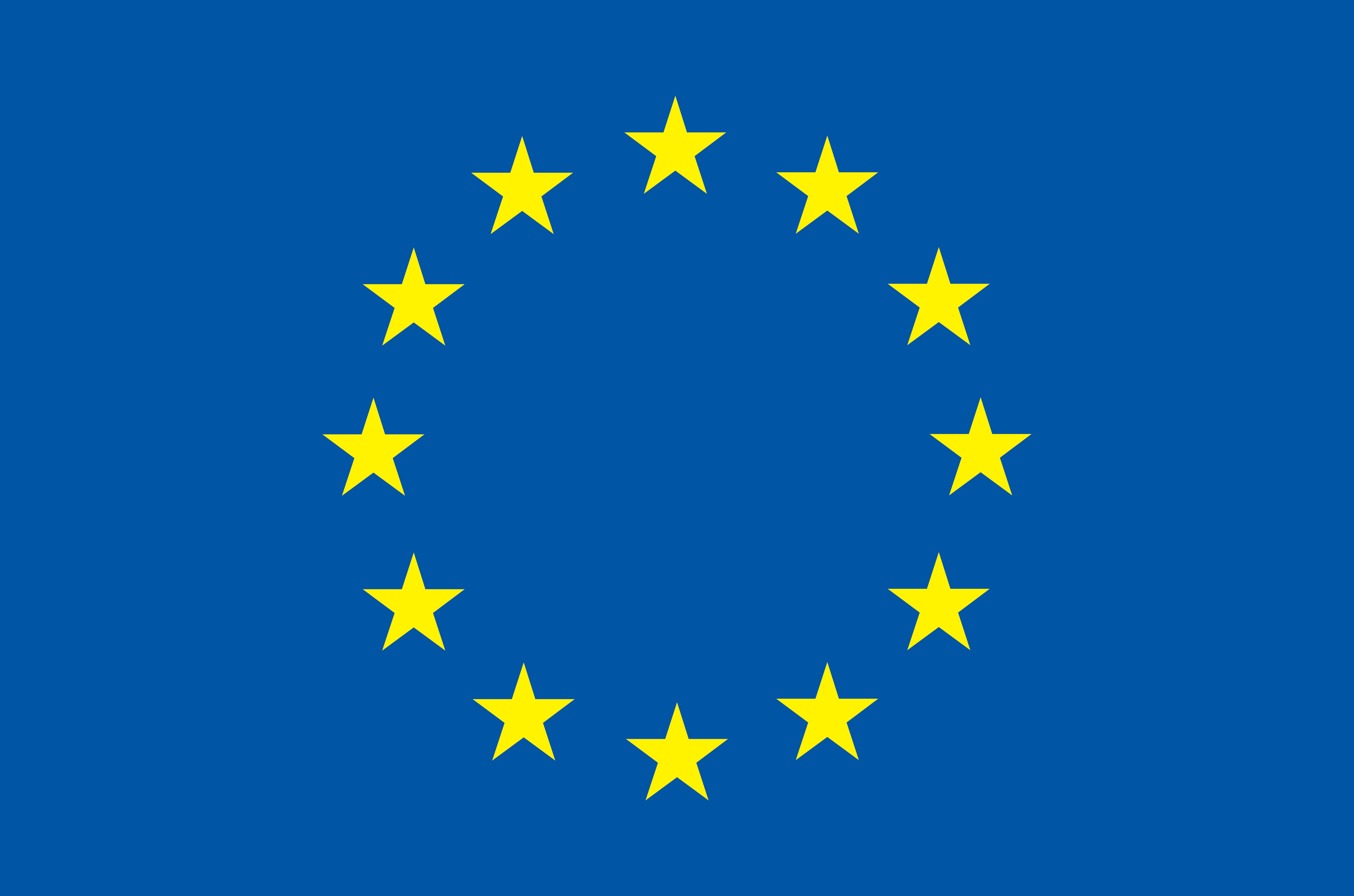} This work is part of the European Union’s Horizon 2020 research and innovation programme under the Marie Sk\l{}odowska-Curie grant no.\@ 101008233.

\bibliographystyle{elsarticle-num}
\bibliography{generic}

\newpage
\appendix

\section{Proofs of Several Results Given in the Main Part}
\label{app:otherProofs}

\separatingLevelDAG*
\begin{proof}
    Since $\aut$ is an FANBA, there are only finitely many accepting $\omega$-branches in $\dagAW{\aut}{w}$.
    Therefore, an accepting $\omega$-branch in $\dagAW{\aut}{w}$ only merges with other (accepting) $\omega$-branches for finitely many times. 
    It follows that given an accepting $\omega$-branch $\hat{\run}$ in $\dagAW{\aut}{w}$, there must exist a separating level $h \geq 1$ such that each vertex $\wordletter{\hat{\run}}{i}$ with $i \geq h$ has exactly one predecessor. 
    Otherwise, there will be infinitely many accepting branches, contradicting the assumption that $\aut$ is an FANBA. 
    Assume that there are $m < \infty$ accepting $\omega$-branches in $\dagAW{\aut}{w}$. 
    Then we can set the separating level $\separatingLevel$ of $\dagAW{\aut}{w}$ to $\max\setcond{h_{i}}{1 \leq i \leq m}$ where $h_{i}$ is the separating level index of $i$-th accepting $\omega$-branch. 
\end{proof}

\FANBAacceptanceCodeterministicDAG*
\begin{proof}
    Instead of $\reduceddagAW{\aut}{w}$, consider the DAG~$\dagAW{\aut}{w}$: 
    the proof is trivial when $\dagAW{\aut}{w}$ is non-accepting, since by definition there are only finitely many accepting vertices in $\dagAW{\aut}{w}$ occurring in each branch; 
    after reducing the edges, there still can only be finitely many accepting vertices occurring in each branch, so $\reduceddagAW{\aut}{w}$ is also non-accepting.
    
    Assume now that $\dagAW{\aut}{w}$ is accepting.
    Let $\hat{\run}$ be an accepting $\omega$-branch and $\separatingLevel \in \naturals$ be the separating level given in Lemma~\ref{lem:separatingLevelDAG}.
    Since $\separatingLevel$ is finite and there are infinitely many accepting vertices in $\hat{\run}$, the $\omega$-branch from $\wordletter{\hat{\run}}{\separatingLevel + 1}$ must be accepting.
    Moreover, $\wordletter{\hat{\run}}{\separatingLevel + 1}$ is reachable from some initial vertex $\vertex{q}{0}$ with $q \in \initialStates$.
    Then there must exist an accepting $\omega$-branch in $\reduceddagAW{\aut}{w}$: 
    by reducing the edges, $\wordletter{\hat{\run}}{\separatingLevel + 1}$ still remains reachable from some (possibly different) initial vertex $\vertex{q'}{0}$ with $q' \in \initialStates$;
    the $\omega$-branch from $\wordletter{\hat{\run}}{\separatingLevel + 1}$ remains unchanged since by definition of separating level, each vertex in such an $\omega$-branch has only one predecessor.
    This allows us to conclude that $w$ is accepted by $\aut$ if and only if $\reduceddagAW{\aut}{w}$ is accepting.
\end{proof}

\FANBAindependentReducedTrans*
\begin{proof}
    By the lemma assumptions, we know that the two DAGs $\reduceddagAW{\aut}{w_{1}}$ and $\reduceddagAW{\aut}{w_{2}}$ are based on the same NBA~$\aut$. 
    We also know that $\reduceddagAW{\aut}{w_{1}}$ and $\reduceddagAW{\aut}{w_{2}}$ share the same set of states $S$ at level $\ell_{1}$ and at level $\ell_{2}$, respectively.
    Similarly, we know that $\wordletter{w_{1}}{\ell_{1}} = \wordletter{w_{2}}{\ell_{2}} = b$.
    By Definition~\ref{def:edge-relation-e}, it follows that the set of successors $S'$ is the same for both DAGs, and so is $S_{\mathit{min}}$.
    It is now immediate to see that for each set of states $S'' \subseteq S$, we have $\cotrans_{w_{1}, \ell_{1}}(S'', \wordletter{w_{1}}{\ell_{1}}) = \trans(S'' \cap S_{\mathit{min}}, \wordletter{w_{1}}{\ell_{1}}) = \trans(S'' \cap S_{\mathit{min}}, b) = \trans(S'' \cap S_{\mathit{min}}, \wordletter{w_{2}}{\ell_{2}}) = \cotrans_{w_{2}, \ell_{2}}(S'', \wordletter{w_{2}}{\ell_{2}})$, that is, $\cotrans_{w_{1}, \ell_{1}} = \cotrans_{w_{2}, \ell_{2}}$, as required.
\end{proof}

\FANBAfiniteOmegaBranchesInCodeterministicDAG*
\begin{proof}
    Let $m_{i}$ with $i \geq 0$ be the number of vertices at level $i$ which occur in some $\omega$-branch.
    (For instance, for the DAG~$\dagAW{\aut}{b^{\omega}}$ shown in Figure~\ref{fig:example2ANBA}, once reduced to $\reduceddagAW{\aut}{b^{\omega}}$ by removing the dashed edges, we have $m_{i} = 1$ for each $i \in \naturals$:
    for all levels in $\naturals \setminus \setnocond{1}$ this is obvious; 
    for level $1$, we have the two vertices $\vertex{q_{1}}{1}$ and $\vertex{q_{2}}{1}$ but only $\vertex{q_{1}}{1}$ occurs in some $\omega$-branch.)
    Since each vertex in $\reduceddagAW{\aut}{w}$ has only one predecessor but it can have several successors, we have that $m_{0} \leq m_{1} \leq m_{2} \leq \cdots$, i.e., the number of vertices in $\omega$-branches on each level does not decrease over the levels.
    If it would decrease, then this means that there are at least two $\omega$-branches $\branch_{1}$ and $\branch_{2}$ such that $\wordletter{\branch_{1}}{\ell + 1} = \wordletter{\branch_{2}}{\ell + 1}$, that is, they share the vertex $\vertex{q}{\ell + 1}$ for some $q \in \states$ and $\ell \in \naturals$, but $\wordletter{\branch_{1}}{\ell} \neq \wordletter{\branch_{2}}{\ell}$. 
    This means that $\vertex{q}{\ell + 1}$ has two different predecessors, namely $\wordletter{\branch_{1}}{\ell}$ and  $\wordletter{\branch_{2}}{\ell}$, but this contradicts the fact that $\reduceddagAW{\aut}{w}$ is codeterministic.
    Moreover, since there are at most $\size{\states}$ states on each level, the number $b_{\ell}$ of $\omega$-branches at each level $\ell$ cannot exceed $\size{\states}$: 
    in order to cross this threshold, there must be a level $k$ such that $b_{k} \leq \size{\states}$ and $b_{k + 1} > \size{\states}$.
    By the Pigeonhole principle, this implies that two different branches must share a common vertex $v$ at level $k + 1$, since there are $b_{k + 1} > \size{\states}$ $\omega$-branches and at most $\size{\states}$ vertices.
    This means that $v$ has two predecessors, which contradicts the fact that $\reduceddagAW{\aut}{w}$ is codeterministic.
\end{proof}

\FANBAwordRejectedStableLevelInCodeterministicDAG*
\begin{proof}
    We prove the two directions independently.
    Assume that $w \notin \lang{\aut}$.
    By Lemma~\ref{lem:FANBAfiniteOmegaBranchesInCodeterministicDAG}, let $m \leq \size{\states}$ be the number of $\omega$-branches in $\reduceddagAW{\aut}{w}$.
    Since $w \notin \lang{\aut}$ by hypothesis, each $\omega$-branch in $\reduceddagAW{\aut}{w}$ is not accepting.
    Therefore, for the $i$-th $\omega$-branch $\hat{\run}_{i}$, there is a vertex $\vertex{q}{\stableLevel_{i}}$ such that every vertex of $\hat{\run}_{i}$ reachable from $\vertex{q}{\stableLevel_{i}}$ is not an $\acc$-vertex.  
    It follows that we can set $\stableLevel = \max\setcond{\stableLevel_{i}}{i \in \irange{m}}$ and thus all the $\acc$-vertices on a level $l \geq \stableLevel$ are finite and not on $\omega$-branches.

    For the other direction, we prove the contrapositive of the claim, that is, if $w \in \lang{\aut}$, then there does not exists a stable level $\stableLevel > 0$ in $\reduceddagAW{\aut}{w}$, so assume that $w \in \lang{\aut}$.
    Since $\aut$ is an FANBA, Lemma~\ref{lem:FANBAacceptanceCodeterministicDAG} implies that there exists an accepting $\omega$-branch $\branch$ in $\reduceddagAW{\aut}{w}$.
    By definition, this implies that there are infinitely many $\acc$-vertices in $\branch$; 
    by the definition of finite $\acc$-vertex, it follows that there does not exist a stable level $\stableLevel$ in $\reduceddagAW{\aut}{w}$ such that each $\acc$-vertex after $\stableLevel$ is finite, as required.
\end{proof}

\maximumRankOfCodeterministicDAGs*
\begin{proof}
    Given $w \notin \lang{\aut}$, our goal is to prove that starting from $\ireduceddagAW{0}{\aut}{w} = \reduceddagAW{\aut}{w}$, we get an empty $\ireduceddagAW{3}{\aut}{w}$.
    By Lemma~\ref{lem:FANBAwordRejectedStableLevelInCodeterministicDAG}, there exists a stable level $\stableLevel > 1$ such that on each level $l \geq \stableLevel$, the $\acc$-vertices are finite.
    Therefore, $\ireduceddagAW{1}{\aut}{w}$ contains only non-$\acc$-vertices after level $\stableLevel$.
    It follows that $\ireduceddagAW{2}{\aut}{w}$ removes all vertices after level $\stableLevel$.
    Thus if $\ireduceddagAW{2}{\aut}{w}$ is not already empty, it contains only finite vertices, which are going to be removed in $\ireduceddagAW{3}{\aut}{w}$.
    We then conclude that $\ireduceddagAW{3}{\aut}{w}$ is empty, as required.
    
    Suppose now that $\ireduceddagAW{3}{\aut}{w}$ is empty: 
    by construction, it means that after removing the original finite vertices, the resulting $\acc$-free vertices, and the newly become finite vertices, there are no vertices left.
    This implies that in $\ireduceddagAW{0}{\aut}{w} = \reduceddagAW{\aut}{w}$ there was no accepting $\omega$-branch, otherwise the vertices on such an $\omega$-branch would not have been removed by the pruning.
    Since $\ireduceddagAW{0}{\aut}{w} = \reduceddagAW{\aut}{w}$ is not accepting, by Lemma~\ref{lem:FANBAacceptanceCodeterministicDAG} we have that $w \notin \lang{\aut}$, as required.
\end{proof}

\identificationOfNonacceptingRunDAGsForFANBAs*
\begin{proof}
    Assume that $w$ is not accepted by $\aut$.
    If $\aut$ rejects the word $w$, all $\omega$-branches have suffixes that are $\acc$-free.
    By Lemma~\ref{lem:FANBAacceptanceCodeterministicDAG} and Definition~\ref{def:rankingFunctionForreducedDAGrejectedWord}, the $\omega$-branches of $\reduceddagAW{\aut}{w}$ eventually get trapped in rank $1$, which is an odd rank.
    
    Assume that all $\omega$-branches of $\reduceddagAW{\aut}{w}$ eventually get trapped in odd ranks.
    By Definition~\ref{def:rankingFunctionForreducedDAGrejectedWord}, it follows that $\ireduceddagAW{3}{\aut}{w}$ must be empty, otherwise the vertices in $\ireduceddagAW{3}{\aut}{w}$ are finite and should have been removed from $\ireduceddagAW{2}{\aut}{w}$.
    Therefore, by Lemma~\ref{lem:FANBAacceptanceCodeterministicDAG}, $\aut$ rejects $w$.
\end{proof}

\sumbsumptionRelationBetweenMacrostatesSLCforFANBAs*
\begin{proof}
    Let $w \in \infwords$ be an $\omega$-word and $\run$, $\run'$ be the macroruns of $\autNewInitial{\aut^{c}}{m}$ and $\autNewInitial{\aut^{c}}{m'}$ over $w$, respectively.
    In particular, we have that $\run = (N_{0} = N, C_{0} = C, B_{0} = B) (N_{1}, C_{1}, B_{1}) \cdots (N_{k}, C_{k}, B_{k}) \cdots$ and, similarly, that $\run' = (N'_{0} = N', C'_{0} = C', B'_{0} = B') (N'_{1}, C'_{1}, B'_{1}) \cdots (N'_{k}, C'_{k}, B'_{k}) \cdots$.
    Recall that $\aut^{c}$ is limit deterministic (cf.\@ Remark~\ref{rem:slcGivesLDBAs}), with the macrostates $(N, C, B)$ belonging to the deterministic part of $\aut^{c}$;
    this means that the macroruns $\run$ and $\run'$ are unique.
    
    Assume that $w \in \lang{\autNewInitial{\aut^{c}}{m'}}$; 
    this implies that $\run'$ is accepting, that is, $\run'$ visits $\acc^{c}$ infinitely often.
    By Definition~\ref{def:FANBAsliceBasedComplementation}, this happens if there are infinitely many empty $B'$-sets in $\run'$.
    This means that the level $0$ in the codeterministic DAG~$\reduceddagAW{\autNewInitial{\aut}{N'}}{w}$ of $\autNewInitial{\aut}{N'}$ over $w$ is a stable level, i.e., each $\acc$-vertex in $\reduceddagAW{\autNewInitial{\aut}{N'}}{w}$ is finite.
    The motivation for this is that by Definition~\ref{def:FANBAsliceBasedComplementation}, each branch from an $\acc$-vertex in $\reduceddagAW{\autNewInitial{\aut}{N'}}{w}$ will eventually be put in a $B'$-set; 
    if one of such branches is not finite, then the $B'$-set will become empty only finitely many times, contradicting  the assumption that $w \in \lang{\autNewInitial{\aut^{c}}{m'}}$.
    By the construction of codeterministic DAGs given in Section~\ref{sec:reduced-run-dag}, the codeterministic DAG~$\reduceddagAW{\autNewInitial{\aut}{N}}{w}$ of $\autNewInitial{\aut}{N}$ over $w$ is identical to $\reduceddagAW{\aut{}{}v[N']}{w}$ since $N = N'$.
    Consequently, level $0$ is also a stable level in $\reduceddagAW{\autNewInitial{\aut}{N}}{w}$, that is, each $\acc$-vertex in $\reduceddagAW{\autNewInitial{\aut}{N}}{w}$ is also finite. 
    Since the $B'$-sets in $\run'$ become empty and are reset to $C'$ infinitely many times, we have that all branches from $C'$ are finite; 
    this implies that also all branches from $B \subseteq C$ are finite, given that $C \subseteq C'$ by hypothesis.
    From this it follows that there exists a least $j \in \naturals$ such that $B_{j} = \emptyset$.
    Since all branches in the $C$-set (including new branches coming from the $N$-set) are finite, there are infinitely many $k \geq j$ such that $B_{k} = \emptyset$ in $\run$.
    This means that $\run$ is also accepting, thus $w \in \lang{\autNewInitial{\aut^{c}}{m}}$; 
    this concludes the proof that $\lang{\autNewInitial{\aut^{c}}{m'}} \subseteq \lang{\autNewInitial{\aut^{c}}{m}} $.
\end{proof}

\section{Proof of Theorem~\ref{thm:FANBAlanguageSizeRankBasedComplement}}
\label{app:FANBAlanguageSizeRankBasedComplement}

\FANBAlanguageSizeRankBasedComplement*

We first prove the result about the language of $\aut^{c}$.
Consider the inclusion $\lang{\aut^{c}} \subseteq \infwords \setminus \lang{\aut}$ and let $w \in \lang{\aut}$, that is, $w \notin \infwords \setminus \lang{\aut}$;
we are going to prove that $w \notin \lang{\aut^{c}}$.
By Lemma~\ref{lem:FANBAacceptanceCodeterministicDAG}, it follows that the run DAG~$\reduceddagAW{\aut}{w}$, based on the unique reduced transition function $\cotrans$ (cf.\@ Corollary~\ref{cor:FANBAuniqueReducedTransitionFunctionDAG}), is codeterministic and accepting.
Therefore, the maximum rank required for identifying whether $\reduceddagAW{\aut}{w}$ is accepting is $2$, according to Lemmas~\ref{lem:maximumRankOfCodeterministicDAGs} and~\ref{lem:identificationOfNonacceptingRunDAGsForFANBAs}.
By Definition~\ref{def:coverage-level-rnk}, the coverage relation between two consecutive level ranking functions $\levelRanking$ and $\levelRanking'$ in Definition~\ref{def:FANBArankBasedComplementation} induces a valid ranking function for $\reduceddagAW{\aut}{w}$, where $\acc$-states only get even ranks.
Since $\reduceddagAW{\aut}{w}$ is accepting, there is an $\omega$-run $\run$ of $\aut$ over $w$ that has infinitely many states with even ranks.
In fact, for every valid ranking function guessed for $\reduceddagAW{\aut}{w}$, this run will eventually get trapped in even ranks: 
if the rank of a state in $\run$ would be changed to an odd rank, then the rank will decrease at least by $1$ to become an even rank once $\run$ visits the next accepting state, according to Definition~\ref{def:DAGlevelRanking}.
Since the rank is bounded and there are infinitely many accepting states in $\run$, the rank of $\run$ will finally get trapped in an even rank.
Moreover, we have that the states of $\run$ will eventually be put in the set $O$ as soon as $O$ becomes empty.
If $O$ will never become empty, then $w$ is not accepted by $\aut^{c}$ by definition;
otherwise, the run $\run$ goes into the $O$ set and stays there forever by Definition~\ref{def:FANBArankBasedComplementation}, which also makes $O$ nonempty forever.
Thus $w$ is not accepted by $\aut^{c}$ when $w \in \lang{\aut}$.
It follows that $\lang{\aut^{c}} \subseteq \infwords \setminus \lang{\aut}$.

Consider now the other inclusion, namely $\infwords \setminus \lang{\aut} \subseteq \lang{\aut^{c}}$;
let $w \in \infwords \setminus \lang{\aut}$, i.e., $w \notin \lang{\aut}$.
According to Definition~\ref{def:rankingFunctionForreducedDAGrejectedWord} and Lemma~\ref{lem:identificationOfNonacceptingRunDAGsForFANBAs}, we can construct a unique classical ranking function for each rejecting codeterministic DAG~$\reduceddagAW{\aut}{w}$ of $\aut$ over $w$.
The maximum rank of such classical ranking functions is at most $2$ (cf.\@ Definition~\ref{def:rankingFunctionForreducedDAGrejectedWord}).
\rkc will nondeterministically guess rankings of $\reduceddagAW{\aut}{w}$, which is reflected in the definition of $\trans^{c}$ in Definition~\ref{def:FANBArankBasedComplementation} where given a state $(\levelRanking, O) \in \states^{c}$, all possible $(\levelRanking', O')$ with $\levelRanking' \coveredBy{\cotrans}{a} \levelRanking$ are successors of $(\levelRanking, O)$.
Thus there must be a guess of such a unique classical ranking function.
This means that all $\omega$-branches in $\reduceddagAW{\aut}{w}$ will eventually get trapped in odd ranks, i.e., all $\omega$-branches will get trapped in odd ranks and the descendants of states with even ranks will be finite after the stable level.
So the breakpoint set $O$ with states in even ranks will become empty infinitely often since the states in $O$ will all disappear eventually.
It follows that $\reduceddagAW{\aut}{w}$ must be accepting in $\aut^{c}$, i.e., $\infwords \setminus \lang{\aut} \subseteq \lang{\aut^{c}}$. 
Thus it holds that $\lang{\aut^{c}} = \infwords \setminus \lang{\aut}$, as required.

Regarding the result relative to the size of $\aut^{c}$, according to Definition~\ref{def:levelRankingFunction} and~\ref{def:FANBArankBasedComplementation} we have that some values cannot be used for all states in $\states$ and some state is excluded from belonging to $O$; 
this means that we need to analyze together $\levelRanking$ and $O$.
For each state $q \in \states$, we have the following cases:
\begin{enumerate}
\item
    $\levelRanking(q) = 0$ and $q \in O$;
\item
    $\levelRanking(q) = 0$ and $q \notin O$;
\item
    $\levelRanking(q) = 2$ and $q \in O$;
\item
    $\levelRanking(q) = 2$ and $q \notin O$;
\item
    $\levelRanking(q) = 1$; or
\item
    $\levelRanking(q) = \bot$.
\end{enumerate}
For the last two cases, by definition we have that $q \notin O$,
thus the number of states is at most $\bigO(6^{n})$.

\section{Proof of Theorem~\ref{thm:FANBAlanguageSizeSliceBasedComplement}}
\label{app:FANBAlanguageSizeSliceBasedComplement}

\FANBAlanguageSizeSliceBasedComplement*

Theorem~\ref{thm:FANBAlanguageSizeSliceBasedComplement} is a direct consequence of the proofs we present in the following subsections, namely \ref{app:FANBAlanguageSizeSliceBasedComplement:languageWinFANBAnotInComplement} that shows that whenever $w \in \lang{\aut}$, then we have $w \notin \lang{\aut^{c}}$, i.e., $w \in \infwords \setminus \lang{\aut^{c}}$;
\ref{app:FANBAlanguageSizeSliceBasedComplement:languageWnotinFANBAinComplement} for the dual property $w \notin \lang{\aut} \implies w \in \lang{\aut^{c}}$, that is, $w \notin \lang{\aut} \implies w \notin \infwords \setminus \lang{\aut^{c}}$; 
and \ref{app:FANBAlanguageSizeSliceBasedComplement:size} taking care of establishing the size of $\aut^{c}$.

\subsection{Proof of \texorpdfstring{$w \in \lang{\aut} \implies w \notin \lang{\aut^{c}}$}{w accepted by A implies w not accepted by the complement of A}}
\label{app:FANBAlanguageSizeSliceBasedComplement:languageWinFANBAnotInComplement}

We first show that $w \in \lang{\aut}$ implies $w \notin \lang{\aut^{c}}$. 
Given $w \in \lang{\aut}$, let $\run$ be one of the finitely many accepting runs of $\aut$ over $w$ and $\run'$ be a macrorun of $\aut^{c}$ over $w$.
We are going to show that $\run'$ is not accepting; 
thus, given the arbitrary choice of $\run'$, we have that each macrorun of $\aut^{c}$ over $w$ is not accepting, hence $w \notin \lang{\aut^{c}}$.

We have the following properties about the macrorun $\run'$:
\begin{enumerate}
\item
    assume that $\run'$ only visits macrostates of the form $S \in 2^{\states}$. 
    This implies that $\run'$ is not accepting in $\aut^{c}$ since it never visits any accepting macrostate $(N, C, \emptyset) \in \acc^{c}$, so it clearly does not visit them infinitely often;
\item
    assume that $\run'$ is a macrorun of the form 
    \[
        S_{0}, \cdots S_{k-1}, (N_{k}, C_{k}, B_{k}), (N_{k+1}, C_{k+1}, B_{k+1}) \cdots
    \] 
    By assumption, we know that $\run$ visits some accepting state, say $q_{f} \in \acc$, infinitely often.
    This means that at some point, say in state $(N_{j}, C_{j}, B_{j})$, we have that $q_{f} \in B_{j}$ or that $q_{f} \in C_{j}$.
    If $q_{f} \in B_{j}$, then for every $l \geq j$, we have that $B_{l} \neq \emptyset$ according to Lemma~\ref{lem:separatingLevelDAG}:
    we can let $j$ be larger than the separating level $\separatingLevel$ of $\reduceddagAW{\aut}{w}$, which, according to Lemma~\ref{lem:separatingLevelDAG}, indicates that the accepting run $\run$ will not join with other runs after state $(N_{j}, C_{j}, B_{j})$.
    It follows that the suffix of $\run$ will not be cut off and will stay in $B_{j} $ forever.
    So we have $B_{l} \neq \emptyset$ for all $l \geq j$.
    If $q_{f} \in C_{j}$, then either at some point, say $i > j$, $q_{f}$ will be moved to $B_{i}$ when $B_{i-1} = \emptyset$, or $q_{f} \in C_{i}$ for each $i \geq j$, which indicates that $B_{i} \neq \emptyset$ for $i \geq j$.
    This is the case because when the set $B$ of $(N, C, B)$ is empty, each successor macrostate $(N', C', B')$ has $B'$ equal to $C'$ (cf.\@ Definition~\ref{def:FANBAsliceBasedComplementation}).
\end{enumerate}
In both cases, we have that no macrorun of $\aut^{c}$ over $w$ is accepting, so $w$ is not accepted by $\aut^{c}$, as required.  

\subsection{Proof of \texorpdfstring{$w \notin \lang{\aut} \implies w \in \lang{\aut^{c}}$}{w not accepted by A implies w accepted by the complement of A}}
\label{app:FANBAlanguageSizeSliceBasedComplement:languageWnotinFANBAinComplement}

Assume now that $w \notin \lang{\aut}$; 
our goal is to show that there exists an accepting macrorun $\run'$ of $\aut^{c}$ over $w$.
The proof idea is to analyze the codeterministic DAG~$\reduceddagAW{\aut}{w}$ of $\aut$ over $w$.
According to Lemma~\ref{lem:FANBAwordRejectedStableLevelInCodeterministicDAG}, there exists a stable level $\stableLevel \geq 1$ such that every $\acc$-vertex on a level after $\stableLevel$ of $\reduceddagAW{\aut}{w}$ is finite.
This means that after reading $\stableLevel$ letters from $w$, $\aut^{c}$ can choose to jump to the accepting phase.
By the definition of $\jtrans^{c}$, the transition relation for the jumping from the initial phase to the accepting phase (cf.\@ Definition~\ref{def:FANBAsliceBasedComplementation}), we get that the successors of the current $\acc$-vertices are all collected in the component $B'$ of the successor macrostate $(N', C', B')$. 
Similarly, in $C'$ we collect $B'$ as well as the accepting states present in $N'$.
Since $\acc$-vertices are finite, all branches from $\acc$-vertices will eventually disappear;
this means that the corresponding states disappear from $C$ and $B$ as well.
Once this happens, the set $B$ becomes empty again; 
the set $C$ may still contain some state since from the jump, some accepting state may have been visited and collected in the $C$ component.
However, since $w \notin \lang{\aut}$, there is $K \in \naturals$ such that no $\acc$-vertex is reachable after level $K$, thus both components $C$ and $B$ will eventually become empty and stay empty forever. 
Therefore, the set $B$ on $\run'$ will become empty infinitely often, so $w$ is accepted by $\aut^{c}$.
This completes the proof for the claim about the language of $\aut^{c}$.

\subsection{Proof of the size of \texorpdfstring{$\aut^{c}$}{the complement of A}}
\label{app:FANBAlanguageSizeSliceBasedComplement:size}

Consider now the size of $\aut^{c}$:
by Definition~\ref{def:FANBAsliceBasedComplementation}, the number of possible macrostates of the form $S \in 2^{\states}$ is $2^{n}$.
For each macrostate $(N, C, B) \in \states^{c}$, we have that $B \subseteq C \subseteq N$ holds: 
recall that Definition~\ref{def:FANBAsliceBasedComplementation} requires $\states^{c}$ to be the smallest set such that $\trans^{c}(\states^{c}, a) \subseteq \states^{c}$ for each $a \in \alphabet$, so to prove that $B \subseteq C \subseteq N$ holds for each macrostate $(N, C, B) \in \states^{c}$, it suffices to show that $B \subseteq C \subseteq N$ holds for each macrostate $(N, C, B)$ in the image of $\trans^{c}$.
By Definition~\ref{def:FANBAsliceBasedComplementation} such macrostates are only those introduced by $\jtrans$ and their successors as generated by $\dtrans$.
Regarding the macrostates in the image of $\jtrans$, given $S \in \states^{c}$ and $b \in \alphabet$, we have that $\jtrans(S, b) = (N', C', B')$ where $B' = \cotrans(S \cap \acc, b) \subseteq \cotrans(S \cap \acc, b) \cup (N' \cap \acc) = C' \subseteq \cotrans(S, b) = N'$.
Regarding the macrostates in the image of $\dtrans$, given $(N, C, B) \in \states^{c}$ such that $B \subseteq C \subseteq N$ holds and $b \in \alphabet$, we have two cases.
If $B \neq \emptyset$, then $B' = \cotrans(B, b) \subseteq \cotrans(C, b) \cup (N' \cap \acc) = C' \subseteq \cotrans(N, b) = N'$; 
the first inclusion follows by $B \subseteq C$ while the second by $C \subseteq N$.
Similarly, if $B = \emptyset$, then $B' = C' = \cotrans(C, b) \cup (N' \cap \acc) \subseteq \cotrans(N, b) = N'$.

Given that $B \subseteq C \subseteq N$ holds for each $(N, C, B) \in \states^{c}$, for a state $q \in \states$, we have these four possibilities:
\begin{enumerate}
\item 
    $q \in B \subseteq C \subseteq N$;
\item 
    $q \in C \subseteq N$ but $q \notin B$;
\item 
    $q \in N$ but $q \notin C \supseteq B$;
\item
    $q \notin N \supseteq C \supseteq B$.
\end{enumerate}
Since for each state $q \in \states$ we have four choices, we have at most $4^{n}$ possible macrostates $(N, C, B)$.
Therefore $\aut^{c}$ has at most $2^{n} + 4^{n} \in \bigO(4^{n})$ states.

\section{Proof of Lemma~\ref{lem:LDBAPropertiesOfCodeterministicDAG}}
\label{app:LDBAPropertiesOfCodeterministicDAG}

\LDBAPropertiesOfCodeterministicDAG*
\noindent
We prove Properties~\ref{lem:LDBAPropertiesOfCodeterministicDAG:acceptance} and~\ref{lem:LDBAPropertiesOfCodeterministicDAG:numberOmegaBranches} independently in~\ref{app:LDBAPropertiesOfCodeterministicDAG:acceptance} and~\ref{app:LDBAPropertiesOfCodeterministicDAG:numberOmegaBranches}, respectively, and then use them in the proof of Property~\ref{lem:LDBAPropertiesOfCodeterministicDAG:stableLevel} in~\ref{app:LDBAPropertiesOfCodeterministicDAG:stableLevel}.

\subsection{Proof of Property~\ref{lem:LDBAPropertiesOfCodeterministicDAG:acceptance}}
\label{app:LDBAPropertiesOfCodeterministicDAG:acceptance}

This property requires us to prove that
$w$ is accepted by $\aut$ if and only if $\lcodagAW{\aut}{w}$ is accepting.
Suppose that $w \notin \lang{\aut}$;
we now show that $\lcodagAW{\aut}{w}$ is not accepting.
By construction, we have that each $\omega$-branch of $\lcodagAW{\aut}{w}$ represents (possibly several) runs of $\aut$: 
it is easy to see that e.g.\@ the $\omega$-branch $\vertex{N_{0}}{0} \cdots \vertex{N_{k}}{k} \vertex{d_{k + 1}}{k + 1} \vertex{d_{k + 2}}{k + 2} \cdots$ where $N_{0} = \initialStates \cap \nstates$, $N_{i + 1} = \ntrans(N_{i}, \wordletter{w}{i})$ for each $i < k$, $d_{k+1} \in \jtrans(N_{k}, \wordletter{w}{k})$ and $d_{j + 1} = \dtrans(d_{j}, \wordletter{w}{j})$ for each $j \geq k$ represents the run $\run = n_{0} n_{1} \cdots n_{k} d_{k + 1} d_{k + 2} \cdots$ with $n_{i} \in N_{i}$ for each $i \leq k$ and $d_{j} \in \dstates$ for each $j > k$. 
The assumption $w \notin \lang{\aut}$ implies that each run $\run$ of $\aut$ over $w$ is not accepting, that is, there are only finitely many accepting states in $\run$;
this means that also each $\omega$-branch of $\lcodagAW{\aut}{w}$ is not accepting.
If this would not be the case, then at least one of the runs represented by such an accepting $\omega$-branch would be accepting, so $w \in \lang{\aut}$ which contradicts the assumption $w \notin \lang{\aut}$.
Since each $\omega$-branch of $\lcodagAW{\aut}{w}$ is not accepting, it follows that $\lcodagAW{\aut}{w}$ is not accepting, as required.

Suppose now that $w \in \lang{\aut}$;
among all accepting runs of $\aut$ over $w$ there is one, say $\run_{m}$, that is minimal in the following sense:
for each accepting run $\run'$ of $\aut$ over $w$ we have that 
\begin{enumerate}
\item
    $\ell_{\run_{m}} \leq \ell_{\run'}$ where $\ell_{\run} = \min\setcond{i \in \naturals}{\wordletter{\run}{i} \in \dstates}$ and
\item
    $\wordletter{\run_{m}}{\ell_{\run_{m}}} \stateOrder \wordletter{\run'}{\ell_{\run'}}$.
\end{enumerate}
Intuitively, $\ell_{\run}$ represents at what step $\run$ jumped to the deterministic part of $\aut$;
note that we can have $\ell_{\run} = 0$, which means that $\wordletter{\run}{0} \in \initialStates \cap \dstates$.
The first condition ensures that $\run_{m}$ is one of the (possibly several) accepting runs that jump to the deterministic states as soon as possible; 
the second condition ensures that when $\run_{m}$ jumps to the deterministic states, it lands on the smallest possible state that is on an accepting run.
Note that there may be several accepting runs $\run'_{m}$ that are minimal: 
the only differences between each of them and $\run_{m}$ are in their initial part involving nondeterministic states:
since $\run_{m}$ and $\run'_{m}$ are both minimal, we have that $\ell_{\run_{m}} = \ell_{\run'_{m}}$ and $\wordletter{\run_{m}}{\ell_{\run_{m}}} \stateOrder \wordletter{\run'_{m}}{\ell_{\run'_{m}}}$.
Moreover, we have that for each $0 \leq i < \ell_{\run_{m}}$, $\wordletter{\run_{m}}{i} \in \nstates$ as well as $\wordletter{\run'_{m}}{i} \in \nstates$.
Since $\ell_{\run_{m}}$ is finite and $\nstates$ is also finite, there are only finitely many such runs $\run'_{m}$ and all of them merge at latest on $\wordletter{\run_{m}}{\ell_{\run_{m}}}$.
After state $\wordletter{\run_{m}}{\ell_{\run_{m}}}$, $\run_{m}$ is deterministic;
it is still possible that other runs merge with it, but the minimality of $\run_{m}$ ensures that either such runs jumped to $\dstates$ after $\ell_{\run_{m}}$, or that they jumped at $\ell_{\run_{m}}$ but landed on states that ``come later'' according to $\stateOrder$.

If $\ell_{\run_{m}} = 0$, then we have that $\wordletter{\run_{m}}{0} \in \initialStates \cap \dstates$ and that $\run_{m}$ is unique.
This implies that in $\lcodagAW{\aut}{w}$ we have the $\omega$-branch $\hat{\run}_{m} = \vertex{\wordletter{\run_{m}}{0}}{0} \vertex{\wordletter{\run_{m}}{1}}{1} \cdots$. 
This is the case because, by Definition~\ref{def:LDBAcodeterministicDAGandPriorityFunction}, we have that $\vertex{\wordletter{\run_{m}}{0}}{0} \in \dvertices$ and, for each $i \in \naturals$, that $(\vertex{\wordletter{\run_{m}}{i}}{i}, \vertex{\wordletter{\run_{m}}{i + 1}}{i + 1}) \in \edges$ and $\priority(\vertex{\wordletter{\run_{m}}{j + 1}}{j + 1}) = \priority(\vertex{\wordletter{\run_{m}}{j}}{j})$, which can be easily shown by induction over $i$ with the help of the two conditions given above about the minimality of $\run_{m}$:
for the case base, suppose for the sake of contradiction that $(\vertex{\wordletter{\run_{m}}{0}}{0}, \vertex{\wordletter{\run_{m}}{1}}{1}) \notin \edges$.
By definition of $\lcodagAW{\aut}{w}$, this can only be the case if there is a vertex $\vertex{q}{0} \in \vertices$ such that $\wordletter{\run}{1} \in \trans(q, \wordletter{w}{0})$ and $\priority(\vertex{q}{0}) < \priority(\vertex{\wordletter{\run_{m}}{0}}{0})$.
Since $\wordletter{\run_{m}}{0} \in \dstates$, this implies that also $q \in \dstates$, because the nondeterministic vertex $\vertex{\initialStates \cap \nstates}{0}$, provided that we have it, has by definition a priority strictly larger than $\priority(\vertex{\wordletter{\run_{m}}{0}}{0})$; 
moreover, $q \stateOrder \wordletter{\run_{m}}{0}$, so the run $q \wordletter{\run_{m}}{1} \wordletter{\run_{m}}{2} \cdots$ is an accepting run of $\aut$ over $w$ that is smaller than $\run_{m}$, against the hypothesis that it is minimal.
Thus we have $(\vertex{\wordletter{\run_{m}}{0}}{0}, \vertex{\wordletter{\run_{m}}{1}}{1}) \in \edges$ and, by definition of $\edges$, $\priority(\vertex{\wordletter{\run_{m}}{1}}{1}) = \priority(\vertex{\wordletter{\run_{m}}{0}}{0})$, as required.

For the inductive step, suppose that for each $0 \leq j < i$, we have the edge $(\vertex{\wordletter{\run_{m}}{j}}{j}, \vertex{\wordletter{\run_{m}}{j + 1}}{j + 1}) \in \edges$ and that $\priority(\vertex{\wordletter{\run_{m}}{j + 1}}{j + 1}) = \priority(\vertex{\wordletter{\run_{m}}{j}}{j})$.
Since $\wordletter{\run_{m}}{0} \in \initialStates \cap \dstates$ and $\aut$ is an LDBA, we also have that $\wordletter{\run_{m}}{j} \in \dstates$, hence $\vertex{\wordletter{\run_{m}}{j}}{j} \in \dvertices$ for each $0 \leq j \leq i$.
Suppose, for the sake of contradiction, that $(\vertex{\wordletter{\run_{m}}{i}}{i}, \vertex{\wordletter{\run_{m}}{i + i1}}{i + 1}) \notin \edges$;
similarly to the base case, this can only happen if there is a vertex $\vertex{q}{i} \in \vertices$ such that $\wordletter{\run}{i + 1} \in \trans(q, \wordletter{w}{i})$ and $\priority(\vertex{q}{i}) < \priority(\vertex{\wordletter{\run_{m}}{i}}{i})$.
Since $\wordletter{\run_{m}}{i} \in \dstates$, this implies that also $q \in \dstates$, because the nondeterministic vertex $\vertex{N_{i}}{i}$, provided that we have it, has by definition a priority strictly larger than $\priority(\vertex{\wordletter{\run_{m}}{i}}{i})$.
By definition of $\edges$, we have that $\vertex{q}{i} \in \vertices$ because there is a run $\run_{q}$ such that $\run_{q}$ merges with $\run_{m}$ at $\wordletter{\run}{i + 1}$, $\wordletter{\run_{q}}{i} = q$, and for each $0 \leq j < i$, $\wordletter{\run_{q}}{j} \in \dstates$ and  $\priority(\vertex{\wordletter{\run_{q}}{j + 1}}{j + 1}) = \priority(\vertex{\wordletter{\run_{q}}{j}}{j})$.
This implies that the run $\run_{q}$ is an accepting run of $\aut$ over $w$ that is smaller than $\run_{m}$, against the hypothesis that it is minimal.
Thus we have $(\vertex{\wordletter{\run_{m}}{i}}{i}, \vertex{\wordletter{\run_{m}}{i + 1}}{i + 1}) \in \edges$ and, by definition of $\edges$, $\priority(\vertex{\wordletter{\run_{m}}{i + 1}}{i + 1}) = \priority(\vertex{\wordletter{\run_{m}}{i}}{i})$, as required.

We also have that $\hat{\run}_{m}$ is accepting: 
by hypothesis we have that $\run_{m}$ is accepting, thus $\wordletter{\run_{m}}{i} \in \acc$ holds for infinitely many $i \in \naturals$, which implies that $\vertex{\wordletter{\run_{m}}{i}}{i}$ is an $\acc$-vertex, i.e., $\hat{\run}_{m}$ is accepting; 
this means that $\lcodagAW{\aut}{w}$ is accepting as well.

If $\ell_{\run_{m}} > 0$, then we have that $\wordletter{\run_{m}}{0} \in \initialStates \cap \nstates$. 
Let $\hat{\run}$ be the sequence of vertices
\[
    \hat{\run} = \vertex{N_{0}}{0} \cdots \vertex{N_{\ell_{\run_{m}} - 1}}{\ell_{\run_{m}} - 1} \vertex{\wordletter{\run_{m}}{\ell_{\run_{m}}}}{\ell_{\run_{m}}} \vertex{\wordletter{\run_{m}}{\ell_{\run_{m}} + 1}}{\ell_{\run_{m}} + 1} \cdots
\]
where $N_{0} = \initialStates \cap \nstates$ and $N_{i + 1} = \ntrans(N_{i}, \wordletter{w}{i})$ for each $0 \leq i < \ell_{\run_{m}} - 1$. 
We claim that $\hat{\run}$ is an $\omega$-branch of $\lcodagAW{\aut}{w}$.
It is easy to see that for each $0 \leq i < \ell_{\run_{m}} - 1$, $\wordletter{\run_{m}}{i} \in N_{i}$ and that $\vertex{N_{i}}{i} \in \nvertices$:
for $i = 0$ this is immediate by the fact that $\wordletter{\run_{m}}{0} \in \initialStates \cap \nstates = N_{0}$ and by definition, $\vertex{\initialStates \cap \nstates}{0} \in \nvertices$; 
if $\wordletter{\run_{m}}{i} \in N_{i}$ and $\vertex{N_{i}}{i} \in \nvertices$, then we have that $\wordletter{\run_{m}}{i + 1} \in \ntrans(\wordletter{\run_{m}}{i}, \wordletter{w}{i})$ implies that $\wordletter{\run_{m}}{i + 1} \in N_{i + 1} = \ntrans(N_{i}, \wordletter{w}{i})$, thus $\vertex{N_{i + 1}}{i + 1} \in \nvertices$.
Note that for each $0 \leq i < \ell_{\run_{m}} - 1$, we have $(\vertex{N_{i}}{i}, \vertex{N_{i + 1}}{i + 1}) \in \edges$.

Regarding the jumping point $\ell_{\run_{m}}$, from the previous analysis we know that $\wordletter{\run_{m}}{\ell_{\run_{m}} - 1} \in N_{\ell_{\run_{m}} - 1}$;
from $\run_{m}$ and the definition of $\ell_{\run_{m}}$ we also know that $\wordletter{\run_{m}}{\ell_{\run_{m}}} \in \jtrans(\wordletter{\run_{m}}{\ell_{\run_{m}} - 1}, \wordletter{w}{\ell_{\run_{m}} - 1}) \subseteq \dstates$. 
This implies that $\wordletter{\run_{m}}{\ell_{\run_{m}}} \in \jtrans(N_{\ell_{\run_{m}} - 1}, \wordletter{w}{\ell_{\run_{m}} - 1})$. 
From this we derive that $\vertex{\wordletter{\run_{m}}{\ell_{\run_{m}}}}{\ell_{\run_{m}}} \in \dvertices$.
It remains to show that $(\vertex{N_{\ell_{\run_{m}} - 1}}{\ell_{\run_{m}} - 1}, \vertex{\wordletter{\run_{m}}{\ell_{\run_{m}}}}{\ell_{\run_{m}}}) \in \edges$.
This is the case because $\wordletter{\run_{m}}{\ell_{\run_{m}}}$ is the minimum element in $J_{\ell_{\run_{m}}}$:
the minimality of $\wordletter{\run_{m}}{\ell_{\run_{m}}}$ follows trivially by the definition of minimality of $\run_{m}$; 
the fact that $\wordletter{\run_{m}}{\ell_{\run_{m}}} \in J_{\ell_{\run_{m}}}$ is due to the fact that $\wordletter{\run_{m}}{\ell_{\run_{m}}} \in \jtrans(N_{\ell_{\run_{m}} - 1}, \wordletter{w}{\ell_{\run_{m}} - 1})$ and that $\wordletter{\run_{m}}{\ell_{\run_{m}}} \notin D_{\ell_{\run_{m}}}$.
For the sake of contradiction, suppose that $\wordletter{\run_{m}}{\ell_{\run_{m}}} \in D_{\ell_{\run_{m}}}$.
This implies that $\wordletter{\run_{m}}{\ell_{\run_{m}}} \in \dtrans(P_{\ell_{\run_{m}} - 1}, \wordletter{\run_{m}}{\ell_{\run_{m}} - 1})$.
Recall that $P_{\ell_{\run_{m}} - 1} \subseteq \dstates$ and that there is some state $d \in \dstates$ such that $\vertex{d}{\ell_{\run_{m}} - 1} \in \dvertices$ has the same priority of $\wordletter{\run_{m}}{\ell_{\run_{m}}}$ and that $\vertex{d}{\ell_{\run_{m}} - 1}$ is reachable in $\lcodagAW{\aut}{w}$ from some vertex at level $0$.
This implies that there is a run $\run_{c}$ such that $\wordletter{\run_{c}}{0} \in \initialStates$, $\wordletter{\run_{c}}{\ell_{\run_{m}} - 1} = d$, $\wordletter{\run_{c}}{\ell_{\run_{m}}} = \wordletter{\run_{m}}{\ell_{\run_{m}}}$.
Thus $\run_{c}$ merges with $\run_{m}$ on $\wordletter{\run_{c}}{\ell_{\run_{m}}}$, so $\run_{c}$ is also accepting.
This implies that $\ell_{\run_{c}} \leq \ell_{\run_{m}} - 1$ (since $\wordletter{\run_{c}}{\ell_{\run_{m}} - 1} = d \in \dstates$), which contradicts the minimality of $\run_{m}$ with respect to $\ell_{\run_{m}}$.
With this, we have that $(\vertex{N_{\ell_{\run_{m}} - 1}}{\ell_{\run_{m}} - 1}, \vertex{\wordletter{\run_{m}}{\ell_{\run_{m}}}}{\ell_{\run_{m}}}) \in \edges$ as desired.

By the same reasoning for the case $\ell_{\run_{m}} = 0$, we have that $\vertex{\wordletter{\run_{m}}{\ell_{\run_{m}}}}{\ell_{\run_{m}}} \in \dvertices$ and, for each $i \geq \ell_{\run_{m}}$, that $(\vertex{\wordletter{\run_{m}}{i}}{i}, \vertex{\wordletter{\run_{m}}{i + 1}}{i + 1}) \in \edges$.
We also have that $\hat{\run}_{m}$ is accepting: 
by hypothesis we have that $\run_{m}$ is accepting, thus $\wordletter{\run_{m}}{i} \in \acc$ holds for infinitely many $i \in \naturals$, which implies that $\vertex{\wordletter{\run_{m}}{i}}{i}$ is an $\acc$-vertex, i.e., $\hat{\run}_{m}$ is accepting; 
this means that $\lcodagAW{\aut}{w}$ is accepting as well.

This completes the proof that $w$ is accepted by $\aut$ implies that $\lcodagAW{\aut}{w}$ is accepting.

\subsection{Proof of Property~\ref{lem:LDBAPropertiesOfCodeterministicDAG:numberOmegaBranches}}
\label{app:LDBAPropertiesOfCodeterministicDAG:numberOmegaBranches}

We have to show that the number of (accepting) $\omega$-branches in $\lcodagAW{\aut}{w}$ is at most the number of states in $\aut$.
Similarly to the proof of Lemma~\ref{lem:FANBAfiniteOmegaBranchesInCodeterministicDAG}, let $m_{\ell}$ be the number of vertices which are in the $\omega$-branches on level $\ell$. 
Since each vertex in $\lcodagAW{\aut}{w}$ has only one predecessor, due to the fact that $\lcodagAW{\aut}{w}$ is codeterministic, we have that $m_{0} \leq m_{1} \leq m_{2} \leq \cdots$, i.e., the number of vertices in the $\omega$-branches on each level does not decrease over the levels. 
In particular, we have that the number of $\omega$-branches can be increased only by the jump from the initial phase, that is nondeterministic, to the accepting phase, that is deterministic. 
Moreover, since there are at most $\size{\states}$ states on each level, the number $b_{\ell}$ of $\omega$-branches at each level $\ell$ cannot exceed $\size{\states}$: 
in order to cross this threshold, there must be a level $k$ such that $b_{k} \leq \size{\states}$ and $b_{k + 1} > \size{\states}$.
By the Pigeonhole principle, this implies that two different branches must share a common vertex $v$ at level $k + 1$, since there are $b_{k + 1} > \size{\states}$ $\omega$-branches and at most $\size{\states}$ vertices.
This means that $v$ has two predecessors, which contradicts the fact that $\lcodagAW{\aut}{w}$ is codeterministic.

\subsection{Proof of Property~\ref{lem:LDBAPropertiesOfCodeterministicDAG:stableLevel}}
\label{app:LDBAPropertiesOfCodeterministicDAG:stableLevel}

We need to prove there exists a stable level $\stableLevel \geq 1$ in $\lcodagAW{\aut}{w}$ such that all $\acc$-vertices after level $\stableLevel$ are finite if and only if $w \notin \lang{\aut}$.
Assume that there is a stable level $\stableLevel$ in $\lcodagAW{A}{w}$;
this implies, by definition, that every $\acc$-vertex $\vertex{f}{k}$ with $f \in \acc$ and $k \geq \stableLevel$ is finite, that is, $\vertex{f}{k}$ does not occur in an $\omega$-branch.
Therefore, each $\omega$-branch of $\lcodagAW{\aut}{w}$ is not accepting, i.e., $\lcodagAW{\aut}{w}$ is not accepting, so by Property~\ref{lem:LDBAPropertiesOfCodeterministicDAG:acceptance} we have $w \notin \lang{\aut}$.
This completes the proof that if there exists a stable level $\stableLevel \geq 1$ in $\lcodagAW{\aut}{w}$ such that all $\acc$-vertices after level $\stableLevel$ are finite, then $w \notin \lang{\aut}$.

Consider now the opposite direction and assume that $w \notin \lang{\aut}$. 
By Property~\ref{lem:LDBAPropertiesOfCodeterministicDAG:acceptance}, it follows that $\lcodagAW{\aut}{w}$ is not accepting, that is, each $\omega$-branch is not accepting, i.e., there are only finitely many $\acc$-vertices in each $\omega$-branch.
Since by Property~\ref{lem:LDBAPropertiesOfCodeterministicDAG:numberOmegaBranches} the number of $\omega$-branches is at most $\size{\states}$, i.e., it is finite, it follows that there exists a level $\stableLevel$ such that each vertex $\vertex{q}{k}$ of $\lcodagAW{\aut}{w}$ with $k \geq \stableLevel$ is not an $\acc$-vertex.
This implies that $\stableLevel$ is a stable level, since Definition~\ref{def:stableLevel} is trivially satisfied, as desired.
Moreover, since there are no $\acc$-vertices after level $\stableLevel$, it is trivial to see that each of them is finite.

This completes the proof of the fact that $w \notin \lang{\aut}$ implies that there exists a stable level $\stableLevel \geq 1$ in $\lcodagAW{\aut}{w}$ such that all $\acc$-vertices after level $\stableLevel$ are finite as well as of Property~\ref{lem:LDBAPropertiesOfCodeterministicDAG:stableLevel}.

\section{Proof of Theorem~\ref{thm:LDBAlanguageSizeComplementNBA}}
\label{app:LDBAlanguageSizeComplementNBA}

\LDBAlanguageSizeComplementNBA*

Theorem~\ref{thm:LDBAlanguageSizeComplementNBA} is a direct consequence of the proofs we present in the following subsections, namely \ref{app:LDBAlanguageSizeComplementNBA:languageWinLDBAnotInComplement} that shows that whenever $w \in \lang{\aut}$, then we have $w \notin \lang{\aut^{c}}$, i.e., $w \in \infwords \setminus \lang{\aut^{c}}$;
\ref{app:LDBAlanguageSizeComplementNBA:languageWnotinLDBAinComplement} for the dual property $w \notin \lang{\aut} \implies w \in \lang{\aut^{c}}$, that is, $w \notin \lang{\aut} \implies w \notin \infwords \setminus \lang{\aut^{c}}$; 
and 
\ref{app:LDBAlanguageSizeComplementNBA:size} taking care of establishing the size of $\aut^{c}$.

\subsection{Proof of \texorpdfstring{$w \in \lang{\aut} \implies w \notin \lang{\aut^{c}}$}{w accepted by A implies w not accepted by the complement of A}}
\label{app:LDBAlanguageSizeComplementNBA:languageWinLDBAnotInComplement}

We need to prove that for each $w \in \lang{\aut}$ it holds that $w \notin \lang{\aut^{c}}$. 
Consider an arbitrary $w \in \lang{\aut}$ and let $\run^{c}$ be a macrorun of $\aut^{c}$ over $w$. 
We are going to show that $\run^{c}$ is not accepting; 
thus, giving the arbitrary choice of $\run^{c}$, we have that each macrorun of $\aut^{c}$ over $w$ is not accepting, hence $w \notin \lang{\aut^{c}}$. 
Similar to the macrorun in Definition~\ref{def:FANBAsliceBasedComplementation}, we can have only two types of macroruns:
the ones staying in $2^{\nstates}$ and the ones jumping to $2^{\nstates} \times 2^{\dstates} \times 2^{\dstates} \times 2^{\dstates}$.

If $\run^{c}$ visits only macrostates of the form $R \in 2^{\nstates}$, then $\run^{c} = R_{0} R_{1} \cdots$, where $R_{0} = \initialStates$.
This means that $\run^{c}$ never visits any accepting macrostate $(N, S, B, C) \in \acc^{c}$, so it trivially also does not visit accepting macrostates $(N, S, \emptyset, C) \in \acc^{c}$ infinitely often, thus $\run^{c}$ is not accepting.

If $\run^{c}$ does not visit only macrostates of the form $R \in 2^{\nstates}$, then we have that $\run^{c}$ is a macrorun of the form 
\[
	\run^{c} = R_{0} \cdots R_{\macrorunJump} (N_{\macrorunJump + 1}, S_{\macrorunJump + 1}, B_{\macrorunJump + 1}, C_{\macrorunJump + 1}) (N_{\macrorunJump + 2}, S_{\macrorunJump + 2}, B_{\macrorunJump + 2}, C_{\macrorunJump + 2}) \cdots
\] 
for some $\macrorunJump \in \naturals$, where $R_{0} = \initialStates$. 
Let $\run$ be an arbitrary accepting run of $\aut$ over $w$;
such a run must exist since $w \in \lang{\aut}$.
Since $\run$ is accepting, it has the following properties.
\begin{enumerate}
\item
	$\wordletter{\run}{0} \in \initialStates$, due to the requirement that $\run$ has to be initial in order to be accepting;
	this implies that $\wordletter{\run}{0} \in R_{0}$.
\item
	$\infstates{\run} \cap \acc \neq \emptyset$, by definition of accepting run, implies that there exists $\runFirstInfiniteAccepting \in \naturals$ such that $\runFirstInfiniteAccepting = \min\setcond{i \in \naturals}{\wordletter{\run}{i} \in \infstates{\run} \cap \acc}$, that is, $\runFirstInfiniteAccepting$ indicates the first moment an accepting state occurring infinitely often has been visited.
\item
	There exists $\runJumped \in \naturals$ such that $\wordletter{\run}{\runJumped} \in \dstates$ and $\wordletter{\run}{i} \in \nstates$ for each $0 \leq i < \runJumped$. 
	This follows from the fact that $\acc \subseteq \dstates$ and $\infstates{\run} \cap \acc \neq \emptyset$;
	note that $\runJumped \leq \min\setcond{i \in \naturals}{\wordletter{\run}{i} \in \acc} \leq \runFirstInfiniteAccepting$.
\end{enumerate}
We also have that for each $0 < i \leq \macrorunJump$, $\wordletter{\run}{i} \in R_{i}$, since $R_{i} = \trans^{c}(R_{i-1}, \wordletter{w}{i-1}) = \trans(R_{i-1}, \wordletter{w}{i-1})$, by definition of $\trans^{c}$, and $\wordletter{\run}{i} \in \trans(\wordletter{\run}{i-1}, \wordletter{w}{i-1})$ by the fact that $\run$ is a run.

At step $\macrorunJump$, the states in $R_{\macrorunJump}$ are split into the pseudo-macrostate $(N_{\macrorunJump}, S_{\macrorunJump}, B_{\macrorunJump}, C_{\macrorunJump}) = (R_{\macrorunJump} \cap \nstates, (R_{\macrorunJump} \cap \dstates) \setminus \acc, R_{\macrorunJump} \cap \acc, \emptyset)$;
the state $\wordletter{\run}{\macrorunJump}$ belongs to only one of these components.
\begin{description}
\item[Case $\macrorunJump \geq \runJumped$:] 
	we have that $\wordletter{\run}{\macrorunJump} \in \dstates$, thus $\wordletter{\run}{\macrorunJump} \notin R_{\macrorunJump} \cap \nstates = N_{\macrorunJump}$.
	If $\wordletter{\run}{\macrorunJump} \in \dstates \setminus \acc$, then $\wordletter{\run}{\macrorunJump} \in (R_{\macrorunJump} \cap \dstates) \setminus \acc = S_{\macrorunJump}$ and $\wordletter{\run}{\macrorunJump} \notin R_{\macrorunJump} \cap \acc = B_{\macrorunJump}$; 
	otherwise, for $\wordletter{\run}{\macrorunJump} \in \acc \subseteq \dstates$, we have that $\wordletter{\run}{\macrorunJump} \notin (R_{\macrorunJump} \cap \dstates) \setminus \acc = S_{\macrorunJump}$ and $\wordletter{\run}{\macrorunJump} \in R_{\macrorunJump} \cap \acc = B_{\macrorunJump}$.
	Clearly, $\wordletter{\run}{\macrorunJump} \notin \emptyset = C_{\macrorunJump}$.
 
\item[Case $\macrorunJump < \runJumped$:]
	in this case we have that $\wordletter{\run}{\macrorunJump} \in \nstates$, so we have that $\wordletter{\run}{\macrorunJump} \in R_{\macrorunJump} \cap \nstates = N_{\macrorunJump}$, $\wordletter{\run}{\macrorunJump} \notin (R_{\macrorunJump} \cap \dstates) \setminus \acc = S_{\macrorunJump}$, and $\wordletter{\run}{\macrorunJump} \notin R_{\macrorunJump} \cap \acc = B_{\macrorunJump}$.
	Also, $\wordletter{\run}{\macrorunJump} \notin \emptyset = C_{\macrorunJump}$.
	
	By a simple proof by induction, we can show that for each $\macrorunJump < i < \runJumped$, we have that $\wordletter{\run}{i} \in N_{i}$ since $\wordletter{\run}{i} \in \nstates$ and the inductive hypothesis $\wordletter{\run}{i} \in N_{i} = \ntrans(N_{i - 1}, \wordletter{w}{i - 1})$.
	The fact that $\wordletter{\run}{i} \notin S_{i} \cup B_{i} \cup C_{i}$ is immediate by the fact that $S_{i} \cup B_{i} \cup C_{i} \subseteq \dstates$ by definition.
	
	Lastly, $\run$ jumps from $\wordletter{\run}{\runJumped - 1} \in \nstates$ to $\dstates$ on reading the letter $\wordletter{w}{\runJumped - 1}$.
	By definition of $\trans^{c}$, it follows that $\wordletter{\run}{\runJumped} \notin N_{\runJumped} \subseteq \nstates$ since $\wordletter{\run}{\runJumped} \in \jtrans(\wordletter{\run}{\runJumped - 1}, \wordletter{w}{\runJumped - 1}) \subseteq \dstates$.
	Moreover, if $B_{\runJumped - 1} \neq \emptyset$, then $\wordletter{\run}{\runJumped} \in C_{\runJumped}$ unless it is already in $S_{\runJumped}$ or already in $B_{\runJumped}$;
	otherwise, for $B_{\runJumped - 1} = \emptyset$, we have that $\wordletter{\run}{\runJumped} \in B_{\runJumped}$ unless it is already in $S_{\runJumped}$.
	Note that by definition of $\trans^{c}$, the sets $S_{\runJumped}$, $B_{\runJumped}$, and $C_{\runJumped}$ are pairwise disjoint, so $\wordletter{\run}{\runJumped}$ belongs to exactly one of them.
\end{description}
In both cases, at the step $\bothJumped = \max\setnocond{\macrorunJump, \runJumped}$, we have that $\wordletter{\run}{\bothJumped} \notin N_{\bothJumped}$ and that $\wordletter{\run}{\bothJumped}$ belongs to exactly one between $S_{\bothJumped}$, $B_{\bothJumped}$, and $C_{\bothJumped}$.

By a simple induction, we can prove that for each $\bothJumped \leq i < \runFirstInfiniteAccepting$, we have that $\wordletter{\run}{i} \notin N_{i}$ and that $\wordletter{\run}{i}$ belongs to exactly one set between $S_{i}$, $B_{i}$, and $C_{i}$: 
the base case for $i = \bothJumped$ is exactly what we have just shown; 
for the inductive case, the fact that $\wordletter{\run}{i} \notin N_{i}$ is trivial since $\wordletter{\run}{i} \in \dstates$. 
Assume that $\wordletter{\run}{i}$ belongs to exactly one set between $S_{i}$, $B_{i}$, and $C_{i}$;
if $\wordletter{\run}{i} \in S_{i}$, then either $\wordletter{\run}{i + 1} \in \dtrans(S_{i}, \wordletter{w}{i}) \setminus \acc = S_{i+1}$, or $\wordletter{\run}{i + 1} \in \dtrans(S_{i}, \wordletter{w}{i}) \cap \acc$, thus it either belongs to $C_{i + 1}$ or to $B_{i + 1}$.
If $\wordletter{\run}{i} \in B_{i}$, then $\wordletter{\run}{i} \notin S_{i}$, so  $\wordletter{\run}{i + 1} \in B_{i + 1} = \dtrans(B_{i}, \wordletter{w}{i}) \setminus S_{i+1}$ unless $\wordletter{\run}{i + 1} \in S_{i + 1}$;
the latter case might occur when $(\dtrans(S_{i}, \wordletter{w}{i}) \cap \acc) \cap \dtrans(B_{i}, \wordletter{w}{i}) \neq \emptyset$.
Clearly $\wordletter{\run}{i + 1} \notin C_{i + 1}$ by definition of $\trans^{c}$.
Lastly, if $\wordletter{\run}{i} \in C_{i}$, then $\wordletter{\run}{i} \notin S_{i}$ and $\wordletter{\run}{i} \notin B_{i}$, thus $\wordletter{\run}{i + 1} \in C_{i + 1}$ by definition of $\trans^{c}$ unless it is already in $S_{i + 1}$ or in $B_{i + 1}$.
This completes the inductive proof that for each $\bothJumped \leq i < \runFirstInfiniteAccepting$, we have $\wordletter{\run}{i} \notin N_{i}$ and that $\wordletter{\run}{i}$ belongs to exactly one set between $S_{i}$, $B_{i}$, and $C_{i}$.

Consider now the next step in the runs: 
from the fact that $\wordletter{\run}{\runFirstInfiniteAccepting - 1} \notin N_{\runFirstInfiniteAccepting - 1}$ and that $\wordletter{\run}{\runFirstInfiniteAccepting - 1}$ belongs to exactly one between $S_{\runFirstInfiniteAccepting - 1}$, $B_{\runFirstInfiniteAccepting - 1}$, and $C_{\runFirstInfiniteAccepting - 1}$, we get that $\wordletter{\run}{\runFirstInfiniteAccepting} \notin N_{\runFirstInfiniteAccepting}$ for the usual motivation and that 
$\wordletter{\run}{\runFirstInfiniteAccepting} \notin S_{\runFirstInfiniteAccepting}$ because $\wordletter{\run}{\runFirstInfiniteAccepting} \in \acc$ by definition of $\runFirstInfiniteAccepting$ and $S_{\runFirstInfiniteAccepting} = \dtrans(S_{\runFirstInfiniteAccepting - 1}, \wordletter{w}{\runFirstInfiniteAccepting - 1}) \setminus \acc$.
Since $\wordletter{\run}{\runFirstInfiniteAccepting - 1} \in S_{\runFirstInfiniteAccepting - 1} \cup B_{\runFirstInfiniteAccepting - 1} \cup C_{\runFirstInfiniteAccepting - 1}$, it follows that $\wordletter{\run}{\runFirstInfiniteAccepting} = \dtrans(\wordletter{\run}{\runFirstInfiniteAccepting - 1}, \wordletter{w}{\runFirstInfiniteAccepting - 1})$ belongs to either $B_{\runFirstInfiniteAccepting}$ or $C_{\runFirstInfiniteAccepting}$.
Let $q_{\runFirstInfiniteAccepting} = \wordletter{\run}{\runFirstInfiniteAccepting}$.

By repeating the same inductive proof we used above, we can show that for each $i \geq \runFirstInfiniteAccepting$, we have that $\wordletter{\run}{i} \notin N_{i}$ and that it belongs to either $S_{i}$, $B_{i}$, or $C_{i}$.
Due to the definition of the component $S$ in $(N, S, B, C)$, we have that $q_{\runFirstInfiniteAccepting}$ can never be in $S_{i}$, since $S_{i} \subseteq \dstates \setminus \acc$ and $q_{\runFirstInfiniteAccepting} \in \acc$, thus $q_{\runFirstInfiniteAccepting}$ occurs infinitely often only in $B_{i}$ or $C_{i}$.
We can generalize this result as follows.

Recall that, by definition of $\trans^{c}$, the set $S_{i}$ is obtained as $S_{i + 1} = \dtrans(S_{i}, \wordletter{w}{i}) \setminus \acc$.
Suppose that there is $i > \runFirstInfiniteAccepting$ such that $\wordletter{\run}{i} \in S_{i}$; 
this implies that there is the minimum $i' > i$ such that $\wordletter{\run}{i'} \in \acc$ (since $\infstates{\run} \cap \acc \neq \emptyset$), for which we have that 
$\wordletter{\run}{i''} \in S_{i''}$ for each $i \leq i'' < i'$ and $\wordletter{\run}{i'} \notin S_{i'}$.
This means that the finite fragment of $\run$ between $\wordletter{\run}{i}$ and $\wordletter{\run}{i'}$ ``dies'' in $S_{i'}$.
Since the component $S$ is a finite set that never gets expanded with new states (cf.\@ the definition of $\trans^{c}$), the component $S$ can only lose states in the step from $S_{k}$ to $S_{k + 1}$:
this happens when two states in $S_{k}$ have the same image under $\dtrans(\functionDot, \wordletter{w}{k})$, or when the $\wordletter{w}{k}$-successor of a state in $S_{k}$ is an accepting state.
This implies that only finitely many finite fragments of $\run$ can die in $S$; 
however we have that $\run$ is infinite, so $\run$ must ``leave'' the component $S$ at some point and never visit states tracked by the $S$ component anymore. 
That is, there must exist $\runNotInS \in \naturals$, with $\runNotInS \geq \runFirstInfiniteAccepting$, such that for each $k \geq \runNotInS$, $\wordletter{\run}{k} \notin S_{k}$.
This implies that for each $i \geq \runNotInS$, we have that $\wordletter{\run}{i} \notin N_{i} \cup S_{i}$ and that $\wordletter{\run}{i}$ belongs to either $B_{i}$ or $C_{i}$.

If $\wordletter{\run}{\runNotInS} \in B_{\runNotInS}$, then we have that $\wordletter{\run}{i} \in B_{i}$ for each $i \geq \runNotInS$.
This can be proved by a trivial induction over $i$ making use of the fact that $B_{i+1} = \dtrans(B_{i}, \wordletter{w}{i}) \setminus S_{i+1}$ and that $\wordletter{\run}{i + 1} \notin S_{i + 1}$.
Thus, we also have that $B_{i} \neq \emptyset$ for each $i \geq \runNotInS$.
This implies that $\run^{c}$ is not accepting, since $\run^{c}$ has visited accepting macrostates $(N, S, \emptyset, C)$ at most $\runNotInS$ times.

If $\wordletter{\run}{\runNotInS} \in C_{\runNotInS}$, then we have two cases: 
either $B_{i} \neq \emptyset$ for each $i \geq \runNotInS$, or there is the minimal $k \geq \runNotInS$ such that $B_{k} = \emptyset$.
In the former case, we have as before that $\run^{c}$ is not accepting, since $\run^{c}$ has visited accepting macrostates $(N, S, \emptyset, C)$ at most $\runNotInS$ times.
In the latter case, we have that $\wordletter{\run}{k + 1} \in B_{k + 1}$:
we already know that for each $i \geq \runNotInS$, we have that $\wordletter{\run}{i} \notin N_{i} \cup S_{i}$ and that $\wordletter{\run}{i}$ belongs to either $B_{i}$ or $C_{i}$.
Since $B_{k}$ is empty, this implies that $\wordletter{\run}{k} \in C_{k}$;
by definition of $\trans^{c}$, we have that $B_{k + 1} = (\dtrans(C_{k}, \wordletter{w}{k}) \cup \jtrans(N_{k}, \wordletter{w}{k}) \cup (\dtrans(S_{k}, \wordletter{w}{k}) \cap \acc)) \setminus S_{k + 1}$ and $C_{k + 1} = \emptyset$.
Since $\wordletter{\run}{k + 1} \notin S_{k + 1}$ and $\wordletter{\run}{k + 1} \in \dtrans(C_{k}, \wordletter{w}{k})$, we derive that $\wordletter{\run}{k + 1} \in B_{k + 1}$, as desired.
Since $\wordletter{\run}{k + 1} \in B_{k + 1}$, as before we have that $B_{i} \neq \emptyset$ for each $i \geq k + 1$.
This implies that $\run^{c}$ is not accepting, since $\run^{c}$ has visited accepting macrostates $(N, S, \emptyset, C)$ at most $k + 1$ times.

Therefore, given $w \in \lang{\aut}$, the arbitrary choice of the accepting run $\run$ and of the macrorun $\run^{c}$ implies that no macrorun of $\aut^{c}$ is accepting, that is, $w$ is not accepted by $\aut^{c}$, i.e., $w \notin \lang{\aut^{c}}$, as required.

\subsection{Proof of \texorpdfstring{$w \notin \lang{\aut} \implies w \in \lang{\aut^{c}}$}{w not accepted by A implies w accepted by the complement of A}}
\label{app:LDBAlanguageSizeComplementNBA:languageWnotinLDBAinComplement}

We need to prove that for each $w \notin \lang{\aut}$ it holds that $w \in \lang{\aut^{c}}$; 
we do this by showing that there exists an accepting macrorun $\run^{c}$ of $\aut^{c}$ over $w$, with the help of the codeterministic DAG~$\lcodagAW{\aut}{w}$ of $\aut$ over $w$.

In order to be accepting, the macrorun must be of the form:
\[
	\run^{c} = R_{0} R_{1} \cdots R_{\macrorunJump} (N_{\macrorunJump + 1}, S_{\macrorunJump + 1}, B_{\macrorunJump + 1}, C_{\macrorunJump + 1}) (N_{\macrorunJump + 2}, S_{\macrorunJump + 2}, B_{\macrorunJump + 2}, C_{\macrorunJump + 2}) \cdots
\]
for some $\macrorunJump \in \naturals$, where $R_{0} = \initialStates$. 
By Definition~\ref{def:LDBAnsbcComplementation}, once the level $\macrorunJump$ is chosen, the reached macrostate and all its successors in $\run^{c}$ are uniquely determined by the deterministic transition function $\dtrans^{c}$.

According to Lemma~\ref{lem:LDBAPropertiesOfCodeterministicDAG}, Property~\ref{lem:LDBAPropertiesOfCodeterministicDAG:stableLevel}, there exists a stable level $\stableLevel \geq 1$ such that all $\acc$-vertices after level $\stableLevel$ are finite. 
Recall that by definition, a vertex $\vertex{q}{i}$ is finite if there is no $\omega$-branch starting from $\vertex{q}{i}$, so Lemma~\ref{lem:LDBAPropertiesOfCodeterministicDAG}, Property~\ref{lem:LDBAPropertiesOfCodeterministicDAG:stableLevel} implies that all $\acc$-vertices after level $\stableLevel$ have only finitely many descendants. 
By choosing accurately $\macrorunJump \geq \stableLevel$ for the jump of $\run^{c}$, we get that the macrorun $\run^{c}$ is accepting, as we prove below.

We start by presenting few preliminary definitions and results we need in the proof.

\begin{definition}
\label{def:LDBAlanguageSizeComplementNBA:mappingRepresentation}
	Let $\runset$ be the set of runs of the LDBA~$\aut$ over $w \in \infwords$ and $\branchset$ be the set of initial branches of $\lcodagAW{\aut}{w}$. 
	We say that $\mappingRunBranch(\run) \in \branchset$ is the representative of $\run$ in $\lcodagAW{\aut}{w}$, where $\mappingRunBranch \colon \runset \to \branchset$ is a total mapping defined as follows.
	\begin{itemize}
	\item
		Let $\run \in \runset$ be of the form $\run = n_{0} n_{1} n_{2} \cdots$ where $n_{0} \in \nstates \cap \initialStates$ and for each $i \in \naturals$, we have that $n_{i + 1} \in \ntrans(n_{i}, \wordletter{w}{i})$. 
		Then we define $\mappingRunBranch(\run) \in \branchset$ as $\mappingRunBranch(\run) = \vertex{N_{0}}{0} \vertex{N_{1}}{1} \vertex{N_{2}}{2} \cdots$, where $N_{0} = \initialStates \cap \nstates$ and $N_{i + 1} = \ntrans(N_{i}, \wordletter{w}{i})$ for each $i \in \naturals$.
	\item 
		Let $\run \in \runset$ be of the form $\run = n_{0} n_{1} \cdots n_{j} d_{j + 1} d_{j + 2} \cdots$ where $n_{0} \in \nstates \cap \initialStates$ and there is $j \in \naturals$ such that for each $0 \leq i < j$, we have that  $n_{i + 1} \in \ntrans(n_{i}, \wordletter{w}{i})$,  $d_{j + 1} \in \jtrans(n_{j}, \wordletter{w}{j})$, and $d_{k + 1} = \dtrans(d_{k}, \wordletter{w}{k})$ for each $k > j$.
		Then we define the branch $\mappingRunBranch(\run) \in \branchset$ step by step as follows.
		For each $0 \leq i \leq j$, we define $\wordletter{\mappingRunBranch(\run)}{i} = \vertex{N_{i}}{i}$, where $N_{0} = \initialStates \cap \nstates$ and $N_{i + 1} = \ntrans(N_{i}, \wordletter{w}{i})$.
		
		Consider now the step $j + 1$ and the vertex $\vertex{d_{j + 1}}{j + 1}$. 
		By Definition~\ref{def:LDBAcodeterministicDAGandPriorityFunction}, the edge from $\vertex{N_{j}}{j}$ to $\vertex{d_{j + 1}}{j + 1}$ may have been removed: 
		this happens if $d_{j + 1}$ is the successor of another deterministic vertex in $\lcodagAW{\aut}{w}$.
		If $(\vertex{N_{j}}{j}, \vertex{d_{j + 1}}{j + 1}) \notin \edges$, then we stop the construction of $\mappingRunBranch(\run)$, that is, we just set $\mappingRunBranch(\run) = \vertex{N_{0}}{0} \vertex{N_{1}}{1} \cdots \vertex{N_{j}}{j}$. 
		Otherwise, if $(\vertex{N_{j}}{j}, \vertex{d_{j + 1}}{j + 1}) \in \edges$, then we just add $\wordletter{\mappingRunBranch(\run)}{j + 1} = \vertex{d_{j + 1}}{j + 1}$ as well as all successors $\wordletter{\mappingRunBranch(\run)}{k} = \vertex{d_{k}}{k}$ for $k > j + 1$, where $d_{k+1} = \dtrans(d_{k}, \wordletter{w}{k})$, as long as $(\vertex{d_{k}}{k}, \vertex{d_{k + 1}}{k + 1}) \in \edges$. 
		This results in the branch 
		\[
			\mappingRunBranch(\run) = \vertex{N_{0}}{0} \vertex{N_{1}}{1} \cdots \vertex{N_{j}}{j} \vertex{d_{j + 1}}{j + 1} \vertex{d_{j + 2}}{j + 2} \cdots \vertex{d_{k}}{k} \cdots
		\]
		that is possibly stopping in $\vertex{d_{k}}{k}$ provided that $(\vertex{d_{k}}{k}, \vertex{d_{k + 1}}{k + 1}) \notin \edges$.
	\item
		Lastly, let $\run \in \runset$ be of the form $\run = d_{0} d_{1} d_{2} \cdots$ where $d_{0} \in \dstates \cap \initialStates$ and for each $i \in \naturals$, we have that $d_{i + 1} \in \dtrans(d_{i}, \wordletter{w}{i})$. 
		Then we define $\mappingRunBranch(\run) \in \branchset$ as 
		\[
			\mappingRunBranch(\run) = \vertex{d_{0}}{0} \vertex{d_{1}}{1} \vertex{d_{2}}{2} \cdots \vertex{d_{k}}{k} \cdots
		\] 
		that is possibly stopping in $\vertex{d_{k}}{k}$ provided that $(\vertex{d_{k}}{k}, \vertex{d_{k + 1}}{k + 1}) \notin \edges$.
	\end{itemize}
\end{definition}

\begin{lemma}
\label{lem:LDBAlanguageSizeComplementNBA:merge}
	For each infinite run $\run$ of the LDBA~$\aut$ over $w \in \infwords$ such that $\mappingRunBranch(\run)$ is finite, it holds that $\run$ eventually merges with a run $\run^{*}$ of $\aut$ over $w$ such that $\mappingRunBranch(\run^{*})$ is an $\omega$-branch.
\end{lemma}
\begin{proof}
	Let $\run \in \runset$ be an infinite run of $\aut$ over $w$ such that $\mappingRunBranch(\run)$ is finite.
	In order for $\mappingRunBranch(\run)$ to be finite, by Definition~\ref{def:LDBAlanguageSizeComplementNBA:mappingRepresentation} we must have that either 
	\begin{itemize}
	\item
		$\run = n_{0} n_{1} \cdots n_{j} d_{j + 1} d_{j + 2} \cdots$ where $n_{0} \in \nstates \cap \initialStates$ and there is $j \in \naturals$ such that for each $0 \leq i < j$, we have that  $n_{i + 1} \in \ntrans(n_{i}, \wordletter{w}{i})$,  $d_{j + 1} \in \jtrans(n_{j}, \wordletter{w}{j})$, and $d_{k + 1} = \dtrans(d_{k}, \wordletter{w}{k})$ for each $k > j$, or
	\item
		$\run = d_{0} d_{1} d_{2} \cdots$ where $d_{0} \in \dstates \cap \initialStates$ and for each $i \in \naturals$, we have that $d_{i + 1} \in \dtrans(d_{i}, \wordletter{w}{i})$.
	\end{itemize}
	Let $\run$ be a run of the first type.
	If $\mappingRunBranch(\run)$ is the finite branch $\vertex{N_{0}}{0} \vertex{N_{1}}{1} \cdots \vertex{N_{j}}{j}$, then by Definition~\ref{def:LDBAlanguageSizeComplementNBA:mappingRepresentation} we have that $(\vertex{N_{j}}{j}, \vertex{d_{j + 1}}{j + 1}) \notin \edges$;
	by Definition~\ref{def:LDBAcodeterministicDAGandPriorityFunction}, this means that the edge from $\vertex{N_{j}}{j}$ to $\vertex{d_{j + 1}}{j + 1}$ has been omitted due to the fact that there is another deterministic vertex $\vertex{p_{j}}{j} \in \dvertices$ such that $p \in P_{j}$, $d_{j + 1} \in \dtrans(p, \wordletter{w}{j})$ and $\priority(\vertex{d_{j + 1}}{j + 1}) > \priority(\vertex{p_{j}}{j})$.
	
	Analogously, if $\mappingRunBranch(\run)$ is the finite branch 
	\[
	    \mappingRunBranch(\run) = \vertex{N_{0}}{0} \vertex{N_{1}}{1} \cdots \vertex{N_{j}}{j} \vertex{d_{j + 1}}{j + 1} \cdots \vertex{d_{k}}{k},
	\]
	then by Definition~\ref{def:LDBAlanguageSizeComplementNBA:mappingRepresentation} we have that $(\vertex{d_{k}}{k}, \vertex{d_{k + 1}}{k + 1}) \notin \edges$;
	as before, by Definition~\ref{def:LDBAcodeterministicDAGandPriorityFunction}, this means that the edge from $\vertex{d_{k}}{k}$ to $\vertex{d_{k + 1}}{k + 1}$ has been omitted due to the fact that there is another deterministic vertex $\vertex{p_{k}}{k} \in \dvertices$ such that $p \in P_{k}$, $d_{k + 1} \in \dtrans(p, \wordletter{w}{k})$, and $\priority(\vertex{d_{k + 1}}{k + 1}) > \priority(\vertex{p_{k}}{k})$.
	These two cases can be analyzed together, with index $k$ playing the role of $j$ for the former case.
	
	Since $p_{k} \neq d_{k}$ (otherwise we would have the edge $(\vertex{d_{k}}{k}, \vertex{d_{k + 1}}{k + 1})$, so $\mappingRunBranch(\run)$ would not have been blocked), by a backward analysis of the construction of $\lcodagAW{\aut}{w}$, we can identify another run $\run'$ such that $\wordletter{\run'}{k} = p_{k}$, $\wordletter{\run'}{k + 1} = d_{k + 1}$, and for each $j > k$, $\wordletter{\run'}{j} = \wordletter{\run}{j}$, that is, $\run'$ and $\run$ merge.
	The fact that $\run$ and $\run'$ agree from position $k + 1$ is an immediate consequence of the fact that $d_{k + 1} \in \dstates$ and that $\wordletter{\run'}{k + 1} = d_{k + 1} = \wordletter{\run}{k + 1}$ by hypothesis.
	Since $\lcodagAW{\aut}{w}$ is codeterministic, $\vertex{p_{k}}{k}$ has exactly one predecessor vertex: 
	either $\vertex{N_{k - 1}}{k - 1}$ with $\emptyset \neq N_{k - 1} \subseteq 2^{\nstates}$, or $\vertex{p_{k - 1}}{k - 1}$ for some $p_{k - 1} \in \dstates$ such that $p_{k} = \dtrans(p_{k - 1}, \wordletter{w}{k - 1})$.
	In the latter case, we require $\wordletter{\run'}{k - 1} = p_{k - 1}$.
	By going backward step by step, either we arrive at $p_{0} \in \dstates \cap \initialStates$ and we have that $\run' = p_{0} p_{1} \cdots p_{k} d_{k + 1} d_{k + 2} \cdots$, or at some step $i$ the former case must occur, that is, the only predecessor of $\vertex{p_{i}}{i}$ is $\vertex{N_{i - 1}}{i - 1}$ with $\emptyset \neq N_{i - 1} \subseteq 2^{\nstates}$.
	By Definition~\ref{def:LDBAcodeterministicDAGandPriorityFunction}, it follows that there is $n_{i - 1} \in N_{i - 1}$ such that $p_{i} \in \jtrans(n_{i - 1}, \wordletter{w}{i - 1})$;
	still going backward, we can choose $n_{i - 2} \in N_{i - 2}$, $n_{i - 3} \in N_{i - 3}$ until $n_{0} \in N_{0} = \nstates \cap \initialStates$ such that $n_{j + 1} \in \ntrans(n_{j}, \wordletter{w}{j})$ for each $0 \leq j < i - 1$.
	The corresponding run $\run'$ is $\run' = n_{0} n_{1} \cdots n_{i - 1} p_{i} p_{i+1} \cdots p_{k} d_{k + 1} d_{k + 2}$.
	Clearly, in both cases, we have that $\run$ and $\run'$ merge at $d_{k + 1}$.
	Moreover, we have that $\mappingRunBranch(\run') \neq \mappingRunBranch(\run)$ since $p_{k} \neq N_{k}$ and $p_{k} \neq d_{k}$ depending on whether the missing edge that blocked $\mappingRunBranch(\run)$ is $(\vertex{N_{k}}{k}, \vertex{d_{k + 1}}{k + 1})$ or $(\vertex{d_{k}}{k}, \vertex{d_{k + 1}}{k + 1})$, respectively.
	
	If $\mappingRunBranch(\run')$ is an $\omega$-branch, the proof is complete.
	Otherwise, for $\mappingRunBranch(\run')$ also being a finite branch, we just repeat the above analysis and get another run $\run''$ that merges with $\run'$; 
	this also means that $\run$ merges with $\run''$.
	Note that the step at which $\run''$ and $\run'$ merge is after the step at which $\run'$ and $\run$ have merged;
	this is a direct consequence of the fact that $\mappingRunBranch(\run)$ was blocked at level $k$, just before reaching $\vertex{d_{k + 1}}{k + 1}$, while $\mappingRunBranch(\run)$ was able to reach $\vertex{d_{k + 1}}{k + 1}$, so the level at which $\mappingRunBranch(\run')$ is blocked must be at least $k + 1$.
	Since the number of merges we can perform is bounded by the priority of the vertex $\vertex{d_{k + 1}}{k + 1}$, we must eventually merge with a run $\run^{*}$ such that $\mappingRunBranch(\run^{*})$ is an $\omega$-branch.
	
	By Definition~\ref{def:LDBAcodeterministicDAGandPriorityFunction}, we have that the priority of deterministic vertices is preserved by the edges between them;
	as we have seen above, the omission of an edge corresponds to a merge of the runs due to the fact that the edge should connect a vertex $v_{s}$ with a vertex $v_{t}$ with $\priority(v_{t}) < \priority(v_{s})$, with the latter due to the presence of a vertex $p_{s} \in \dvertices$ such that $(p_{s}, v_{t}) \in \edges$, thus $\priority(v_{t}) = \priority(p_{s})$.
	Since the priority is a function with codomain $\setnocond{1, 2, \cdots, \size{\dstates} + 1}$ and $\dstates$ is a finite set, we can only merge at most $\size{\dstates} + 1$ many times. 
	After all these merges, no other merge is possible otherwise the corresponding priority would be at most $0$, which is impossible.
	Let $\run^{*}$ be the last run with which $\run$ merges; 
        $\run^{*}$ ensures that $\mappingRunBranch(\run^{*})$ is an $\omega$-branch, as required by the statement of the lemma.
	
	Suppose, for the sake of contradiction, that $\mappingRunBranch(\run^{*})$ is still finite.
	This means that the last vertex of $\mappingRunBranch(\run^{*})$ is a deterministic vertex, because the first merge of $\run$ and $\run'$ happened at $d_{k + 1} \in \dstates$, so $\mappingRunBranch(\run')$ blocked in a deterministic vertex, and each of the following blocks happened at levels higher and higher, which can only be deterministic as well.
	Since the last vertex $\vertex{d_{z}}{z}$ of $\mappingRunBranch(\run^{*})$ is deterministic, the only motivation for this block is the presence of another vertex $\vertex{p_{z}}{z}$ such that $\priority(\vertex{p_{z}}{z}) < \priority(\vertex{d_{z}}{z})$ and $\dtrans(d_{z}, \wordletter{w}{z}) = \dtrans(p_{z}, \wordletter{w}{z})$.
	By the same reasoning as above, we can identify a run $\run_{z}$ of $\aut$ over $w$ such that $\wordletter{\run_{z}}{z} = p_{z}$ and $\wordletter{\run_{z}}{z + 1} = \dtrans(p_{z}, \wordletter{w}{z}) = \dtrans(d_{z}, \wordletter{w}{z}) = \wordletter{\run^{*}}{z + 1}$, contradicting the fact that $\run^{*}$ is the last possible run $\run$ can merge with.
\end{proof}

By Lemma~\ref{lem:LDBAPropertiesOfCodeterministicDAG}, Property~\ref{lem:LDBAPropertiesOfCodeterministicDAG:numberOmegaBranches}, there are at most $\size{\states}$ $\omega$-branches in $\lcodagAW{\aut}{w}$; 
let $\hat{\run}_{1}, \cdots, \hat{\run}_{n}$ be such $\omega$-branches and, for each of them, $j_{m} \in \naturals$ be the level at which $\hat{\run}_{m}$ enters $\dvertices$, that is, the branch $\hat{\run}_{m}$ is as follows:
\[
    \hat{\run}_{m} = \vertex{N_{0}}{0} \cdots \vertex{N_{j_{m} - 1}}{j_{m} - 1} \vertex{d_{j_{m}}}{j_{m}} \vertex{d_{j_{m} + 1}}{j_{m} + 1} \cdots
\] 
where we have $N_{0} = \initialStates \cap \nstates$ and $N_{i + 1} = \ntrans(N_{i}, \wordletter{w}{i})$ for each $i < j_{m} - 1$, as well as $d_{j_{m}} \in \jtrans(N_{j_{m} - 1}, \wordletter{w}{j_{m} - 1})$ and $d_{i + 1} = \dtrans(d_{i}, \wordletter{w}{i})$ for each $i \geq j_{m}$. 
Then $\hat{\run}_{m}$ is the representative of the run $\run = n_{0} n_{1} \cdots n_{j_{m} - 1} d_{j_{m}} d_{j_{m} + 1} \cdots$ where $n_{i} \in N_{i}$ for each $0 \leq i < j_{m}$, $n_{j + 1} \in \ntrans(n_{j}, \wordletter{w}{j})$ for each $0 \leq j < j_{m} - 1$, $d_{j_{m}} \in \jtrans(N_{j_{m} - 1}, \wordletter{w}{j_{m} - 1})$, and $d_{i + 1} = \dtrans(d_{i}, \wordletter{w}{i})$ for each $i \geq j_{m}$. 
Note that $j_{m}$ can also be $0$; 
in this case, it means that $\hat{\run}_{m}$ is just $\hat{\run}_{m} = \vertex{d_{0}}{0} \vertex{d_{1}}{1} \cdots$, thus $d_{0} \in \initialStates \cap \dstates$ since $\hat{\run}_{m}$ is an initial branch.

Recall that $\stableLevel \geq 1$ is the stable level of $\lcodagAW{\aut}{w}$ such that all $\acc$-vertices after level $\stableLevel$ are finite (cf.\@ Lemma~\ref{lem:LDBAPropertiesOfCodeterministicDAG}, Property~\ref{lem:LDBAPropertiesOfCodeterministicDAG:stableLevel}). 

\begin{lemma}
\label{lem:LDBAlanguageSizeComplementNBA:saferun}
	Given an LDBA~$\aut$, an $\omega$-word $w \in \infwords$, and the NBA~$\aut^{c}$ constructed according to Definition~\ref{def:LDBAnsbcComplementation}, let 
	\[
	    \run^{c} = R_{0} R_{1} \cdots R_{\macrorunJump} (N_{\macrorunJump + 1}, S_{\macrorunJump + 1}, B_{\macrorunJump + 1}, C_{\macrorunJump + 1}) (N_{\macrorunJump + 2}, S_{\macrorunJump + 2}, B_{\macrorunJump + 2}, C_{\macrorunJump + 2}) \cdots
	\] 
	be the macrorun of $\aut^{c}$ over $w$ jumping from the initial phase to the accepting phase exactly at step $\macrorunJump = \max \setnocond{\stableLevel + 1, j_{1}, j_{2}, \cdots, j_{n}}$.
	Then, for each run $\run$ of $\aut$ over $w$ such that $\wordletter{\run}{k} \in \dstates$ for some $k \in \naturals$ and $\mappingRunBranch(\run)$ is an $\omega$-branch, it holds that $\wordletter{\run}{i} \in S_{i}$ for each $i > \macrorunJump$.
\end{lemma}
\begin{proof}
	By Definition~\ref{def:LDBAnsbcComplementation}, at the jump from $R_{\macrorunJump}$ to $(N_{\macrorunJump + 1}, S_{\macrorunJump + 1}, B_{\macrorunJump + 1}, C_{\macrorunJump + 1})$, the states in $R_{\macrorunJump}$ are split into the pseudo-macrostate $(N_{\macrorunJump}, S_{\macrorunJump}, B_{\macrorunJump}, C_{\macrorunJump}) = (R_{\macrorunJump} \cap \nstates, (R_{\macrorunJump} \cap \dstates) \setminus{\acc}, R_{\macrorunJump} \cap \acc, \emptyset)$.
	
	Consider now the $\omega$-branches of $\lcodagAW{\aut}{w}$; 
	let the set of vertices at level $\macrorunJump$ in the $\omega$-branches be $\vertices_{\macrorunJump} = \setcond{\wordletter{\branch}{\macrorunJump}}{\text{$\branch$ is an $\omega$-branch of $\lcodagAW{\aut}{w}$}}$. 
	This implies that $\vertices_{\macrorunJump}$ can be partitioned into the sets of deterministic $\vertices_{\macrorunJump} \cap \dvertices = \setnocond{\vertex{s_{1}}{\macrorunJump} \vertex{s_{2}}{\macrorunJump} \cdots \vertex{s_{n}}{\macrorunJump}}$ and nondeterministic $\vertices_{\macrorunJump} \cap \nvertices = \setnocond{\vertex{N_{\macrorunJump}}{\macrorunJump}}$ vertices;
	note that at most one between these two sets can be empty.
	
	Let $\run$ be a run of $\aut$ over $w$ such that $\wordletter{\run}{k} \in \dstates$ for some $k \in \naturals$ and $\mappingRunBranch(\run)$ is an $\omega$-branch.
	This implies that $\wordletter{\run}{\macrorunJump} \in \setnocond{s_{1}, \cdots, s_{n}}$. 
	In fact, by definition of $M$ (cf.\@ Definition~\ref{def:LDBAlanguageSizeComplementNBA:mappingRepresentation}), we have that $\wordletter{\mappingRunBranch(\run)}{i} = \vertex{\wordletter{\run}{\macrorunJump}}{i}$ for all $i \geq k$; 
	since $\mappingRunBranch(\run)$ is an $\omega$-branch by hypothesis, the definition of $\macrorunJump$ implies that $k \leq \macrorunJump$, thus $\wordletter{\mappingRunBranch(\run)}{\macrorunJump} \in \vertices_{\macrorunJump} \cap \dvertices = \setnocond{\vertex{s_{1}}{\macrorunJump} \vertex{s_{2}}{\macrorunJump} \cdots \vertex{s_{n}}{\macrorunJump}}$, hence $\wordletter{\run}{\macrorunJump} \in \setnocond{s_{1}, s_{2}, \cdots, s_{n}}$.
	
	Since $\stableLevel$ is the stable level of $\lcodagAW{\aut}{w}$ and $\macrorunJump > \stableLevel$, it follows that 
	every $\omega$-branch does not visit any $\acc$-vertex after level $\stableLevel$. 
	This implies that $\setnocond{s_{1}, s_{2}, \cdots, s_{n}} \cap \acc = \emptyset$, thus $\setnocond{s_{1}, s_{2}, \cdots, s_{n}} \subseteq S_{\macrorunJump}$. 
	Moreover, by combining the previous two results, we have that $\wordletter{\run}{\macrorunJump} \in S_{\macrorunJump}$.
	In addition, we have that $\wordletter{\run}{i} \in S_{i}$ holds for each $i \geq \macrorunJump$, as it can be easily shown by induction: 
	we have just seen the base case $i = \macrorunJump$;
	for the inductive case, suppose that $\wordletter{\run}{i} \in S_{i}$ holds for $i \geq \macrorunJump$.
	The, by definition of $\run$ and $\trans$, we have that $\wordletter{\run}{i + 1} = \dtrans(\wordletter{\run}{i}, \wordletter{w}{i})$.
	Since $\wordletter{\run}{i + 1} \notin \acc$ is a direct consequence of the fact that $i \geq \macrorunJump > \stableLevel$, $\wordletter{\mappingRunBranch(\run)}{i} = \vertex{\wordletter{\run}{\macrorunJump}}{i}$, and $\mappingRunBranch(\run)$ does not visit any $\acc$-vertex after level $\stableLevel$, it follows that $\wordletter{\run}{i + 1} \in S_{i + 1}$ as required, since $S_{i + 1} = \dtrans(S_{i}, \wordletter{w}{i}) \setminus \acc$.
\end{proof}

\begin{lemma}
\label{lem:LDBAlanguageSizeComplementNBA:breakpointEventuallyEmpty}
	Given an LDBA~$\aut$, an $\omega$-word $w \in \infwords$, and the NBA~$\aut^{c}$ constructed according to Definition~\ref{def:LDBAnsbcComplementation}, let 
	\[
	    \run^{c} = R_{0} R_{1} \cdots R_{\macrorunJump} (N_{\macrorunJump + 1}, S_{\macrorunJump + 1}, B_{\macrorunJump + 1}, C_{\macrorunJump + 1}) (N_{\macrorunJump + 2}, S_{\macrorunJump + 2}, B_{\macrorunJump + 2}, C_{\macrorunJump + 2}) \cdots
	\] 
	be the macrorun of $\aut^{c}$ over $w$ jumping from the initial phase to the accepting phase exactly at step $\macrorunJump = \max \setnocond{\stableLevel + 1, j_{1}, j_{2}, \cdots, j_{n}}$.
	Then, for each run $\run$ of $\aut$ over $w$ such that $\wordletter{\run}{\runEntersB} \in B_{\runEntersB}$ for some $\runEntersB > \macrorunJump$ with $B_{\runEntersB - 1} = \emptyset$, it holds that $\mappingRunBranch(\run)$ is finite.
\end{lemma}
\begin{proof}
	Consider the usual pseudo-macrostate $(N_{\macrorunJump}, S_{\macrorunJump}, B_{\macrorunJump}, C_{\macrorunJump}) = (R_{\macrorunJump} \cap \nstates, (R_{\macrorunJump} \cap \dstates) \setminus{\acc}, R_{\macrorunJump} \cap \acc, \emptyset)$ used to split the states in $R_{\macrorunJump}$ for the jump at step $\macrorunJump$ from the initial to the accepting phase. 
	
	Let $\run$ be a run of $\aut$ over $w$ such that $\wordletter{\run}{\runEntersB} \in B_{\runEntersB}$ and $B_{\runEntersB - 1} = \emptyset$ for some $\runEntersB > \macrorunJump$.
	We can distinguish two cases, depending on whether $\run$ enters $B$ just after the jump in $\run^{c}$.
	\begin{description}
	\item[Case $\runEntersB = \macrorunJump + 1$.]
		In this case, the set $B_{\runEntersB - 1} = B_{\macrorunJump}$ is taken from the pseudo-macrostate $(N_{\macrorunJump}, S_{\macrorunJump}, B_{\macrorunJump}, C_{\macrorunJump})$ that is used in the definition of $\jtrans^{c}$.
	    Since $B_{\macrorunJump} = \emptyset$ by assumption, by Definition~\ref{def:LDBAnsbcComplementation} we have that 
	    $B_{\macrorunJump + 1} = (\dtrans(C_{\macrorunJump}, \wordletter{w}{\macrorunJump}) \cup
	    \jtrans(N_{\macrorunJump}, \wordletter{w}{\macrorunJump}) \cup
	    (\dtrans(S_{\macrorunJump}, \wordletter{w}{\macrorunJump}) \cap \acc)) \setminus{S_{\macrorunJump + 1}}$.
	    Given that by hypothesis we have that $\wordletter{\run}{\macrorunJump + 1} \in B_{\macrorunJump + 1}$, this implies that $\wordletter{\run}{\macrorunJump + 1} \notin S_{\macrorunJump + 1} = \dtrans(S_{\macrorunJump}, \wordletter{w}{\macrorunJump}) \setminus \acc$, so $\wordletter{\run}{\macrorunJump + 1}$ must belong to at least one of the sets $\dtrans(C_{\macrorunJump}, \wordletter{w}{\macrorunJump})$, $\jtrans(N_{\macrorunJump}, \wordletter{w}{\macrorunJump})$, and $\dtrans(S_{\macrorunJump}, \wordletter{w}{\macrorunJump}) \cap \acc$.
	    Since $C_{\macrorunJump} = \emptyset$, we get that $\dtrans(C_{\macrorunJump}, \wordletter{w}{\macrorunJump}) = \emptyset$, so we must have $\wordletter{\run}{\macrorunJump + 1}  \in \jtrans(N_{\macrorunJump}, \wordletter{w}{\macrorunJump})$ or $\wordletter{\run}{\macrorunJump + 1}  \in \dtrans(S_{\macrorunJump}, \wordletter{w}{\macrorunJump}) \cap \acc$.
	    We analyze the two cases independently.
	    \begin{description}
		\item[Case $\wordletter{\run}{\macrorunJump + 1} \in \jtrans(N_{\macrorunJump}, \wordletter{w}{\macrorunJump})$.] 
			Assume, for the sake of contradiction, that $\mappingRunBranch(\run)$ is an $\omega$-branch. 
			By definition of $\mappingRunBranch$ (cf.\@ Definition~\ref{def:LDBAlanguageSizeComplementNBA:mappingRepresentation}), it follows that $\mappingRunBranch(\run)$ is of the form 
			\[
				\mappingRunBranch(\run) = \vertex{N_{0}}{0} \vertex{N_{1}}{1} \cdots \vertex{N_{\macrorunJump}}{\macrorunJump} \vertex{d_{\macrorunJump + 1}}{\macrorunJump + 1} \vertex{d_{\macrorunJump + 2}}{\macrorunJump + 2} \cdots,
			\] 
			thus $j_{\mappingRunBranch(\run)} = \macrorunJump + 1$.
			This, however, contradicts the definition of $\macrorunJump$ itself, namely, that $\macrorunJump = \max \setnocond{\stableLevel + 1, j_{1}, j_{2}, \cdots, j_{n}}$ where $j_{\mappingRunBranch(\run)} \in \setnocond{j_{1}, j_{2}, \cdots, j_{n}}$ since $\setnocond{j_{1}, j_{2}, \cdots, j_{n}}$ is the set of moments $j_{m}$ at which each $\omega$-branch $\branch_{m}$ enters $\dvertices$.
			Therefore, $\mappingRunBranch(\run)$ must be finite branch.
   
		\item[Case $\wordletter{\run}{\macrorunJump + 1} \in \dtrans(S_{\macrorunJump}, \wordletter{w}{\macrorunJump}) \cap \acc$.]
	        This means that $\wordletter{\run}{\runEntersB}$ is an accepting state. 
	        Suppose, for the sake of contradiction, that $\mappingRunBranch(\run)$ is an $\omega$-branch; 
	        by definition of $\mappingRunBranch$ we have that $\wordletter{\mappingRunBranch(\run)}{\macrorunJump + 1} = \vertex{\wordletter{\run}{\macrorunJump + 1}}{\macrorunJump + 1}$, that is, $\wordletter{\mappingRunBranch(\run)}{\macrorunJump + 1}$ is an $\acc$-vertex. 
	        This however contradicts the assumption that $\macrorunJump > \stableLevel$, where $\stableLevel$ is the stable level of $\lcodagAW{\aut}{w}$ such that all $\acc$-vertices after level $\stableLevel$ (hence, also $\wordletter{\mappingRunBranch(\run)}{\macrorunJump + 1}$) are finite. 
	        Thus $\mappingRunBranch(\run)$ can only be finite branch.
	    \end{description}
	    In both cases, we have that $\mappingRunBranch(\run)$ is a finite branch, as required.
     
	\item [Case $\runEntersB > \macrorunJump + 1$.]
		Consider the state $\wordletter{\run}{\macrorunJump}$;
		we have several cases depending on how $\wordletter{\run}{\macrorunJump}$ is classified in $(N_{\macrorunJump}, S_{\macrorunJump}, B_{\macrorunJump}, C_{\macrorunJump})$.
		For sure, $\wordletter{\run}{\macrorunJump} \notin C_{\macrorunJump}$, since $C_{\macrorunJump} = \emptyset$ by Definition~\ref{def:LDBAnsbcComplementation}.
		It is possible that $\wordletter{\run}{\macrorunJump} \in S_{\macrorunJump}$;
		in this case, it can only be that $\mappingRunBranch(\run)$ is a finite branch.
		Assume, for the sake of contradiction, that $\wordletter{\run}{\macrorunJump} \in S_{\macrorunJump}$ and that $\mappingRunBranch(\run)$ is an $\omega$-branch.
		Since $\macrorunJump > \stableLevel$, we have for each $i \geq \macrorunJump$ that $\wordletter{\mappingRunBranch(\run)}{i}$ is not an $\acc$-vertex thus, by definition of $\mappingRunBranch$, $\wordletter{\run}{i} \notin \acc$.
		By Lemma~\ref{lem:LDBAlanguageSizeComplementNBA:saferun}, we also have that $\wordletter{\run}{i} \in S_{i}$ for each $i > \macrorunJump$.
		By Definition~\ref{def:LDBAnsbcComplementation}, this implies that $\wordletter{\run}{i} \notin B_{i}$ for each $i > \macrorunJump$; 
		in particular we have that $\wordletter{\run}{\runEntersB} \notin B_{\runEntersB}$ since $\runEntersB > \macrorunJump + 1$, against the hypothesis that $\wordletter{\run}{\runEntersB} \in B_{\runEntersB}$. 
		Therefore, it must be the case that if $\wordletter{\run}{\macrorunJump} \in S_{\macrorunJump}$, then $\mappingRunBranch(\run)$ is finite.
		
		Suppose that $\wordletter{\run}{\macrorunJump} \in B_{\macrorunJump}$;
		then we have that either $\wordletter{\run}{i} \in B_{i}$ for each $\macrorunJump < i < \runEntersB - 1$ or $\mappingRunBranch(\run)$ is a finite branch.
		Suppose that it is not the case that $\wordletter{\run}{i} \in B_{i}$ for each $\macrorunJump < i < \runEntersB - 1$: 
		this means that there is $\macrorunJump < i < \runEntersB - 1$ such that $\wordletter{\run}{i - 1} \in B_{i - 1}$ but $\wordletter{\run}{i} \notin B_{i}$.
		Since $\wordletter{\run}{i} = \dtrans(\wordletter{\run}{i - 1}, \wordletter{w}{i - 1})$, by Definition~\ref{def:LDBAnsbcComplementation} it follows that $\wordletter{\run}{i} \notin B_{i}$ can only be justified by the fact that $\wordletter{\run}{i} \in S_{i}$.
		Since $\macrorunJump < i < \runEntersB - 1$, by the same argument by contradiction we used before for the case $\wordletter{\run}{\macrorunJump} \in S_{\macrorunJump}$ we have that $\mappingRunBranch(\run)$ is finite, as required.
		
		Assume that $\wordletter{\run}{\macrorunJump} \in B_{\macrorunJump}$ and that $\wordletter{\run}{i} \in B_{i}$ for each $\macrorunJump < i < \runEntersB - 1$.
		By hypothesis, we have that $\wordletter{\run}{\runEntersB - 1} \notin B_{\runEntersB - 1}$; 
		since $\wordletter{\run}{\runEntersB - 2} \in B_{\runEntersB - 2}$, as before this implies that $\wordletter{\run}{\runEntersB - 1} \in S_{\runEntersB - 1}$.
		By the same argument by contradiction we used before for the case $\wordletter{\run}{\macrorunJump} \in S_{\macrorunJump}$ we have that $\mappingRunBranch(\run)$ is finite, as required.

		The last case is $\wordletter{\run}{\macrorunJump} \in N_{\macrorunJump}$.
		Since $\wordletter{\run}{\runEntersB} \in B_{\runEntersB}$ and $B_{\runEntersB - 1} = \emptyset$, there must be $ \macrorunJump < \runJumped < \runEntersB$ such that $\wordletter{\run}{\runJumped - 1} \in N_{\runJumped - 1}$ while $\wordletter{\run}{\runJumped} \notin N_{\runJumped}$, thus $\wordletter{\run}{\runJumped} \in \dstates$;
		in particular, we have that $\wordletter{\run}{\runJumped} \in \jtrans(\wordletter{\run}{\runJumped - 1}, \wordletter{w}{\runJumped - 1})$.
		Moreover, by the fact that $\aut$ is an LDBA, for each $i \geq \runJumped$ we have that $\wordletter{\run}{i} \notin N_{i}$. 
		For the same motivation as before, if $\wordletter{\run}{\runJumped} \in S_{\runJumped}$, or in general $\wordletter{\run}{i} \in S_{i}$ for some $\runJumped \leq i < \runEntersB - 1$, then $\mappingRunBranch(\run)$ is finite.
		
		Assume that $\wordletter{\run}{\runJumped} \notin S_{\runJumped}$;
		since $\wordletter{\run}{\runJumped} \notin N_{\runJumped}$, then either $\wordletter{\run}{\runJumped} \in B_{\runJumped}$ or $\wordletter{\run}{\runJumped} \in C_{\runJumped}$.
		If $\wordletter{\run}{\runJumped} \in B_{\runJumped}$, then there is $\runJumped < k < \macrorunJump - 1$ such that $\wordletter{\run}{k} \in B_{k}$ and $\wordletter{\run}{k + 1} \notin B_{k + 1}$; 
		such a $k$ must exist since by hypothesis $B_{\macrorunJump - 1} = \emptyset$.
		As in all other cases, $\wordletter{\run}{k + 1} \notin B_{k + 1}$ can only be due to  $\wordletter{\run}{k + 1} \in S_{k + 1}$, so $\mappingRunBranch(\run)$ is finite.

		If $\wordletter{\run}{\runJumped} \in C_{\runJumped}$, then, for the sake of contradiction, suppose that $\mappingRunBranch(\run)$ is an $\omega$-branch.
		Since $\mappingRunBranch(\run)$ is an $\omega$-branch, it follows that $j_{\mappingRunBranch(\run)} = \runJumped > \macrorunJump$; 
		this, however, contradicts the fact that $\macrorunJump = \max \setnocond{\stableLevel + 1, j_{1}, j_{2}, \cdots, j_{n}}$ where $j_{\mappingRunBranch(\run)} \in \setnocond{j_{1}, j_{2}, \cdots, j_{n}}$ since $\setnocond{j_{1}, j_{2}, \cdots, j_{n}}$ is the set of moments $j_{m}$ at which each $\omega$-branch $\branch_{m}$ enters $\dvertices$.
		Therefore, $\mappingRunBranch(\run)$ must be finite branch.
	\end{description}
	This completes the proof of the lemma, since we have shown that $\mappingRunBranch(\run)$ is a finite branch in all possible cases.
\end{proof}

\begin{lemma}
\label{lem:LDBAlanguageSizeComplementNBA:breakpointFirstBNotEmpty}
	Given an LDBA~$\aut$, an $\omega$-word $w \in \infwords$, and the NBA~$\aut^{c}$ constructed according to Definition~\ref{def:LDBAnsbcComplementation}, let 
	\[
	    \run^{c} = R_{0} R_{1} \cdots R_{\macrorunJump} (N_{\macrorunJump + 1}, S_{\macrorunJump + 1}, B_{\macrorunJump + 1}, C_{\macrorunJump + 1}) (N_{\macrorunJump + 2}, S_{\macrorunJump + 2}, B_{\macrorunJump + 2}, C_{\macrorunJump + 2}) \cdots
	\] 
	be the macrorun of $\aut^{c}$ over $w$ jumping from the initial phase to the accepting phase exactly at step $\macrorunJump = \max \setnocond{\stableLevel + 1, j_{1}, j_{2}, \cdots, j_{n}}$.
	Then, for each run $\run$ of $\aut$ over $w$ such that $\wordletter{\run}{\macrorunJump + 1} \in B_{\macrorunJump + 1}$, it holds that $\mappingRunBranch(\run)$ is finite.
\end{lemma}
\begin{proof}
    There are two cases, depending on whether $B_{\macrorunJump} = R_{\macrorunJump} \cap \acc$ is empty.
    If $B_{\macrorunJump} = \emptyset$, then the result follows immediately by Lemma~\ref{lem:LDBAlanguageSizeComplementNBA:breakpointEventuallyEmpty}, so suppose that $B_{\macrorunJump} \neq \emptyset$.
    
    The hypothesis $B_{\macrorunJump} \neq \emptyset$, by Definition~\ref{def:LDBAnsbcComplementation}, implies that $B_{\macrorunJump + 1} = \dtrans(B_{\macrorunJump}, \wordletter{w}{\macrorunJump}) \setminus S_{\macrorunJump + 1}$, thus $\wordletter{\run}{\macrorunJump} \in \acc$.
    Similarly to the proof of Lemma~\ref{lem:LDBAlanguageSizeComplementNBA:breakpointEventuallyEmpty}, suppose, for the sake of contradiction, that $\mappingRunBranch(\run)$ is an $\omega$-branch; 
	by definition of $\mappingRunBranch$ we have that $\wordletter{\mappingRunBranch(\run)}{\macrorunJump} = \vertex{\wordletter{\run}{\macrorunJump}}{\macrorunJump}$, that is, $\wordletter{\mappingRunBranch(\run)}{\macrorunJump}$ is an $\acc$-vertex. 
	This however contradicts the assumption that $\macrorunJump > \stableLevel$, where $\stableLevel$ is the stable level of $\lcodagAW{\aut}{w}$ such that all $\acc$-vertices after level $\stableLevel$ (hence, also $\wordletter{\mappingRunBranch(\run)}{\macrorunJump}$) are finite. 
	Thus $\mappingRunBranch(\run)$ can only be finite branch.
\end{proof}

We have now all the results needed to prove that for each $w \notin \lang{\aut}$ it holds that $w \in \lang{\aut^{c}}$.
Recall that we are considering the macrorun
\[
	\run^{c} = R_{0} R_{1} \cdots R_{\macrorunJump} (N_{\macrorunJump + 1}, S_{\macrorunJump + 1}, B_{\macrorunJump + 1}, C_{\macrorunJump + 1}) (N_{\macrorunJump + 2}, S_{\macrorunJump + 2}, B_{\macrorunJump + 2}, C_{\macrorunJump + 2}) \cdots
\]
where $R_{0} = \initialStates$ and $\macrorunJump = \max \setnocond{\stableLevel + 1, j_{1}, j_{2}, \cdots, j_{n}}$ as our candidate accepting run over $w$. 
By means of the lemmas presented above, we can prove that, after the step $\macrorunJump$, the $B$ component of $\run^{c}$ becomes empty again and again, thus $\run^{c}$ is accepting so $w \in \lang{\aut^{c}}$ as required.

If $B_{i} = \emptyset$ for each $i > \macrorunJump$, then the claim that the $B$ component of $\run^{c}$ becomes empty again and again trivially holds. 
Suppose now that $B_{i} \neq \emptyset$ for some $i > \macrorunJump$; 
let $\runEntersB = \min\setcond{i > \macrorunJump}{B_{i} \neq \emptyset}$: 
this implies that $B_{\runEntersB}$ is the first occurrence of a non-empty set for the component $B$.
By inspecting the definition of $\trans^{c}$ in Definition~\ref{def:LDBAnsbcComplementation}, we have that for each $q \in B_{\runEntersB}$ there is a run $\run_{q}$ of $\aut$ over $w$ such that $\wordletter{\run_{q}}{\runEntersB} = q$.
If $\runEntersB = \macrorunJump + 1$, then by Lemma~\ref{lem:LDBAlanguageSizeComplementNBA:breakpointFirstBNotEmpty} we know that $\mappingRunBranch(\run_{q})$ is finite.
If $\runEntersB > \macrorunJump + 1$, then we derive that $B_{\runEntersB - 1} = \emptyset$, thus  Lemma~\ref{lem:LDBAlanguageSizeComplementNBA:breakpointEventuallyEmpty} implies that $\mappingRunBranch(\run_{q})$ is finite.

Since in all cases $\mappingRunBranch(\run_{q})$ is a finite branch, Lemma~\ref{lem:LDBAlanguageSizeComplementNBA:merge} ensures that $\run_{q}$ eventually merges with a run $\run^{*}_{q}$ such that $\mappingRunBranch(\run^{*}_{q})$ is an $\omega$-branch.
Lemma~\ref{lem:LDBAlanguageSizeComplementNBA:saferun} ensures that $\run^{*}_{q}$ is such that $\wordletter{\run^{*}_{q}}{i} \in S_{i}$ for each $i \geq \macrorunJump$.
Since $\run_{q}$ and $\run^{*}_{q}$ merge, there exists $m_{q} \geq \macrorunJump$ such that $\wordletter{\run_{q}}{m_{q}} = \wordletter{\run^{*}_{q}}{m_{q}}$, thus $\wordletter{\run_{q}}{i} \in S_{i}$ for each $i \geq m_{q}$. 
By Definition~\ref{def:LDBAnsbcComplementation}, this implies that $\wordletter{\run_{q}}{i} \notin B_{i}$ for each $i \geq m_{q}$.
Since this happens for all $q \in B_{\runEntersB}$, we get that $B_{m} = \emptyset$, where $m = \max\setcond{m_{q}}{q \in B_{\runEntersB}}$;
this is the case because the component $B$ never tracks new runs until all runs $\run_{q}$ have left $B$ to reach the component $S$.
Given the arbitrary choice of $\runEntersB > \macrorunJump$, this means that the $B$ component of $\run^{c}$ becomes empty again and again, that is, $\run^{c}$ is accepting so $w \in \lang{\aut^{c}}$ as required.

\subsection{Proof of the size of \texorpdfstring{$\aut^{c}$}{the complement of A}}
\label{app:LDBAlanguageSizeComplementNBA:size}

We need to prove that the NBA~$\aut^{c}$ constructed according to Definition~\ref{def:LDBAnsbcComplementation} has at most $2^{\size{\states}} + 2^{\size{\nstates}} \times 3^{\size{\acc}} \times 4^{\size{\dstates \setminus \acc}} \in \bigO(4^{\size{\states}})$ states. 

By the construction of $\aut^{c}$ given in Definition~\ref{def:LDBAnsbcComplementation}, in the worst case we have $2^{\states}$ macrostates for the initial phase;
for the accepting phase, we have at most $2^{\size{\nstates}} \times 3^{\size{\acc}} \times 4^{\size{\dstates \setminus {\acc}}}$ macrostates of the form $(N, S, B, C)$.
This is due to the following motivations:
\begin{itemize}
\item 
	for each $n \in \nstates$, we have that either $n \in N$ or $n \notin N$, for a total of $2^{\size{\nstates}}$ possible choices for the states in $\nstates$; 
\item
	for each $f \in \acc$, we have that $f \notin S$.
	Regarding $B$ and $C$, we have that either $f \in B$, or $f \in C$, or $f \notin B \cup C$. 
	Note that by construction we have $B \cap C = \emptyset$, as well as $S \cap C = S \cap B = \emptyset$.
	This means that we have a total of $3^{\size{\acc}}$ possible choices for the states in $\acc$;
\item
	for each $d \in \dstates \setminus \acc$, we have that either $d \in S$, or $d \in C$, or $d \in B$, or $d \notin S \cup B \cup C$, for a total of $4^{\size{\dstates \setminus \acc}}$ possible choices for the states in $\dstates \setminus \acc$.
\end{itemize}
This means that the overall number of macrostates is $2^{\size{\states}} + 2^{\size{\nstates}} \times 3^{\size{\acc}} \times 4^{\size{\dstates \setminus {\acc}}} \leq 2^{\size{\states}} + 4^{\size{\nstates}} \times 4^{\size{\acc}} \times 4^{\size{\dstates \setminus {\acc}}} = 2^{\size{\states}} + 4^{\size{\nstates} + \size{\acc} + \size{\dstates \setminus {\acc}}} = 2^{\size{\states}} + 4^{\size{\states}} \in \bigO(4^{\size{\states}})$, as required.

\end{document}